%% file: main.tex
\documentclass[letterpaper,twocolumn,10pt]{article}
\usepackage{usenix}

\PassOptionsToPackage{pdfpagelabels=false}{hyperref}
\newcommand{\allnotes}[1]{}

\usepackage{hyperref}
\hypersetup{colorlinks=true,citecolor=blue,filecolor=black,linkcolor=blue,urlcolor=black}
\usepackage{amsmath}
\usepackage{amsthm}
\usepackage{endnotes}
\usepackage{enumitem}
\usepackage[T1]{fontenc}
\usepackage{tabularx}
\usepackage{booktabs}
\usepackage[small,compact]{titlesec}
\usepackage{alltt}
\usepackage[skip=5pt]{caption}
\usepackage[aboveskip=2pt]{subcaption}
\usepackage{balance}
\usepackage{fancyvrb}
\usepackage{epsfig}
\usepackage{textcomp}
\usepackage{listings}
\usepackage[framemethod=tikz]{mdframed}
\mdfsetup{skipabove=3pt,skipbelow=0pt}
\usepackage[scaled]{beramono}
\usepackage[T1]{fontenc}
\usepackage{multirow}
\usepackage{mathtools}

\usepackage[ruled, vlined, noend]{algorithm2e}

\mdfdefinestyle{protocolBox}{,,roundcorner=5pt,tikzsetting={draw=gray, line width=1pt}}%

\newcommand{\captionfonts}{\bf \footnotesize}

\makeatletter  
\long\def\@makecaption#1#2{%
	\vskip\abovecaptionskip
	\sbox\@tempboxa{{\captionfonts #1: #2}}%
	\ifdim \wd\@tempboxa >\hsize
	{\captionfonts #1: #2\par}
	\else
	\hbox to\hsize{\hfil\box\@tempboxa\hfil}%
	\fi
	\vskip\belowcaptionskip}
\makeatother   

\newcommand{\squishlist}{
  \begin{list}{$\bullet$}{
    \setlength{\itemsep}{0pt}       \setlength{\parsep}{3pt}
    \setlength{\topsep}{3pt}        \setlength{\partopsep}{0pt}
    \setlength{\leftmargin}{1em}    \setlength{\labelwidth}{1em}
    \setlength{\labelsep}{0.5em} } }

\newcommand{\squishend}{
  \end{list} }

\newcommand{\squishenum}{
  \begin{enumerate}
    \renewcommand{\labelenumi}{(\roman{enumi})}
    \setlength{\itemsep}{0pt}       \setlength{\parsep}{3pt}
    \setlength{\topsep}{3pt}        \setlength{\partopsep}{0pt}
    \setlength{\leftmargin}{1em}    \setlength{\labelwidth}{1em}
    \setlength{\labelsep}{0.5em} }

\newcommand{\squishenumend}{
\end{enumerate} }

\setlist[enumerate]{leftmargin=1.5em,topsep=3pt,itemsep=0pt}
\microtypecontext{spacing=nonfrench}

\usepackage{ulem}
\usepackage{xspace}
\usepackage{xcolor}
\usepackage{outlines}
\usepackage{tikz}
\usetikzlibrary{snakes,arrows,shapes,calc,positioning}

\usepackage{pdflscape}
\usepackage{array}
\usepackage{thmtools,thm-restate}
\usepackage{rotating}
\usepackage{amsmath}
\usepackage{amssymb}
\usepackage{comment}
\usepackage{graphicx}
\usepackage{soul}

\excludecomment{Exclude}

\definecolor{light-gray}{gray}{0.94}
\sethlcolor{light-gray}

\usepackage[textsize=tiny,textwidth=0.6in]{todonotes}
\newcolumntype{x}[1]{>{\centering\arraybackslash\hspace{0pt}}p{#1}}

\SetKwFunction{emap}{routing}

\newcommand{\abs}[1]{\left|#1\right|}

\DeclareMathOperator{\R}{\mathbb{R}}
\DeclareMathOperator{\Q}{\mathbb{Q}}
\DeclareMathOperator{\Z}{\mathbb{Z}}
\DeclareMathOperator{\N}{\mathbb{N}}
\DeclareMathOperator{\I}{\mathbb{I}}
\DeclareMathOperator{\cF}{\mathcal{F}}
\DeclareMathOperator{\cV}{\mathcal{V}}

\DeclareMathOperator{\cO}{\mathcal{O}}
\DeclareMathOperator{\cE}{\mathcal{E}}

\DeclareMathOperator*{\argmax}{arg\,max}

\begin{document}

\title{
ForestColl: Throughput-Optimal Collective Communications on \\ Heterogeneous Network Fabrics
}

\author{
{\rm Liangyu Zhao}\thanks{The work was partially done during an internship at Microsoft Research.}\\
University of Washington
\and
{\rm Saeed Maleki}\thanks{The work was done at Microsoft Research.}\\
Independent Researcher
\and
{\rm Yuanhong Wang}\thanks{The work was done at the University of Washington.}\\
Tsinghua University
\and
{\rm Zezhou Wang}\\
University of Washington
\and
{\rm Ziyue Yang}\\
Microsoft Research
\and
{\rm Hossein Pourreza}\\
Microsoft
\and
{\rm Arvind Krishnamurthy}\\
University of Washington
}



\maketitle
\input{abstract}
\pagestyle{plain}

\input{intro}
\input{background}
\input{problem}
\input{algorithm}
\input{discussion}
\input{evaluation}

\input{conclusion}

\newpage
\bibliographystyle{acm}
\bibliography{refs}

\newpage
\appendix
\input{appendix/appendix_main}

\end{document}

%% file: abstract.tex
\begin{abstract}
As modern DNN models grow ever larger, collective communications between the accelerators (allreduce, etc.) emerge as a significant performance bottleneck. Designing efficient communication schedules is challenging, given today's heterogeneous and diverse network fabrics. We present ForestColl, a tool that generates throughput-optimal schedules for any network topology. ForestColl constructs broadcast/aggregation spanning trees as the communication schedule, achieving theoretical optimality. Its schedule generation runs in polynomial time and is highly scalable. ForestColl supports any network fabric, including both switching fabrics and direct accelerator connections. We evaluated ForestColl on AMD MI250 and NVIDIA DGX A100 \& H100 clusters. ForestColl showed significant improvements over the vendors' own optimized communication libraries across various settings and in LLM training. ForestColl also outperformed other state-of-the-art schedule generation techniques with both more efficient generated schedules and substantially faster generation speed.
\end{abstract}

%% file: intro.tex
\section{Introduction}\label{sec:intro}

Collective communications have become a cornerstone of distributed machine learning training~\cite{ringallreduce,horovod,megatronlm}.
As large language models (LLMs) scale to hundreds of billions of parameters~\cite{llama3,nemotron,dbrx,mixtral8x22b}, their training demands immense volumes of collective communication traffic, creating a performance bottleneck~\cite{wang2023build,alpa,pope2023efficiently,megatronlm,topoopt,megascale,metardma,alibabahpn}. Operational insights from AI hyperscalers (e.g., Meta~\cite{metardma}, Alibaba~\cite{alibabahpn}, AWS~\cite{trainium}) highlight that LLM training traffic, characterized by \emph{bursty elephant flows} capable of saturating NIC line rate, is predominantly \textbf{throughput-bound}.
As a result, today's ML hardware providers have focused on enhancing inter-accelerator network speed~\cite{dgxv100,dgxa100,dgxh100,cdna2,cdna3}. 

A key observation is that \textit{today's ML network topologies are becoming \textbf{heterogeneous} within individual networks and highly \textbf{diverse} across different hardware platforms.}
Because scaling high-speed networks homogeneously (i.e., uniformly across the fabric) is both technically challenging~\cite{rdmaethernet,azurerdma,deadlocksdatacenter} and prohibitively costly~\cite{topoopt,wang2024rail}, hardware providers adopt \emph{heterogeneous} networks, which typically consist of separate high-speed scale-up networks within multi-GPU boxes and lower-speed scale-out networks between boxes. Figure~\ref{fig:topologiesexample} shows the network topologies of NVIDIA DGX A100~\cite{dgxa100} and AMD MI250~\cite{cdna2}. In both topologies, the intra-box network is an order of magnitude faster than the inter-box one---300GB/s vs 25GB/s per GPU in DGX A100 and 350GB/s vs 16GB/s per GPU in MI250. Moreover, different hardware platforms feature highly \emph{diverse} network designs. DGX A100 uses NVSwitch for inter-GPU traffic within a box, while MI250 relies on direct connections between GPUs. Although both platforms use InfiniBand for inter-box traffic, the IB switches can also be configured in various topologies, such as fat-tree~\cite{fattree} or rail networks~\cite{wang2024rail,nvidiarail}.

\begin{figure}[tb]
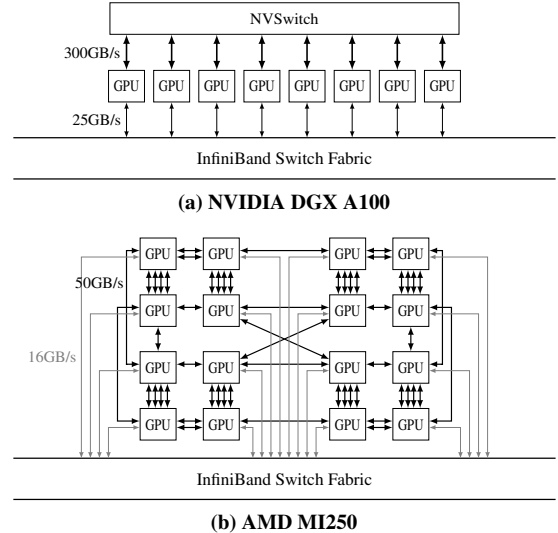

    \centering
    \begin{subfigure}{\columnwidth}
        \centering
        \include{figures/a100_topo_single}
        \caption{NVIDIA DGX A100}
        \label{fig:topoa100example}
        \vspace{9pt}
    \end{subfigure}
    \begin{subfigure}{\columnwidth}
        \centering
        \include{figures/mi250_topo_single}
        \caption{AMD MI250}
        \label{fig:topomi250example}
    \end{subfigure}
    \caption{Network Topologies of NVIDIA DGX A100 and AMD MI250. \textnormal{The PCIe switches and IB NICs are omitted for simplicity.}}
    \label{fig:topologiesexample}
\end{figure}

Given the heterogeneity and diversity, traditional \emph{static} collective algorithms (e.g., ring, recursive halving/doubling), assuming a simple homogeneous network, are ill-suited for today's ML networks. The mismatch between the assumed homogeneity and the actual heterogeneity leads to network imbalance, congestion, and, ultimately, poor throughput. For instance, ring allreduce~\cite{ringallreduce} assumes a flat network where each node sends data to the next node at equal bandwidth. When applied to multi-box DGX A100, however, ring allreduce is bottlenecked by the slower inter-box bandwidth and underutilizes the faster intra-box bandwidth (see example in Fig~\ref{fig:ringexample}). Further, existing static algorithms (e.g., ring, recursive halving/doubling, Bruck algorithm, BlueConnect)~\cite{halving-doubling,blueconnect} assume that communicating with a single peer can saturate a node's bandwidth. Yet, today's ML networks often feature multi-ported nodes with connections to multiple GPUs/switches, which these static algorithms cannot fully exploit.

To address the heterogeneity and diversity of ML network topologies, recent works (e.g., SCCL~\cite{sccl}, TACCL~\cite{taccl}, TE-CCL~\cite{teccl}, TACOS~\cite{tacos}, BFB~\cite{bfb}, Blink~\cite{blink}, MultiTree~\cite{multitree}, TTO~\cite{TTO}, SyCCL~\cite{syccl}) seek to \emph{dynamically generate a collective communication schedule tailored to a given topology.} However, these schedule generation methods still face limitations. They either rely on NP-hard optimizations (SCCL, TACCL, TE-CCL, SyCCL), use suboptimal greedy algorithms (MultiTree, TACOS), support limited collective operations (Blink), or work only with specific topologies (TTO, BFB). Moreover, network switches pose challenges for all methods. Unlike direct GPU connections, switches support flexible traffic patterns that require specialized modeling in schedule generation. Existing methods either ignore switches (SCCL, TTO, BFB) or rely on suboptimal solutions (Blink, MultiTree, TACCL, TE-CCL, TACOS, SyCCL).

\begin{figure}[tb]
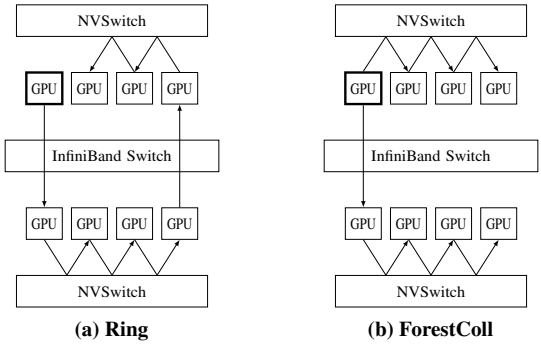

    \centering
    \begin{subfigure}{0.49\columnwidth}
        \centering
        \include{figures/a100_comparison}
        \caption{Ring}
    \end{subfigure}
    \begin{subfigure}{0.49\columnwidth}
        \centering
        \include{figures/a100_comparison2}
        \caption{ForestColl}
        \label{fig:ringexampleoptimized}
    \end{subfigure}
    \caption{Example of ring's suboptimality. \textnormal{(a) and (b) show two broadcast paths from one of the GPUs in ring allgather and ForestColl, respectively. Note that in (a), ring’s path crosses IB switch twice, whereas the path in (b) crosses only once. In allgather, each GPU broadcasts a distinct shard of data to all other GPUs. When all GPUs broadcast simultaneously, ring allgather generates nearly twice the traffic across IB compared to (b), making it suboptimal due to IB’s much lower bandwidth compared to NVSwitch.}}
    \label{fig:ringexample}
\end{figure}

We argue that an ideal schedule generation method should achieve a \textbf{triathlon}: \emph{\ul{optimality}} to produce throughput-efficient schedules, \emph{\ul{generality}} to support heterogeneous and diverse topologies, and \emph{\ul{scalability}} to handle large topologies. In this work, we present ForestColl, a triathlete method capable of generating collective communication schedules with theoretically \textbf{optimal throughput} for any \textbf{heterogeneous topology} in \textbf{polynomial time}. Our approach leverages \textit{spanning tree packing} from graph theory for schedule generation. However, applying spanning tree packing directly is hindered by several key technical challenges, including deriving the throughput optimality, limiting the number of trees, and supporting traffic patterns of switches. ForestColl overcomes these challenges through algorithmic innovations, introducing the following key novelties:
\squishenum
\item \textbf{Optimality:} To the best of our knowledge, ForestColl is the first work capable of deriving the optimal throughput of any topology (\S\ref{sec:throughputoptimality} \& \ref{sec:binarysearch}) and generating optimal schedules for \emph{reduce-scatter}, \emph{allgather}, and \emph{allreduce}.
\item \textbf{Generality:} ForestColl introduces a novel transformation for switch topologies that supports their flexible traffic patterns while preserving optimal throughput (\S\ref{sec:switchremove}). ForestColl also leverages in-network multicast/aggregation capabilities when supported by switches (\S\ref{sec:innetwork}).
\item \textbf{Scalability:} Every part of ForestColl runs in polynomial time (Appendix~\ref{app:sec:runtime}), scalable to large topologies (\S\ref{sec:geneval}).
\squishenumend

We evaluated ForestColl on AMD MI250 and NVIDIA DGX A100 platforms, as well as on a large-scale 128-GPU DGX H100 cluster. At 1GB data size, ForestColl achieves 16\%$\sim$61\% higher throughput than state-of-the-art schedule generation methods on AMD and NVIDIA platforms, and 14\%$\sim$32\% higher throughput than NCCL on the 128-GPU cluster. ForestColl also reduces iteration time in PyTorch FSDP training by 20\% on 70B+ LLMs. In schedule generation, ForestColl is orders of magnitude faster than competing methods while maintaining throughput optimality.

%% file: figures/a100_topo_single.tex
\scalebox{0.6}{
\begin{tikzpicture}[node/.style={rectangle,draw=black,minimum size=7mm,align=center}]
    \node[rectangle, minimum width=12cm, minimum height=0.9cm] (20) at (0,-0.1) {\large{InfiniBand Switch Fabric}};
    \draw[thick] (20.north west) -- (20.north east);
    \draw[thick] (20.south west) -- (20.south east);

    \node[node]	(10)	at (-3.5,2.5-1) {\scalebox{0.8}[1]{GPU}};
    \node[node]	(11)	at (-2.5,2.5-1) {\scalebox{0.8}[1]{GPU}};
    \node[node]	(12)	at (-1.5,2.5-1) {\scalebox{0.8}[1]{GPU}};
    \node[node]	(13)	at (-0.5,2.5-1) {\scalebox{0.8}[1]{GPU}};
    \node[node]	(14)	at (0.5,2.5-1) {\scalebox{0.8}[1]{GPU}};
    \node[node]	(15)	at (1.5,2.5-1) {\scalebox{0.8}[1]{GPU}};
    \node[node]	(16)	at (2.5,2.5-1) {\scalebox{0.8}[1]{GPU}};
    \node[node]	(17)	at (3.5,2.5-1) {\scalebox{0.8}[1]{GPU}};
    \node[node,text width=7.5cm]	(18)	at (0,0.5+2.5) {NVSwitch};

    
    \path[latex-latex,anchor=east,line width=1pt] (10.north) edge node {300GB/s} (10.north|-18.south);
    \path[latex-latex,line width=1pt] (11.north) edge (11.north|-18.south);
    \path[latex-latex,line width=1pt] (12.north) edge (12.north|-18.south);
    \path[latex-latex,line width=1pt] (13.north) edge (13.north|-18.south);
    \path[latex-latex,line width=1pt] (14.north) edge (14.north|-18.south);
    \path[latex-latex,line width=1pt] (15.north) edge (15.north|-18.south);
    \path[latex-latex,line width=1pt] (16.north) edge (16.north|-18.south);
    \path[latex-latex,anchor=west,line width=1pt] (17.north) edge node {\phantom{300GB/s}} (17.north|-18.south);
            
    
    \path[latex-latex,anchor=east] (10.south) edge node {25GB/s} (10.south|-20.north);
    \path[latex-latex] (11.south) edge (11.south|-20.north);
    \path[latex-latex] (12.south) edge (12.south|-20.north);
    \path[latex-latex] (13.south) edge (13.south|-20.north);
    \path[latex-latex] (14.south) edge (14.south|-20.north);
    \path[latex-latex] (15.south) edge (15.south|-20.north);
    \path[latex-latex] (16.south) edge (16.south|-20.north);
    \path[latex-latex] (17.south) edge (17.south|-20.north);
\end{tikzpicture}
}

%% file: figures/mi250_topo_single.tex
\scalebox{0.6}{
\begin{tikzpicture}[node/.style={rectangle,draw=black,minimum size=7mm,align=center}]

\node[rectangle, minimum width=12cm, minimum height=0.9cm] (20) at (0,-0.1) {\large{InfiniBand Switch Fabric}};
\draw[thick] (20.north west) -- (20.north east);
\draw[thick] (20.south west) -- (20.south east);

\node[node] (0) at (-1.4, 4.890000000000001) {\scalebox{0.8}[1]{GPU}};
\node[node] (1) at (-1.4, 3.63) {\scalebox{0.8}[1]{GPU}};
\node[node] (2) at (-1.4, 2.37) {\scalebox{0.8}[1]{GPU}};
\node[node] (3) at (-1.4, 1.1099999999999999) {\scalebox{0.8}[1]{GPU}};
\node[node] (4) at (-2.8, 4.890000000000001) {\scalebox{0.8}[1]{GPU}};
\node[node] (5) at (-2.8, 3.63) {\scalebox{0.8}[1]{GPU}};
\node[node] (6) at (-2.8, 2.37) {\scalebox{0.8}[1]{GPU}};
\node[node] (7) at (-2.8, 1.1099999999999999) {\scalebox{0.8}[1]{GPU}};
\node[node] (8) at (1.4, 4.890000000000001) {\scalebox{0.8}[1]{GPU}};
\node[node] (9) at (1.4, 3.63) {\scalebox{0.8}[1]{GPU}};
\node[node] (10) at (1.4, 2.37) {\scalebox{0.8}[1]{GPU}};
\node[node] (11) at (1.4, 1.1099999999999999) {\scalebox{0.8}[1]{GPU}};
\node[node] (12) at (2.8, 4.890000000000001) {\scalebox{0.8}[1]{GPU}};
\node[node] (13) at (2.8, 3.63) {\scalebox{0.8}[1]{GPU}};
\node[node] (14) at (2.8, 2.37) {\scalebox{0.8}[1]{GPU}};
\node[node] (15) at (2.8, 1.1099999999999999) {\scalebox{0.8}[1]{GPU}};

\path[latex-latex,line width=0.75pt] (0.240) edge (0.240|-1.north);
\path[latex-latex,line width=0.75pt] (0.260) edge (0.260|-1.north);
\path[latex-latex,line width=0.75pt] (0.280) edge (0.280|-1.north);
\path[latex-latex,line width=0.75pt] (0.300) edge (0.300|-1.north);
\path[latex-latex,line width=0.75pt] (2.240) edge (2.240|-3.north);
\path[latex-latex,line width=0.75pt] (2.260) edge (2.260|-3.north);
\path[latex-latex,line width=0.75pt] (2.280) edge (2.280|-3.north);
\path[latex-latex,line width=0.75pt] (2.300) edge (2.300|-3.north);
\path[latex-latex,line width=0.75pt] (4.240) edge (4.240|-5.north);
\path[latex-latex,line width=0.75pt] (4.260) edge (4.260|-5.north);
\path[latex-latex,line width=0.75pt] (4.280) edge (4.280|-5.north);
\path[latex-latex,line width=0.75pt] (4.300) edge (4.300|-5.north);
\path[latex-latex,line width=0.75pt] (6.240) edge (6.240|-7.north);
\path[latex-latex,line width=0.75pt] (6.260) edge (6.260|-7.north);
\path[latex-latex,line width=0.75pt] (6.280) edge (6.280|-7.north);
\path[latex-latex,line width=0.75pt] (6.300) edge (6.300|-7.north);
\path[latex-latex,line width=0.75pt] (8.240) edge (8.240|-9.north);
\path[latex-latex,line width=0.75pt] (8.260) edge (8.260|-9.north);
\path[latex-latex,line width=0.75pt] (8.280) edge (8.280|-9.north);
\path[latex-latex,line width=0.75pt] (8.300) edge (8.300|-9.north);
\path[latex-latex,line width=0.75pt] (10.240) edge (10.240|-11.north);
\path[latex-latex,line width=0.75pt] (10.260) edge (10.260|-11.north);
\path[latex-latex,line width=0.75pt] (10.280) edge (10.280|-11.north);
\path[latex-latex,line width=0.75pt] (10.300) edge (10.300|-11.north);
\path[latex-latex,line width=0.75pt] (12.240) edge (12.240|-13.north);
\path[latex-latex,line width=0.75pt] (12.260) edge (12.260|-13.north);
\path[latex-latex,line width=0.75pt] (12.280) edge (12.280|-13.north);
\path[latex-latex,line width=0.75pt] (12.300) edge (12.300|-13.north);
\path[latex-latex,line width=0.75pt] (14.240) edge (14.240|-15.north);
\path[latex-latex,line width=0.75pt] (14.260) edge (14.260|-15.north);
\path[latex-latex,line width=0.75pt] (14.280) edge (14.280|-15.north);
\path[latex-latex,line width=0.75pt] (14.300) edge (14.300|-15.north);

\path[latex-latex,line width=0.75pt] (4.10) edge (4.10-|0.west);
\path[latex-latex,line width=0.75pt] (4.350) edge (4.350-|0.west);
\path[latex-latex,line width=0.75pt] (7.10) edge (7.10-|3.west);
\path[latex-latex,line width=0.75pt] (7.350) edge (7.350-|3.west);
\path[latex-latex,line width=0.75pt] (8.10) edge (8.10-|12.west);
\path[latex-latex,line width=0.75pt] (8.350) edge (8.350-|12.west);
\path[latex-latex,line width=0.75pt] (11.10) edge (11.10-|15.west);
\path[latex-latex,line width=0.75pt] (11.350) edge (11.350-|15.west);

\path[latex-latex,line width=0.75pt] (5.10) edge (1.170);
\path[latex-latex,line width=0.75pt] (6.10) edge (2.170);
\path[latex-latex,line width=0.75pt] (0.10) edge (8.170);
\path[latex-latex,line width=0.75pt] (1.10) edge (9.170);
\path[latex-latex,line width=0.75pt] (9) edge (2);
\path[latex-latex,line width=0.75pt] (10) edge (1);
\path[latex-latex,line width=0.75pt] (2.10) edge (10.170);
\path[latex-latex,line width=0.75pt] (3.10) edge (11.170);
\path[latex-latex,line width=0.75pt] (9.10) edge (13.170);
\path[latex-latex,line width=0.75pt] (10.10) edge (14.170);
\path[latex-latex,line width=0.75pt] (5) edge (6);
\path[latex-latex,line width=0.75pt] (13) edge (14);

\draw[latex-latex,line width=0.75pt] (4.170) -| (-3.5, 2.43) |- (6.170);
\draw[latex-latex,line width=0.75pt] (5.170) -| (-3.6999999999999997, 2.43) |- (7.170);
\draw[latex-latex,line width=0.75pt] (12.10) -| (3.5, 2.43) |- (14.10);
\draw[latex-latex,line width=0.75pt] (13.10) -| (3.6999999999999997, 2.43) |- (15.10);

\draw[latex-latex,gray] (0.350) -| (-0.1,0.35);
\draw[latex-latex,gray] (1.350) -| (-0.30000000000000004,0.35);
\draw[latex-latex,gray] (2.350) -| (-0.5,0.35);
\draw[latex-latex,gray] (3.350) -| (-0.7000000000000001,0.35);
\draw[latex-latex,gray] (4.190) -| (-4.5,0.35);
\draw[latex-latex,gray] (5.190) -| (-4.3,0.35);
\draw[latex-latex,gray] (6.190) -| (-4.1,0.35);
\draw[latex-latex,gray] (7.190) -| (-3.8999999999999995,0.35);
\draw[latex-latex,gray] (8.190) -| (0.1,0.35);
\draw[latex-latex,gray] (9.190) -| (0.30000000000000004,0.35);
\draw[latex-latex,gray] (10.190) -| (0.5,0.35);
\draw[latex-latex,gray] (11.190) -| (0.7000000000000001,0.35);
\draw[latex-latex,gray] (12.350) -| (4.5,0.35);
\draw[latex-latex,gray] (13.350) -| (4.3,0.35);
\draw[latex-latex,gray] (14.350) -| (4.1,0.35);
\draw[latex-latex,gray] (15.350) -| (3.8999999999999995,0.35);

\node[gray] at (-5.15, 2.6) {16GB/s};
\node[] at (-4.1, 4.26) {50GB/s};
\end{tikzpicture}
}

%% file: figures/a100_comparison.tex
\scalebox{0.6}{
\begin{tikzpicture}[node/.style={rectangle,draw=black,minimum size=7mm,align=center}]
    \node[node,line width=1.5pt]	(10)	at (-1.5,2.5-1) {\scalebox{0.8}[1]{GPU}};
    \node[node]	(11)	at (-0.5,2.5-1) {\scalebox{0.8}[1]{GPU}};
    \node[node]	(12)	at (0.5,2.5-1) {\scalebox{0.8}[1]{GPU}};
    \node[node]	(13)	at (1.5,2.5-1) {\scalebox{0.8}[1]{GPU}};
    \node[node,text width=4cm]	(18)	at (0,0.5+2.5) {NVSwitch};

    \draw[latex-] (11.north) -- ($(18.south -| 0, 0)$);
    \draw ($(18.south -| 0, 0)$) -- (12.north);
    \draw[latex-] (12.north) -- ($(18.south -| 1, 0)$);
    \draw ($(18.south -| 1, 0)$) -- (13.north);

    \node[node]	(0)	at (-1.5,1-2.5) {\scalebox{0.8}[1]{GPU}};
    \node[node]	(1)	at (-0.5,1-2.5) {\scalebox{0.8}[1]{GPU}};
    \node[node]	(2)	at (0.5,1-2.5) {\scalebox{0.8}[1]{GPU}};
    \node[node]	(3)	at (1.5,1-2.5) {\scalebox{0.8}[1]{GPU}};
    \node[node,text width=4cm]	(8)	at (0,-0.5-2.5) {NVSwitch};

    \path[-latex] (10.south) edge (0.north);

    \draw (0.south) -- ($(8.north -| -1, 0)$);
    \draw[-latex] ($(8.north -| -1, 0)$) -- (1.south);
    \draw (1.south) -- ($(8.north -| 0, 0)$);
    \draw[-latex] ($(8.north -| 0, 0)$) -- (2.south);
    \draw (2.south) -- ($(8.north -| 1, 0)$);
    \draw[-latex] ($(8.north -| 1, 0)$) -- (3.south);

    \path[latex-] (13.south) edge (3.north);

    \node[node,text width=4.5cm]	(20)	at (0,0) {InfiniBand Switch};
\end{tikzpicture}
}

%% file: figures/a100_comparison2.tex
\scalebox{0.6}{
\begin{tikzpicture}[node/.style={rectangle,draw=black,minimum size=7mm,align=center}]
    \node[node,line width=1.5pt]	(10)	at (-1.5,2.5-1) {\scalebox{0.8}[1]{GPU}};
    \node[node]	(11)	at (-0.5,2.5-1) {\scalebox{0.8}[1]{GPU}};
    \node[node]	(12)	at (0.5,2.5-1) {\scalebox{0.8}[1]{GPU}};
    \node[node]	(13)	at (1.5,2.5-1) {\scalebox{0.8}[1]{GPU}};
    \node[node,text width=4cm]	(18)	at (0,0.5+2.5) {NVSwitch};

    \draw (10.north) -- ($(18.south -| -1, 0)$);
    \draw[-latex] ($(18.south -| -1, 0)$) -- (11.north);
    \draw (11.north) -- ($(18.south -| 0, 0)$);
    \draw[-latex] ($(18.south -| 0, 0)$) -- (12.north);
    \draw (12.north) -- ($(18.south -| 1, 0)$);
    \draw[-latex] ($(18.south -| 1, 0)$) -- (13.north);

    \node[node]	(0)	at (-1.5,1-2.5) {\scalebox{0.8}[1]{GPU}};
    \node[node]	(1)	at (-0.5,1-2.5) {\scalebox{0.8}[1]{GPU}};
    \node[node]	(2)	at (0.5,1-2.5) {\scalebox{0.8}[1]{GPU}};
    \node[node]	(3)	at (1.5,1-2.5) {\scalebox{0.8}[1]{GPU}};
    \node[node,text width=4cm]	(8)	at (0,-0.5-2.5) {NVSwitch};

    \path[-latex] (10.south) edge (0.north);

    \draw (0.south) -- ($(8.north -| -1, 0)$);
    \draw[-latex] ($(8.north -| -1, 0)$) -- (1.south);
    \draw (1.south) -- ($(8.north -| 0, 0)$);
    \draw[-latex] ($(8.north -| 0, 0)$) -- (2.south);
    \draw (2.south) -- ($(8.north -| 1, 0)$);
    \draw[-latex] ($(8.north -| 1, 0)$) -- (3.south);

    \node[node,text width=4.5cm]	(20)	at (0,0) {InfiniBand Switch};
\end{tikzpicture}
}

%% file: background.tex
\section{Background \& Related Work}

Based on the generated schedules, current schedule generation methods can be categorized into step schedules (SCCL, TACCL, TE-CCL, TACOS, BFB, SyCCL) and tree-flow schedules (Blink, MultiTree, TTO). \textbf{Step schedules} specify the data exchanged between GPUs at each step, with the entire network progressing through steps in sync. \textbf{Tree-flow schedules} let data flow fluidly through a set of spanning trees, either broadcasting from or reducing to the roots. In this section, we examine the three goals of the \emph{triathlon} and explain why existing methods fail to achieve all three simultaneously.

\textbf{Scalability:} Optimizing collective communication is computationally challenging. Unlike point-to-point traffic, where data dependencies can be enforced by flow conservation, collective operations involve one-to-many multicast and many-to-one aggregation, rendering flow conservation inapplicable. Step schedules like SCCL, TACCL, TE-CCL, and SyCCL choose to track data dependencies in discrete chunks and formulate the scheduling problem as NP-hard SMT or MILP. As a result, they struggle with even modestly sized topologies (\S\ref{sec:geneval}). In contrast, tree-flow schedules scale better as dependencies are naturally maintained through trees. However, existing methods fail to ensure optimality, as we discuss next.

\textbf{Optimality:} The performance of collective operations depends on two key metrics: \emph{throughput} and \emph{latency}. Latency, the fixed time cost incurred by send/recv hops, is critical for small data transfers. However, as data size grows, throughput becomes the dominant factor, as the bandwidth-bound transmission cost quickly outweighs the fixed latency. Step schedules are convenient for optimizing latency, as the number of steps directly corresponds to the number of hops. However, optimizing throughput with step schedules is challenging, requiring careful minimization of congestion across possibly heterogeneous links \emph{within} each step while maintaining data dependencies \emph{across} steps. Tree-flow schedules are better suited for throughput optimization, reducing the problem to minimizing congestion/overlap between trees. The remaining challenge lies in constructing optimal trees, where existing methods---such as Blink’s approximate tree packing and MultiTree’s greedy construction---are inherently suboptimal.

\textbf{Generality:} A general schedule generation should support topologies with (i) diverse graph structures, (ii) links with varying bandwidths, and (iii) switches.\footnote{Support for heterogeneous GPUs is captured by varying link bandwidths.} While most existing works address (i), support for (ii, iii) remains limited. Heterogeneous links present a challenge for step schedules, as they require synchronized step execution across the entire network. Switches add further complexity as they do not produce/consume data, and many cannot multicast/aggregate, necessitating an operating model different from GPUs.

\textbf{Related Work:} Apart from the scalability consideration discussed earlier, none of the existing methods simultaneously achieves both optimality and generality. SCCL can achieve optimality for a given number of data chunks, but the optimal chunking is unknown, and it does not support switch topologies. TACCL, TE-CCL, and SyCCL rely on heuristic tuning (e.g., sketches in TACCL and SyCCL, the reward-based objective in TE-CCL) that does not guarantee optimality. Beyond scalability issues, they are tuned and evaluated on only limited scales and topology types, offering no guarantees for broader settings (\S\ref{sec:geneval}). TACOS and MultiTree employ greedy approaches to assign traffic to links, which also do not ensure optimality. BFB and TTO provide optimality but only for specific types of non-switch topologies. Blink's tree packing constructs all trees rooted at a single node, lacking support for allgather and reduce-scatter. The single root becomes a bottleneck in allreduce, as Blink performs allreduce via reduce+broadcast. Its solution for multi-box switch settings is ad hoc and unrelated to its tree packing.

Other efforts to accelerate ML training communications, such as network infra optimizations, hybrid parallelism, and comp-comm overlap, are orthogonal and complementary to ForestColl. In particular, ForestColl's communication acceleration reduces the need for compute and memory trade-offs to hide communication costs in hybrid parallelism and comp-comm overlap. Detailed discussion is in Appendix~\ref{sec:otherrelated}.

\section{Overview of ForestColl}\label{sec:treepackingoverview}

We now introduce ForestColl. To ensure throughput optimality for throughput-bound LLM training~\cite{metardma,alibabahpn,trainium,kaichen}, ForestColl adopts tree-flow schedules. Unlike prior tree-flow methods, ForestColl leverages spanning tree packing to construct a ``forest'' of spanning trees---an equal number of trees rooted at each node---that achieves the triathlon for multi-root collectives. To accomplish this, ForestColl introduces several key techniques, including a method to compute the optimal throughput of any given topology and a topology transformation to support throughput optimality in switch networks.

\textbf{Spanning Tree Packing} is a well-studied topic in graph theory that focuses on determining the maximum number of spanning trees that can be constructed in a graph given edge capacities~\cite{bang-jensen, tarjan, edmonds, berczi2010packing, schrijver}. In ForestColl's tree-flow schedules, each tree occupies and utilizes an equal share of bandwidth, making spanning tree packing a useful tool for constructing trees that make optimal use of the available link bandwidths. While efficient and optimal tree-packing algorithms have been proposed in graph theory, these algorithms are not directly applicable to our schedule generation, leaving several challenges.

\begin{figure}[tb]
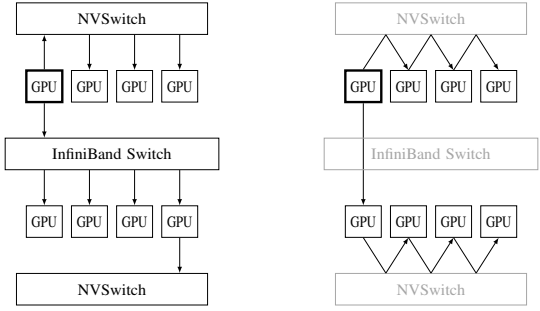

    \centering
    \begin{subfigure}{0.49\columnwidth}
        \centering
        \include{figures/a100_small_wrong}
        \caption{Direct Tree Construction}
    \end{subfigure}
    \begin{subfigure}{0.49\columnwidth}
        \centering
        \include{figures/a100_small_correct}
        \caption{Desired ``Tree'' Construction}
    \end{subfigure}
    \caption{Example of spanning tree construction on a switch topology. \textnormal{(a) shows a spanning tree constructed directly on the input switch topology, resulting in two issues: (1) the construction assumes switches are capable of in-network multicast/aggregation, which is not always supported; (2) the tree unnecessarily spans the bottom NVSwitch, which does not consume data. (b)~shows the desired ``tree'' construction, as provided by ForestColl, which is a spanning tree among GPU nodes only. Switches are not part of the tree but serve only to provide connections between the GPUs.}}
    \label{fig:switchtree}
\end{figure}

\textbf{Technical Challenges:} Traditional spanning tree packing defines link capacity as the number of trees a link can support. As a result, directly using link bandwidth as capacity in our case would require constructing hundreds or even thousands of trees, as bandwidth values are typically large. This presents a scalability challenge: \textit{\ul{How to achieve throughput optimality with a small number of trees?}} To overcome this challenge, we choose to first determine the optimal bandwidth each tree occupies and scale link capacities accordingly before applying tree packing. This, however, raises a broader, previously unsolved optimality challenge: \textit{\ul{How can the optimal collective communication throughput of a topology be determined?}} Finally, as shown in Figure~\ref{fig:switchtree}, the traffic patterns of network switches render the desired schedule no longer properly defined by spanning trees, raising a generality challenge: \textit{\ul{How to apply spanning tree packing while supporting the unique traffic patterns of switches?}}

\textbf{Overview:} ForestColl effectively solves the challenges. We studied the throughput optimality of a topology and found that it is determined by the \emph{throughput bottleneck cut}---the network cut with the highest ratio of minimal required traffic to available bandwidth (\S\ref{sec:throughputoptimality}). In ForestColl, we combine binary search and network flow to compute this optimality, determining the optimal number of trees and the bandwidth each tree occupies (\S\ref{sec:binarysearch}). Once edge capacities are scaled accordingly, the tree packing algorithm can be applied to construct the trees (\S\ref{sec:treeconstruct}). For switches, we devised a way to transform a switch topology into a switch-free logical topology before applying tree packing (\S\ref{sec:switchremove}). Unlike TACCL and TACOS, which remove switches and connect nodes in preset patterns that overlook the performance impact of losing switches' all-to-all connectivity, ForestColl's transformation ensures no compromise in performance. Together, ForestColl achieves the triathlon of ideal schedule generation.

\textbf{Limitations:} ForestColl does not solve all problems in collective communication. First, efficient communication requires not only effective scheduling but also optimized implementations for different hardware. ForestColl derives optimal schedules to guide such implementations. Second, ForestColl prioritizes throughput over latency. Latency is better optimized at the implementation level (e.g., through low-latency protocols and CUDA kernels), and low-latency scheduling has been extensively studied in step schedules. Although ForestColl does not guarantee latency optimality, it still achieves strong latency performance for small data sizes in practice (\S\ref{sec:evaluation}). Third, ForestColl is designed as an offline schedule generator that runs once per topology, rather than real-time online scheduling in subseconds. A generation time of minutes is trivial when amortized over cluster setup and model training that run for hours or even days. Fourth, ForestColl targets allreduce-type collectives used in data/tensor/context parallelism in ML, but not P2P all-to-all communication used in expert parallelism, which often requires workload-dynamic online scheduling. Finally, ForestColl does not exploit network symmetry to speedup scheduling. While it can generate symmetric schedules, preserving both symmetry and optimality does not yield performance benefits. Since the current algorithm has no scalability issues, leveraging network symmetry is left for future work.

%% file: figures/a100_small_wrong.tex
\scalebox{0.6}{
\begin{tikzpicture}[node/.style={rectangle,draw=black,minimum size=7mm,align=center}]
    \node[node,line width=1.5pt]	(10)	at (-1.5,2.5-1) {\scalebox{0.8}[1]{GPU}};
    \node[node]	(11)	at (-0.5,2.5-1) {\scalebox{0.8}[1]{GPU}};
    \node[node]	(12)	at (0.5,2.5-1) {\scalebox{0.8}[1]{GPU}};
    \node[node]	(13)	at (1.5,2.5-1) {\scalebox{0.8}[1]{GPU}};
    \node[node,text width=4cm]	(18)	at (0,0.5+2.5) {NVSwitch};

    \path[-latex] (10.north) edge (10.north|-18.south);
    \path[latex-] (11.north) edge (11.north|-18.south);
    \path[latex-] (12.north) edge (12.north|-18.south);
    \path[latex-] (13.north) edge (13.north|-18.south);

    \node[node]	(0)	at (-1.5,1-2.5) {\scalebox{0.8}[1]{GPU}};
    \node[node]	(1)	at (-0.5,1-2.5) {\scalebox{0.8}[1]{GPU}};
    \node[node]	(2)	at (0.5,1-2.5) {\scalebox{0.8}[1]{GPU}};
    \node[node]	(3)	at (1.5,1-2.5) {\scalebox{0.8}[1]{GPU}};
    \node[node,text width=4cm]	(8)	at (0,-0.5-2.5) {NVSwitch};

    \node[node,text width=4.5cm]	(20)	at (0,0) {InfiniBand Switch};

    \path[-latex] (3.south) edge (3.south|-8.north);
    
    \path[latex-] (0.north) edge (0.north|-20.south);
    \path[latex-] (1.north) edge (1.north|-20.south);
    \path[latex-] (2.north) edge (2.north|-20.south);
    \path[latex-] (3.north) edge (3.north|-20.south);
    
    \path[-latex] (10.south) edge (10.south|-20.north);
\end{tikzpicture}
}

%% file: figures/a100_small_correct.tex
\scalebox{0.6}{
\begin{tikzpicture}[node/.style={rectangle,draw=black,minimum size=7mm,align=center}]
    \node[node,line width=1.5pt]	(10)	at (-1.5,2.5-1) {\scalebox{0.8}[1]{GPU}};
    \node[node]	(11)	at (-0.5,2.5-1) {\scalebox{0.8}[1]{GPU}};
    \node[node]	(12)	at (0.5,2.5-1) {\scalebox{0.8}[1]{GPU}};
    \node[node]	(13)	at (1.5,2.5-1) {\scalebox{0.8}[1]{GPU}};
    \node[node,gray!70,text width=4cm]	(18)	at (0,0.5+2.5) {NVSwitch};

    \draw (10.north) -- ($(18.south -| -1, 0)$);
    \draw[-latex] ($(18.south -| -1, 0)$) -- (11.north);
    \draw (11.north) -- ($(18.south -| 0, 0)$);
    \draw[-latex] ($(18.south -| 0, 0)$) -- (12.north);
    \draw (12.north) -- ($(18.south -| 1, 0)$);
    \draw[-latex] ($(18.south -| 1, 0)$) -- (13.north);

    \node[node]	(0)	at (-1.5,1-2.5) {\scalebox{0.8}[1]{GPU}};
    \node[node]	(1)	at (-0.5,1-2.5) {\scalebox{0.8}[1]{GPU}};
    \node[node]	(2)	at (0.5,1-2.5) {\scalebox{0.8}[1]{GPU}};
    \node[node]	(3)	at (1.5,1-2.5) {\scalebox{0.8}[1]{GPU}};
    \node[node,gray!70,text width=4cm]	(8)	at (0,-0.5-2.5) {NVSwitch};

    \path[-latex] (10.south) edge (0.north);

    \draw (0.south) -- ($(8.north -| -1, 0)$);
    \draw[-latex] ($(8.north -| -1, 0)$) -- (1.south);
    \draw (1.south) -- ($(8.north -| 0, 0)$);
    \draw[-latex] ($(8.north -| 0, 0)$) -- (2.south);
    \draw (2.south) -- ($(8.north -| 1, 0)$);
    \draw[-latex] ($(8.north -| 1, 0)$) -- (3.south);

    \node[node,gray!70,text width=4.5cm]	(20)	at (0,0) {InfiniBand Switch};
\end{tikzpicture}
}

%% file: problem.tex
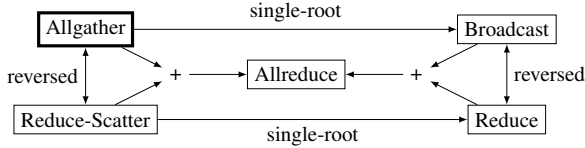
\begin{figure}[tb]
    \centering
    \input{figures/collective_relation}
    \caption{Relationships between collective operations. \textnormal{Reduce and reduce-scatter can be constructed by reversing the communications of broadcast and allgather~\cite{conc_comp}. Reduce and broadcast are single-root versions of reduce-scatter and allgather. Finally, allreduce can be performed via reduce-scatter plus allgather or reduce plus broadcast. While this paper focuses on allgather, the method applies to other operations.}}
    \label{fig:collective-rel}
\end{figure}

\section{Throughput Optimality for Collectives}\label{sec:throughputoptimality}

Collective operations can be classified into \textit{aggregation only} (e.g., reduce, reduce-scatter), \textit{broadcast only} (e.g., broadcast, allgather), and \textit{aggregation plus broadcast} (e.g., allreduce). Aggregation requires \emph{in-trees}, with edge directions flowing from leaves to the root, while broadcast requires \emph{out-trees}, where edges flow from the root to leaves. In terms of roots, collective operations can also be categorized as \textit{single-root} (e.g., reduce, broadcast) or \textit{multi-root} (e.g., reduce-scatter, allgather, allreduce). Figure~\ref{fig:collective-rel} shows the relationships between operations. While we explain ForestColl in the context of allgather, it can be easily applied to other operations (\S\ref{sec:othercollectives}).

\begin{figure}[t]
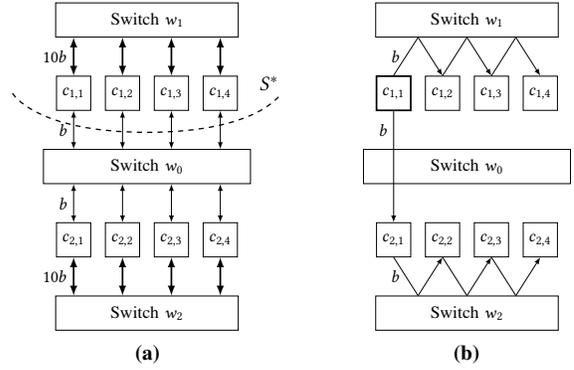

    \centering
    \begin{subfigure}{0.49\columnwidth}
        \centering
        \include{figures/topo}
        \caption{}
    \end{subfigure}
    \begin{subfigure}{0.49\columnwidth}
        \centering
        \include{figures/topo_tree2}
        \caption{}
    \end{subfigure}
    \caption{Example of Spanning Out-Tree. \textnormal{(a) shows a 2-box 8-compute-node switch topology along with the throughput bottleneck cut. The intra-box connections (thick lines) have 10x the bandwidth of inter-box ones (thin lines). (b) shows one example of ForestColl's spanning out-trees that is rooted at compute node $c_{1,1}$.}}
    \label{fig:tree-example}
\end{figure}

Knowing the throughput optimality of a given network is crucial for optimizing schedules. Previous work often defines optimality as \mbox{$\frac{M(N\!-\!1)}{N}\!\cdot\!\beta$}~\cite{sccl,bfb,conc_comp}, where $\frac{M(N\!-\!1)}{N}$ is the amount of data each node must receive in allgather, and $\beta$ is the time cost per unit of data. However, this definition only holds when the bottleneck is each individual node's bandwidth. In ML hardware, the bottleneck is often the scale-out network, e.g., the IB bandwidth of a multi-GPU box. In this section, we introduce the concept of \textit{throughput bottleneck cut}, which determines the throughput optimality of allgather.

We model a network topology as a directed graph $G$, where edge capacities signify link bandwidths, and the vertex set $V$ consists of \textit{compute nodes}~$V_c$ (e.g., GPUs) and \textit{switch nodes}~$V_s$. Figure~\ref{fig:tree-example}(a) shows an example. In allgather, each compute node needs to broadcast an equal shard of data to all other compute nodes. We denote the total amount of data $M$, the number of compute nodes $|V_c|\!=\!N$, and thus shard size $\frac{M}{N}$.

In Figure~\ref{fig:tree-example}(a), consider the network cut $S^*$, which contains all nodes in the top box: compute nodes $c_{1,1},c_{1,2},c_{1,3},c_{1,4}$ and switch node $w_1$. To finish an allgather, each compute node within $S^*$ must send at least one copy of its shard across the cut to the bottom box; otherwise, $c_{2,1},c_{2,2},c_{2,3},c_{2,4}$ will fail to receive some shards. Therefore, at least $4\!\cdot\!\frac{M}{N}$ amount of data has to exit cut $S^*$. Note that the total bandwidth exiting $S^*$ is $4b$, counting all four links connecting $c_{1,*}$ to the inter-box switch $w_0$. Thus, a lower bound for the allgather communication time in this topology is $4\!\cdot\!\frac{M}{N}/(4b)\!=\!\frac{M}{8b}$.

The lower bound can be generalized to any topology $G$. Given an arbitrary network cut $S\!\subset\! V$ in $G$, if there is any compute node not in $S$ (i.e., $S\!\not\supseteq\! V_c$), then at least $\frac{M}{N}|S\cap V_c|$ amount of data has to exit $S$. Let $B^+(S)$ denote the exiting bandwidth of $S$, i.e., the sum of bandwidths of links going from $S$ to $V\!-\!S$, then $\frac{M}{N}\!\cdot\!\frac{|S\cap V_c|}{B^+(S)}$ is a lower bound for allgather communication time $T_{\text{comm}}$ in topology $G$. Consider all such cuts in $G$, then $T_{\text{comm}}$ satisfies
\begin{equation*}
    T_{\text{comm}}\geq\frac{M}{N}\max_{S\subset V,S\not\supseteq V_c}\frac{|S\cap V_c|}{B^+(S)}.
    \tag{$\star$}\label{eq:optimality}
\end{equation*}
We call any cut $S$ that maximizes $\frac{|S\cap V_c|}{B^+(S)}$, the ratio of compute nodes within the cut to the exiting bandwidth, as the \emph{throughput~bottleneck~cut}\footnote{Different from traditional min cut, which only minimizes $B^+(S)$.}. In this work, we present ForestColl, which can achieve the RHS of (\ref{eq:optimality}). Since (\ref{eq:optimality}) is a lower bound of allgather time, ForestColl achieves throughput optimality.

%% file: figures/collective_relation.tex
\scalebox{0.8}{
\begin{tikzpicture}
    \node[rectangle,draw=black, line width=1.5pt] (0) at (0,0) {Allgather};
    \node[rectangle,draw=black] (1) at (0,-1.5) {Reduce-Scatter};
    \node[rectangle,draw=black] (2) at (7,0) {Broadcast};
    \node[rectangle,draw=black] (3) at (7,-1.5) {Reduce};
    \node (4) at (1.5,-0.75) {$+$};
    \node[rectangle,draw=black] (5) at (3.5,-0.75) {Allreduce};
    \node (6) at (5.5,-0.75) {$+$};

    \path[latex-latex,anchor=east] (0) edge node {reversed} (1);
    \path[latex-latex,anchor=west] (2) edge node {reversed} (3);
    \path[-latex,anchor=south] (0) edge node {single-root} (2);
    \path[-latex,anchor=north] (1) edge node {single-root} (3);

    \path[-latex] (0) edge (4);
    \path[-latex] (1) edge (4);
    \path[-latex] (4) edge (5);
    \path[-latex] (2) edge (6);
    \path[-latex] (3) edge (6);
    \path[-latex] (6) edge (5);
\end{tikzpicture}
}

%% file: figures/topo.tex
\scalebox{0.65}{
\begin{tikzpicture}[node/.style={rectangle,draw=black,minimum size=7mm,align=center}]
    \node[node]	(10)	at (-1.5,2.5-1) {$c_{1,1}$};
    \node[node]	(11)	at (-0.5,2.5-1) {$c_{1,2}$};
    \node[node]	(12)	at (0.5,2.5-1) {$c_{1,3}$};
    \node[node]	(13)	at (1.5,2.5-1) {$c_{1,4}$};
    \node[node,text width=3.5cm]	(14)	at (0,0.5+2.5) {Switch $w_1$};
    
    \path[latex-latex,anchor=east,line width=1pt] (10.north) edge node {$10b$} (10.north|-14.south);
    \path[latex-latex,line width=1pt] (11.north) edge (11.north|-14.south);
    \path[latex-latex,line width=1pt] (12.north) edge (12.north|-14.south);
    \path[latex-latex,line width=1pt] (13.north) edge (13.north|-14.south);
    
    \node[node]	(0)	at (-1.5,1-2.5) {$c_{2,1}$};
    \node[node]	(1)	at (-0.5,1-2.5) {$c_{2,2}$};
    \node[node]	(2)	at (0.5,1-2.5) {$c_{2,3}$};
    \node[node]	(3)	at (1.5,1-2.5) {$c_{2,4}$};
    \node[node,text width=3.5cm]	(4)	at (0,-0.5-2.5) {Switch $w_2$};
    
    \path[latex-latex,anchor=east,line width=1pt] (0.south) edge node {$10b$} (0.south|-4.north);
    \path[latex-latex,line width=1pt] (1.south) edge (1.south|-4.north);
    \path[latex-latex,line width=1pt] (2.south) edge (2.south|-4.north);
    \path[latex-latex,line width=1pt] (3.south) edge (3.south|-4.north);
            
    \node[node,text width=4cm]	(20)	at (0,0) {Switch $w_0$};
    
    \path[latex-latex,anchor=east] (0.north) edge node {$b$} (0.north|-20.south);
    \path[latex-latex] (1.north) edge (1.north|-20.south);
    \path[latex-latex] (2.north) edge (2.north|-20.south);
    \path[latex-latex] (3.north) edge (3.north|-20.south);
    
    \path[latex-latex,anchor=east] (10.south) edge node {$b$} (10.south|-20.north);
    \path[latex-latex] (11.south) edge (11.south|-20.north);
    \path[latex-latex] (12.south) edge (12.south|-20.north);
    \path[latex-latex] (13.south) edge (13.south|-20.north);

    \draw[line width=0.75pt, dashed] plot [smooth, tension=1.5] coordinates { (-2.75,1.5) (0,0.7) (2.75,1.5) };
    \node[font=\Large] at (2.5, 1.7) {$S^*$};
\end{tikzpicture}
}

%% file: figures/topo_tree2.tex
\scalebox{0.65}{
\begin{tikzpicture}[node/.style={rectangle,draw=black,minimum size=7mm,align=center}]
    \node[node, line width=1pt]	(10)	at (-1.5,2.5-1) {$c_{1,1}$};
    \node[node]	(11)	at (-0.5,2.5-1) {$c_{1,2}$};
    \node[node]	(12)	at (0.5,2.5-1) {$c_{1,3}$};
    \node[node]	(13)	at (1.5,2.5-1) {$c_{1,4}$};
    \node[node,text width=3.5cm]	(14)	at (0,0.5+2.5) {Switch $w_1$};
    
    \draw (10.north) -- node [anchor=east] {$b$} ($(14.south -| -1, 0)$);
    \draw[-latex] ($(14.south -| -1, 0)$) -- (11.north);
    
    \draw (11.north) -- ($(14.south -| 0, 0)$);
    \draw[-latex] ($(14.south -| 0, 0)$) -- (12.north);
    
    \draw (12.north) -- ($(14.south -| 1, 0)$);
    \draw[-latex] ($(14.south -| 1, 0)$) -- (13.north);

    \node[node]	(0)	at (-1.5,1-2.5) {$c_{2,1}$};
    \node[node]	(1)	at (-0.5,1-2.5) {$c_{2,2}$};
    \node[node]	(2)	at (0.5,1-2.5) {$c_{2,3}$};
    \node[node]	(3)	at (1.5,1-2.5) {$c_{2,4}$};
    \node[node,text width=3.5cm]	(4)	at (0,-0.5-2.5) {Switch $w_2$};
            
    \node[node,text width=4cm]	(20)	at (0,0) {Switch $w_0$};
    
    \path[-latex] (10.south) edge (0.north);
    
    \draw (0.south) -- node [anchor=east] {$b$} ($(4.north -| -1, 0)$);
    \draw[-latex] ($(4.north -| -1, 0)$) -- (1.south);
    
    \draw (1.south) -- ($(4.north -| 0, 0)$);
    \draw[-latex] ($(4.north -| 0, 0)$) -- (2.south);
    
    \draw (2.south) -- ($(4.north -| 1, 0)$);
    \draw[-latex] ($(4.north -| 1, 0)$) -- (3.south);

    \path[anchor=east] (10.south) edge [draw=none] node {$b$} (10.south|-20.north);
\end{tikzpicture}
}

%% file: algorithm.tex
\section{Algorithm Design}\label{sec:algorithm}

In this section, we delve into the details of ForestColl's algorithm, which solves the following problem:

\begin{mdframed}[innerleftmargin=3pt,innerrightmargin=3pt,innertopmargin=3pt,innerbottommargin=3pt,nobreak=true]
\textbf{ForestColl Problem Definition}

\noindent\textbf{Input:} A topology\footnotemark\ modeled as a directed graph $G$ with integer link bandwidths and a vertex set $V$ consisting of compute nodes $V_c$ and switch nodes $V_s$.

\noindent\textbf{Output:} A set of spanning out-trees $\{T_{u,i}\}_{u\in V_c,i\in [k]}$ over compute nodes, where each tree occupies an equal amount of bandwidth and collectively, they achieve optimality~(\ref{eq:optimality}).
\end{mdframed}

\footnotetext{Each node must have equal total ingress and egress bandwidth. \textit{Does not exclude oversubscription,} as network tiers can still have varying bandwidths.}

The set of spanning out-trees $\{T_{u,i}\}_{u\in V_c,i\in [k]}$ consists of $k$ trees rooted at each compute node $u$, with $k$ determined algorithmically. Correspondingly, a $1/k$ shard of data is broadcast along each out-tree simultaneously.
Note that the out-trees are spanning trees of compute nodes only, as explained in Figure~\ref{fig:switchtree}. Figure~\ref{fig:tree-example}(b) shows an example of the out-tree with allocated bandwidth $b$.

This section introduces the high-level intuitions and steps of ForestColl's algorithms. We provide detailed mathematical analysis in Appendix~\ref{app:sec:algodesign} and proofs in Appendix~\ref{app:sec:proofs}. Appendix~\ref{app:sec:notations} includes a summary of notations used in this paper.

\input{overview}
\input{binary_search}

\input{edge_splitting}
\input{tree_construction}
\input{fixed_k}
\input{in_network}

%% file: overview.tex
\subsection{Algorithm Overview}

ForestColl starts with a binary search to compute the optimality (\ref{eq:optimality}) established by throughput bottleneck cut. Iterating through all cuts to find the bottleneck cut is intractable due to the exponential number of possible cuts. Instead, we design an auxiliary network on which we can compute maxflow to determine if a given value is $\geq$ or $<$ than optimality, thus enabling a binary search. Knowing the optimality is crucial for deciding the number of trees per compute node (i.e., $k$) and the bandwidth per tree to achieve optimality.

In a switch-free topology, after knowing the bandwidth per tree and the number of trees, we directly apply \emph{spanning tree packing}~\cite{bang-jensen, tarjan, edmonds, berczi2010packing, schrijver} to construct the optimal set of out-trees. In a switch topology, however, we retrofit the \emph{edge splitting technique}~\cite{bang-jensen,frank,jackson} to eliminate switch nodes before constructing spanning trees. We replace each switch node with direct logical links between its neighboring nodes, creating a switch-free logical topology where spanning tree packing can be applied to construct throughput-optimal trees. A post-processing step can further enable in-network switch multicast/aggregation.

%% file: binary_search.tex
\begin{figure}[tb]
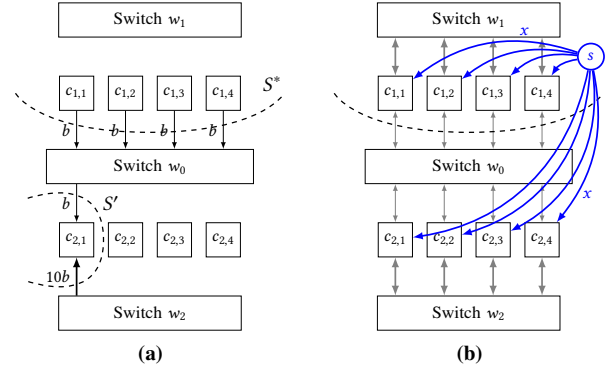

    \centering
    \begin{subfigure}{0.49\columnwidth}
        \centering
        \include{figures/flow}
        \caption{}
    \end{subfigure}
    \begin{subfigure}{0.49\columnwidth}
        \centering
        \include{figures/flow2}
        \caption{}
    \end{subfigure}
    \caption{The auxiliary network for optimality binary search. \textnormal{(a) shows two cuts: $S^*,S'$, along with their exiting bandwidths. Note that $S'$ is $V\!-\! c_{2,1}$ instead of $\{c_{2,1}\}$. (b) shows the auxiliary network that there exists a set of spanning out-trees broadcasting $x$ amount of flow from each compute node if and only if the maxflow from $s$ to every compute node is $Nx$.}}
    \label{fig:binary}
\end{figure}

\subsection{Optimality Binary Search}\label{sec:binarysearch}

We present ForestColl's binary search to compute the throughput optimality~(\ref{eq:optimality}). 
Let the total bandwidth of the out-trees rooted at each compute node be $x$. As each node simultaneously broadcasts a shard of data, the communication time $T_{\text{comm}}\!=\!\frac{M}{N}\!\cdot\!\frac{1}{x}.$
Therefore, to minimize $T_{\text{comm}}$, we need to maximize $x$. The goal of the binary search is to \textit{find the maximum $x$ such that there exists a set of spanning out-trees broadcasting $x$ amount of flow from each compute node.} We denote the maximum such $x$ as $x^*$.
Before describing the binary search for computing $x^*$, we first show that $\frac{1}{x^*}$ is precisely the ratio of compute nodes to exiting bandwidth at the throughput bottleneck cut, i.e., $\frac{1}{x^*}\!=\!\max_{S\subset V,S\not\supseteq V_c}\frac{|S\cap V_c|}{B^+(S)}$ in optimality~(\ref{eq:optimality}).

Since each compute node broadcasts $x$ amount of flow to every other compute node, the exiting flow of any cut $S$ is at least $|S\cap V_c|x$ if there is a compute node outside of $S$.
Figure~\ref{fig:binary}(a) shows two such cuts: $S^*$ and $S'$. $S^*$ includes four compute nodes, resulting in an exiting flow of $4x$. $S'$ includes \textit{all} compute and switch nodes except $c_{2,1}$, with an exiting flow of $7x$ to $c_{2,1}$. Suppose $x\!=\!b$ (the inter-box link bandwidth). For cut $S'$, the exiting bandwidth $B^+(S')$ is $11b$, more than sufficient for the exiting flow $7b$. However, for cut $S^*$, the exiting bandwidth $B^+(S^*)$ is exactly $4b$, equal to the required amount of exiting flow. Thus, we are bottlenecked by $S^*$: if $x\!>\!b$, cut $S^*$ cannot sustain the cumulative exiting flow from $c_{1,*}$ to $c_{2,*}$ anymore. Consequently, $x^*\!=\!b$ bounds the maximum flow each compute node can simultaneously broadcast. In an arbitrary topology, as we increase $x$, we will always be bottlenecked by a cut like $S^*$. This cut is exactly the throughput bottleneck cut in optimality~(\ref{eq:optimality}), where $\frac{1}{x^*}\!=\!\frac{|S^*\cap V_c|}{B^+(S^*)}\!=\!\max_{S\subset V,S\not\supseteq V_c}\frac{|S\cap V_c|}{B^+(S)}$. Therefore, $x^*$ is the maximum $x$ that does not overwhelm any cut in the topology.

\textbf{Detect Overwhelmed Cut:} To conduct a binary search for $x^*$, given a value of $x$, we determine if $x\!\leq\! x^*$ or $>\!x^*$ by detecting if $x$ overwhelms any cut in the topology. This presents a complex problem because (i)~both the amount of exiting flow and the bandwidth of the cut need to be considered, e.g., $S'$'s bandwidth is not saturated when $x\!=\!b$ despite having a larger exiting flow than $S^*$; (ii)~testing every possible cut is intractable due to the exponential number of cuts. 

\textbf{Auxiliary Network:} To address the above issues, we construct an auxiliary network as in Figure~\ref{fig:binary}(b). We add a source node $s$ and connect $s$ to every compute node with capacity $x$. Suppose we want to check if $S^*$ is overwhelmed. We pick an arbitrary compute node outside of $S^*$, say $c_{2,2}$, and calculate the maxflow from $s$ to $c_{2,2}$. If no cut is overwhelmed, then $c_{2,2}$ should get all the flow that $s$ can emit, which equals~$8x$. However, if we set $x\!>\!b$, while the $4x$ amount of flow from $s$ to $c_{2,*}$ (compute nodes outside of $S^*$) can directly bypass $B^+(S^*)$, the $4x$ amount of flow from $s$ to $c_{1,*}$ (compute nodes within $S^*$) must pass through $B^+(S^*)$ to reach $c_{2,2}$, capped at~$4b$. Thus, if $x\!>\!b$, then the maxflow from $s$ to $c_{2,2}$ is capped at~$4x\!+\!4b$, which is less than $8x$, signaling a cut being overwhelmed. \textit{The maxflow to $c_{2,2}$ checks all cuts that do not contain $c_{2,2}$.} To check all the exponential number of cuts in the network, we only need to compute a maxflow from $s$ to every compute node $c$. If the maxflow from $s$ to any $c$ is $<\!Nx$, then some cut between $s$ and $c$ is overwhelmed, indicating $x\!>\!x^*$; otherwise, $x\!\leq\!x^*$ for the binary search.

{\small
\begin{algorithm}[tb]
    \small
    \caption{Optimality Binary Search}
    \label{algo:binarysearch}
    \SetAlgoLined
    \linespread{0.9}\selectfont
    \DontPrintSemicolon
    \KwIn{A directed graph $G=(V_s\cup V_c,E)$}
    \KwOut{$\frac{1}{x^*}\!=\!\max_{S\subset V,S\not\supseteq V_c}\frac{|S\cap V_c|}{B^+(S)}$}
    \Begin{
        $l\leftarrow\frac{N-1}{\min_{v\in V_c}B^-(v)}$\tcp*{\normalfont\textit{a lower bound of $\frac{1}{x^*}$}}
        $r\leftarrow N-1$\tcp*{\normalfont\textit{an upper bound of $\frac{1}{x^*}$}}
        \While{$r-l\geq 1/\min_{v\in V_c}B^-(v)^2$}{
            $\frac{1}{x}\leftarrow(l+r)/2$\;
            Add node $s$ to $G$.\;
            \ForEach{\normalfont compute node $c\in V_c$}{
                Add an edge from $s$ to $c$ with capacity $x$.\;
            }
            \eIf{\normalfont the maxflow from $s$ to each $c\in V_c$ is $Nx$}{
                $r\leftarrow\frac{1}{x}$\tcp*{\normalfont\textit{case $\frac{1}{x}\!\geq\!\frac{1}{x^*}$}}
            }{
                $l\leftarrow\frac{1}{x}$\tcp*{\normalfont\textit{case $\frac{1}{x}\!<\!\frac{1}{x^*}$}}
            }
        }
        Find the unique fractional number $\frac{p}{q}\in[l,r]$ such that denominator $q\leq\min_{v\in V_c}B^-(v)$.\;
        \Return{\normalfont $\frac{p}{q}$ as $\frac{1}{x^*}$}
    }
\end{algorithm}
}

\textbf{Binary Search:} Algorithm~\ref{algo:binarysearch} shows ForestColl's search for $\frac{1}{x^*}$.
We iteratively narrow $l$ and $r$ by adjusting the edge capacities from $s$ and recomputing the maxflows to determine if the midpoint $(l+r)/2$ is $\geq$ or $<$ than $\frac{1}{x^*}$.
Thus, we can shrink the range $[l,r]$ small enough for us to determine $\frac{1}{x^*}$ exactly by finding the unique fractional number $\frac{p}{q}$ within $[l,r]$ with denominator $q\!\leq\!\min_{v\in V_c}B^-(v)$ (details in Appendix~\ref{app:sec:binarysearch}).

\textbf{Determine $k$:} As previously mentioned, knowing the optimality $x^*$ helps us decide the number of trees rooted at each compute node (i.e., $k$) and the bandwidth allocated per tree. In the spanning tree packing and edge splitting algorithms that ForestColl will apply later, each unit of edge capacity is interpreted as the allocation of one tree instead of one unit of bandwidth. Suppose $y$ is the bandwidth of each tree. Then, we need to adjust the edge capacities by dividing the bandwidth of each edge $b_e$ by $y$, so that the new capacity $b_e/y$ is the number of trees edge $e$ can sustain. This leads to two requirements for $y$: (i) $k\!=\!x^*/y$ must be an integer, and (ii) $b_e/y$ must be an integer for all edge bandwidth $b_e$. In Algorithm~\ref{algo:binarysearch}, we have computed $\frac{1}{x^*}\!=\!\frac{p}{q}$. Thus, by setting $y\!=\!\gcd(q,\{b_e\}_{e\in E})/p$, we ensure that both requirements are satisfied, and $k$, the number of trees rooted at each compute node, is simply $x^*/y$. For example, the optimality of Figure~\ref{fig:tree-example}(a) is $\frac{1}{x^*}\!=\!\frac{4}{4b}\!=\!\frac{1}{b}$ bottlenecked by $S^*$. We have $y\!=\!\gcd\{b,b,10b\}\!=\!b$, so the bandwidths of edges are scaled from $\{b,10b\}$ to $\{1,10\}$, and $k\!=\!1$. Figure~\ref{fig:splitting}(a) shows the resulting topology. A detailed mathematical analysis of the binary search and the derivation of $k$ is included in Appendix~\ref{app:sec:binarysearch}.

%% file: figures/flow.tex
\scalebox{0.65}{
\begin{tikzpicture}[node/.style={rectangle,draw=black,minimum size=7mm,align=center}]
    \node[node]	(10)	at (-1.5,2.5-1) {$c_{1,1}$};
    \node[node]	(11)	at (-0.5,2.5-1) {$c_{1,2}$};
    \node[node]	(12)	at (0.5,2.5-1) {$c_{1,3}$};
    \node[node]	(13)	at (1.5,2.5-1) {$c_{1,4}$};
    \node[node,text width=3.5cm]	(14)	at (0,0.5+2.5) {Switch $w_1$};
    
    \node[node]	(0)	at (-1.5,1-2.5) {$c_{2,1}$};
    \node[node]	(1)	at (-0.5,1-2.5) {$c_{2,2}$};
    \node[node]	(2)	at (0.5,1-2.5) {$c_{2,3}$};
    \node[node]	(3)	at (1.5,1-2.5) {$c_{2,4}$};
    \node[node,text width=3.5cm]	(4)	at (0,-0.5-2.5) {Switch $w_2$};
    
    \path[latex-,anchor=east,line width=1pt] (0.south) edge node {$10b$} (0.south|-4.north);
            
    \node[node,text width=4cm]	(20)	at (0,0) {Switch $w_0$};
    
    \path[latex-,anchor=east] (0.north) edge node {$b$} (0.north|-20.south);
    
    \path[-latex,anchor=east] (10.south) edge node {$b$} (10.south|-20.north);
    \path[-latex,anchor=east] (11.south) edge node {$b$} (11.south|-20.north);
    \path[-latex,anchor=east] (12.south) edge node {$b$} (12.south|-20.north);
    \path[-latex,anchor=east] (13.south) edge node {$b$} (13.south|-20.north);

    \draw[line width=0.75pt, dashed] plot [smooth, tension=1.5] coordinates { (-2.75,1.5) (0,0.7) (2.75,1.5) };
    \node[font=\Large] at (2.5, 1.7) {$S^*$};

    \draw[line width=0.75pt, dashed] plot [smooth, tension=3] coordinates { (-2.5,-0.65) (-0.95,-1.5) (-2.5,-2.35) };
    \node[font=\Large] at (-0.8,-0.75) {$S'$};
\end{tikzpicture}
}

%% file: figures/flow2.tex
\scalebox{0.65}{
\begin{tikzpicture}[node/.style={rectangle,draw=black,minimum size=7mm,align=center}]
    \node[node]	(10)	at (-1.5,2.5-1) {$c_{1,1}$};
    \node[node]	(11)	at (-0.5,2.5-1) {$c_{1,2}$};
    \node[node]	(12)	at (0.5,2.5-1) {$c_{1,3}$};
    \node[node]	(13)	at (1.5,2.5-1) {$c_{1,4}$};
    \node[node,text width=3.5cm]	(14)	at (0,0.5+2.5) {Switch $w_1$};
    
    \path[latex-latex,line width=1pt,gray] (10.north) edge (10.north|-14.south);
    \path[latex-latex,line width=1pt,gray] (11.north) edge (11.north|-14.south);
    \path[latex-latex,line width=1pt,gray] (12.north) edge (12.north|-14.south);
    \path[latex-latex,line width=1pt,gray] (13.north) edge (13.north|-14.south);
    
    \node[node]	(0)	at (-1.5,1-2.5) {$c_{2,1}$};
    \node[node]	(1)	at (-0.5,1-2.5) {$c_{2,2}$};
    \node[node]	(2)	at (0.5,1-2.5) {$c_{2,3}$};
    \node[node]	(3)	at (1.5,1-2.5) {$c_{2,4}$};
    \node[node,text width=3.5cm]	(4)	at (0,-0.5-2.5) {Switch $w_2$};
    
    \path[latex-latex,line width=1pt,gray] (0.south) edge (0.south|-4.north);
    \path[latex-latex,line width=1pt,gray] (1.south) edge (1.south|-4.north);
    \path[latex-latex,line width=1pt,gray] (2.south) edge (2.south|-4.north);
    \path[latex-latex,line width=1pt,gray] (3.south) edge (3.south|-4.north);
            
    \node[node,text width=4cm]	(20)	at (0,0) {Switch $w_0$};
    
    \path[latex-latex,gray] (0.north) edge (0.north|-20.south);
    \path[latex-latex,gray] (1.north) edge (1.north|-20.south);
    \path[latex-latex,gray] (2.north) edge (2.north|-20.south);
    \path[latex-latex,gray] (3.north) edge (3.north|-20.south);
    
    \path[latex-latex,gray] (10.south) edge (10.south|-20.north);
    \path[latex-latex,gray] (11.south) edge (11.south|-20.north);
    \path[latex-latex,gray] (12.south) edge (12.south|-20.north);
    \path[latex-latex,gray] (13.south) edge (13.south|-20.north);

    \draw[line width=0.75pt, dashed] plot [smooth, tension=1.5] coordinates { (-2.75,1.5) (0,0.7) (2.75,1.5) };

    \node[circle,line width=1pt,draw=blue] (00)    at (2.5, 2.25) {\textcolor{blue}{$s$}};
    \draw[-latex,line width=1pt,blue] (00) to[bend right=28] node[above,pos=0.3]{$x$} (10);
    \draw[-latex,line width=1pt,blue] (00) to[bend right=25] (11);
    \draw[-latex,line width=1pt,blue] (00) to[bend right=25] (12);
    \draw[-latex,line width=1pt,blue] (00) to[bend right=25] (13);

    \draw[-latex,line width=1pt,blue] (00) to[bend left=35] (0);
    \draw[-latex,line width=1pt,blue] (00) to[bend left=35] (1);
    \draw[-latex,line width=1pt,blue] (00) to[bend left=35] (2);
    \draw[-latex,line width=1pt,blue] (00) to[bend left=25] node[right,pos=0.8]{$x$} (3);
\end{tikzpicture}
}

%% file: edge_splitting.tex
\begin{figure*}[t]
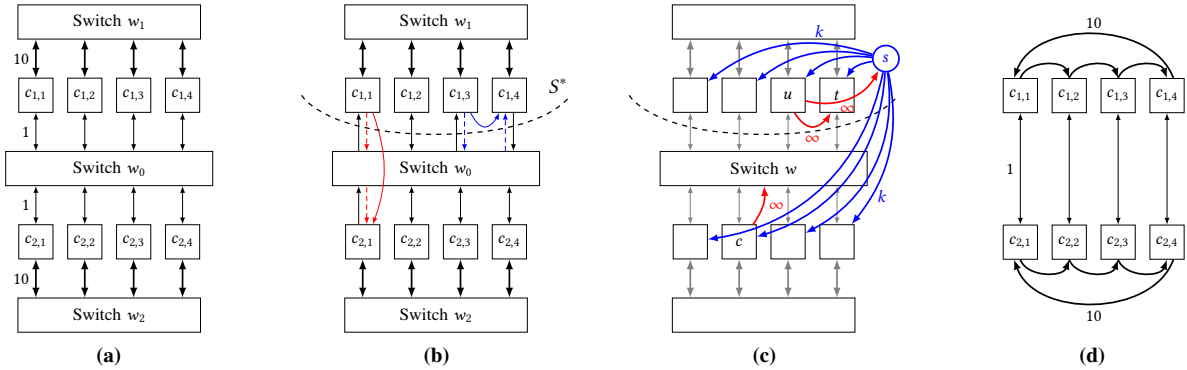

    \centering
    \begin{subfigure}{0.24\textwidth}
        \centering
        \include{figures/topo2}
        \caption{}
    \end{subfigure}
    \begin{subfigure}{0.24\textwidth}
        \centering
        \include{figures/topo_split_example}
        \caption{}
    \end{subfigure}
    \begin{subfigure}{0.24\textwidth}
        \centering
        \include{figures/topo_split_auxiliary}
        \caption{}
    \end{subfigure}
    \begin{subfigure}{0.24\textwidth}
        \centering
        \include{figures/topo_split3}
        \caption{}
    \end{subfigure}
    \caption{Figures explaining the switch node removal process. \textnormal{(a) is the starting topology after optimality binary search scales edge capacities from $\{b,10b\}$ to $\{1,10\}$. (b) contains examples of replacing the ingress and egress capacities of a switch node with a direct capacity bypassing the switch. (c)~shows an example of the auxiliary network ForestColl uses to compute the $\gamma$ in Algorithm~\ref{algo:switchremove}. The $\infty$ edges are a maxflow trick to ensure $\{u,t,s\}$ and $\{w,c\}$ are on opposite sides of the min cut, as we only want to consider cuts that cut through both $(u,w),(w,t)$. (d) is the final resulting switch-free logical topology.}}
    \label{fig:splitting}
\end{figure*}

\subsection{Switch Node Removal}\label{sec:switchremove}

We introduce ForestColl's process to iteratively replace all switch nodes with direct logical links between their neighboring nodes. This allows us to subsequently apply the spanning tree packing algorithm on the resulting switch-free logical topology. The process ensures two key outcomes: (i) \textit{Equivalence:} The spanning trees generated in the logical topology can be mapped back to the original without violating capacity constraints; (ii) \textit{Optimality:} The logical topology retains the same optimal throughput~(\ref{eq:optimality}). Thus, mapping the optimal trees generated on the logical topology back to the original yields the optimal trees for the switch topology. Although TACCL~\cite{taccl} and TACOS~\cite{tacos} also proposed transforming a switch topology into a switch-free logical one, they replace each switch with fixed, preset connection patterns that guarantee only (i) but not (ii), leading to performance loss.

\textbf{Edge Splitting:} Originally a graph theory technique for proving connectivity properties~\cite{bang-jensen,frank,jackson}, we adapt edge splitting to handle switch topologies in the context of collective communications. Starting with the scaled topology as in Figure~\ref{fig:splitting}(a), for each switch node $w$, we pair one capacity of an egress link $(w,t)$ with one capacity of an ingress link $(u,w)$, and replace them with one capacity of a direct link $(u,t)$ that bypasses the switch node $w$. Figure~\ref{fig:splitting}(b) shows two such examples. In both the red and blue examples, we replace the dashed ingress and egress capacities of switch node $w_0$ with a direct unit of capacity bypassing $w_0$. By continuously doing so, we can eliminate all capacities to/from the switch node $w_0$, which is guaranteed by the assumption of equal ingress and egress bandwidth. Once isolated, $w_0$ can be safely removed from the topology. Note that $u,t$ do not have to be compute nodes; they can also be other switch nodes that have not yet been removed. By applying this process to each switch node, we derive a switch-free topology, as shown in Figure~\ref{fig:splitting}(d). The resulting logical topology guarantees \emph{equivalence} to the original, since \textit{we only allocate the capacities of switches to logical connections between compute nodes.}

\textbf{Choose Ingress Link:} Given an egress capacity, there are often multiple ingress links that we can pair and replace. However, arbitrarily choosing an ingress link may lead to performance sacrifice. In the two examples of Figure~\ref{fig:splitting}(b), the exiting capacity of $S^*$ remains unchanged in the red example but decreases from $4$ to $3$ in the blue example. This corresponds to decreasing the exiting bandwidth $B^+(S^*)$ from $4b$ to $3b$ in the original topology. Since $S^*$ is a throughput bottleneck cut, any decrease in its bandwidth $B^+(S^*)$ further bottlenecks the overall performance. \textit{Therefore, when replacing capacities, we must ensure that we do not create a bottleneck cut worse than the existing ones.} For an already bottlenecked cut, any decrease in exiting bandwidth is unacceptable. For a non-bottleneck cut, we can only reduce its exiting bandwidth to the point where its ratio of compute nodes to exiting bandwidth just becomes a bottleneck.

From the two examples in Figure~\ref{fig:splitting}(b), we observe that replacing capacities decreases the exiting capacities of cuts that cut through both the ingress and egress links. Consider replacing a certain capacity of $(u,w),(w,t)$ with $(u,t)$. If we compute among all cuts that cut through both $(u,w),(w,t)$, the minimum decrease $\gamma$ in exiting capacity that would turn any cut into a bottleneck, then by replacing at most this amount of capacity, we are safe from creating a worse bottleneck cut. Figure~\ref{fig:splitting}(c) shows the auxiliary network we use to compute $\gamma$. Similar to the optimality binary search, we compute maxflow with respect to each compute node and take the minimum. We leave the details to compute $\gamma$ in Theorem~\ref{app:thm:multicapasplit} in Appendix~\ref{app:sec:edgesplit}.

{\small
\begin{algorithm}[tb]
    \small
    \SetInd{0.4em}{0.5em}
    \caption{Switch Node Removal}
    \label{algo:switchremove}
    \SetAlgoLined
    \linespread{0.9}\selectfont
    \DontPrintSemicolon
    \KwIn{A directed graph $G=(V_s\cup V_c,E)$ and $k$.}
    \KwOut{A directed graph $H=(V_c,E')$.}
    \Begin{
        \ForEach{\normalfont switch node $w\in V_s$}{
            \ForEach{\normalfont egress edge $f=(w,t)\in E$}{
                \ForEach{\normalfont ingress edge $e=(u,w)\in E$}{
                    Compute $\gamma$, the maximum capacity we can safely replace $f,e$ by $(u,t)$, as in Theorem~\ref{app:thm:multicapasplit}.\;
                    \lIf{$\gamma=0$}{\textbf{continue}}
                    Decrease $f$'s and $e$'s capacity by $\gamma$. Remove $e$ if its capacity reaches 0.\;
                    Increase the capacity of $(u,t)$ by $\gamma$. Add the edge if $(u,t)\!\notin\! E$.\;
                    \lIf{\normalfont $f$'s capacity reaches 0}{\textbf{break}}
                }
                \tcp{\normalfont\textit{Edge $f$ should have 0 capacity at this point.}}
                Remove edge $f$ from $G$.\;
            }
            \tcp{\normalfont\textit{Node $w$ should be isolated at this point.}}
            Remove node $w$ from $G$.\;
        }
        \Return{\normalfont the resulting $G$ as $H$}
    }
\end{algorithm}
}

Algorithm~\ref{algo:switchremove} shows the pseudocode of the switch node removal process. For each switch node $w$ and each egress edge $f\!=\!(w,t)$, we pair it with each ingress edge $e\!=\!(u,w)$ and calculate $\gamma$, the maximum capacity we can safely replace $e,f$ by $(u,t)$. The process iteratively removes switch edges and nodes.
Once all switch nodes are removed, we obtain a switch-free logical topology $H$ that is equivalent to the original $G$ and has the same optimal throughput~(\ref{eq:optimality}).
Appendix~\ref{app:sec:edgesplit} provides more details of the algorithm.

%% file: figures/topo2.tex
\scalebox{0.65}{
\begin{tikzpicture}[node/.style={rectangle,draw=black,minimum size=7mm,align=center}]
    \node[node]	(10)	at (-1.5,2.5-1) {$c_{1,1}$};
    \node[node]	(11)	at (-0.5,2.5-1) {$c_{1,2}$};
    \node[node]	(12)	at (0.5,2.5-1) {$c_{1,3}$};
    \node[node]	(13)	at (1.5,2.5-1) {$c_{1,4}$};
    \node[node,text width=3.5cm]	(14)	at (0,0.5+2.5) {Switch $w_1$};
    
    \path[latex-latex,anchor=east,line width=1pt] (10.north) edge node {$10$} (10.north|-14.south);
    \path[latex-latex,line width=1pt] (11.north) edge (11.north|-14.south);
    \path[latex-latex,line width=1pt] (12.north) edge (12.north|-14.south);
    \path[latex-latex,line width=1pt] (13.north) edge (13.north|-14.south);
    
    \node[node]	(0)	at (-1.5,1-2.5) {$c_{2,1}$};
    \node[node]	(1)	at (-0.5,1-2.5) {$c_{2,2}$};
    \node[node]	(2)	at (0.5,1-2.5) {$c_{2,3}$};
    \node[node]	(3)	at (1.5,1-2.5) {$c_{2,4}$};
    \node[node,text width=3.5cm]	(4)	at (0,-0.5-2.5) {Switch $w_2$};
    
    \path[latex-latex,anchor=east,line width=1pt] (0.south) edge node {$10$} (0.south|-4.north);
    \path[latex-latex,line width=1pt] (1.south) edge (1.south|-4.north);
    \path[latex-latex,line width=1pt] (2.south) edge (2.south|-4.north);
    \path[latex-latex,line width=1pt] (3.south) edge (3.south|-4.north);
            
    \node[node,text width=4cm]	(20)	at (0,0) {Switch $w_0$};
    
    \path[latex-latex,anchor=east] (0.north) edge node {$1$} (0.north|-20.south);
    \path[latex-latex] (1.north) edge (1.north|-20.south);
    \path[latex-latex] (2.north) edge (2.north|-20.south);
    \path[latex-latex] (3.north) edge (3.north|-20.south);
    
    \path[latex-latex,anchor=east] (10.south) edge node {$1$} (10.south|-20.north);
    \path[latex-latex] (11.south) edge (11.south|-20.north);
    \path[latex-latex] (12.south) edge (12.south|-20.north);
    \path[latex-latex] (13.south) edge (13.south|-20.north);
\end{tikzpicture}
}

%% file: figures/topo_split_example.tex
\scalebox{0.65}{
\begin{tikzpicture}[node/.style={rectangle,draw=black,minimum size=7mm,align=center}]
    \node[node]	(10)	at (-1.5,2.5-1) {$c_{1,1}$};
    \node[node]	(11)	at (-0.5,2.5-1) {$c_{1,2}$};
    \node[node]	(12)	at (0.5,2.5-1) {$c_{1,3}$};
    \node[node]	(13)	at (1.5,2.5-1) {$c_{1,4}$};
    \node[node,text width=3.5cm]	(14)	at (0,0.5+2.5) {Switch $w_1$};
    
    \path[latex-latex,line width=1pt] (10.north) edge (10.north|-14.south);
    \path[latex-latex,line width=1pt] (11.north) edge (11.north|-14.south);
    \path[latex-latex,line width=1pt] (12.north) edge (12.north|-14.south);
    \path[latex-latex,line width=1pt] (13.north) edge (13.north|-14.south);
    
    \node[node]	(0)	at (-1.5,1-2.5) {$c_{2,1}$};
    \node[node]	(1)	at (-0.5,1-2.5) {$c_{2,2}$};
    \node[node]	(2)	at (0.5,1-2.5) {$c_{2,3}$};
    \node[node]	(3)	at (1.5,1-2.5) {$c_{2,4}$};
    \node[node,text width=3.5cm]	(4)	at (0,-0.5-2.5) {Switch $w_2$};
    
    \path[latex-latex,line width=1pt] (0.south) edge (0.south|-4.north);
    \path[latex-latex,line width=1pt] (1.south) edge (1.south|-4.north);
    \path[latex-latex,line width=1pt] (2.south) edge (2.south|-4.north);
    \path[latex-latex,line width=1pt] (3.south) edge (3.south|-4.north);
            
    \node[node,text width=4cm]	(20)	at (0,0) {Switch $w_0$};
    
    \path[latex-,densely dashed,red] (0.77) edge (0.77|-20.south);
    \path[latex-] (0.103|-20.south) edge (0.103);
    \path[latex-latex] (1.north) edge (1.north|-20.south);
    \path[latex-latex] (2.north) edge (2.north|-20.south);
    \path[latex-latex] (3.north) edge (3.north|-20.south);
    
    \path[-latex,densely dashed,red] (10.-77) edge (10.-77|-20.north);
    \path[-latex] (10.-103|-20.north) edge (10.-103);
    \path[latex-latex] (11.south) edge (11.south|-20.north);
    \path[-latex,densely dashed,blue] (12.-77) edge (12.-77|-20.north);
    \path[-latex] (12.-103|-20.north) edge (12.-103);
    \path[-latex] (13.-77) edge (13.-77|-20.north);
    \path[-latex,densely dashed,blue] (13.-103|-20.north) edge (13.-103);

    \path[latex-,red] (0.60) edge[bend right=20] (10.-60);
    \draw[-latex,color=blue] plot [smooth, tension=1] coordinates { (12.-60) (1,0.85) (13.-120) };

    \draw[line width=0.75pt, dashed] plot [smooth, tension=1.5] coordinates { (-2.75,1.5) (0,0.7) (2.75,1.5) };
    \node[font=\Large] at (2.5, 1.7) {$S^*$};
\end{tikzpicture}
}

%% file: figures/topo_split_auxiliary.tex
\scalebox{0.65}{
\begin{tikzpicture}[node/.style={rectangle,draw=black,minimum size=7mm,align=center}]
    \node[node]	(10)	at (-1.5,2.5-1) {};
    \node[node]	(11)	at (-0.5,2.5-1) {};
    \node[node]	(12)	at (0.5,2.5-1) {$u$};
    \node[node]	(13)	at (1.5,2.5-1) {$t$};
    \node[node,text width=3.5cm]	(14)	at (0,0.5+2.5) {};
    
    \path[latex-latex,line width=1pt,gray] (10.north) edge (10.north|-14.south);
    \path[latex-latex,line width=1pt,gray] (11.north) edge (11.north|-14.south);
    \path[latex-latex,line width=1pt,gray] (12.north) edge (12.north|-14.south);
    \path[latex-latex,line width=1pt,gray] (13.north) edge (13.north|-14.south);
    
    \node[node]	(0)	at (-1.5,1-2.5) {};
    \node[node]	(1)	at (-0.5,1-2.5) {$c$};
    \node[node]	(2)	at (0.5,1-2.5) {};
    \node[node]	(3)	at (1.5,1-2.5) {};
    \node[node,text width=3.5cm]	(4)	at (0,-0.5-2.5) {};
    
    \path[latex-latex,line width=1pt,gray] (0.south) edge (0.south|-4.north);
    \path[latex-latex,line width=1pt,gray] (1.south) edge (1.south|-4.north);
    \path[latex-latex,line width=1pt,gray] (2.south) edge (2.south|-4.north);
    \path[latex-latex,line width=1pt,gray] (3.south) edge (3.south|-4.north);
            
    \node[node,text width=4cm]	(20)	at (0,0) {Switch $w$};
    
    \path[latex-latex,gray] (0.north) edge (0.north|-20.south);
    \path[latex-latex,gray] (1.north) edge (1.north|-20.south);
    \path[latex-latex,gray] (2.north) edge (2.north|-20.south);
    \path[latex-latex,gray] (3.north) edge (3.north|-20.south);
    
    \path[latex-latex,gray] (10.south) edge (10.south|-20.north);
    \path[latex-latex,gray] (11.south) edge (11.south|-20.north);
    \path[latex-latex,gray] (12.south) edge (12.south|-20.north);
    \path[latex-latex,gray] (13.south) edge (13.south|-20.north);

    \draw[line width=0.75pt, dashed] plot [smooth, tension=1.5] coordinates { (-2.75,1.5) (0,0.7) (2.75,1.5) };

    \node[circle,line width=1pt,draw=blue] (00)    at (2.5, 2.25) {\textcolor{blue}{$s$}};
    \draw[-latex,line width=1pt,blue] (00) to[bend right=28] node[above,pos=0.3]{$k$} (10);
    \draw[-latex,line width=1pt,blue] (00) to[bend right=25] (11);
    \draw[-latex,line width=1pt,blue] (00) to[bend right=25] (12);
    \draw[-latex,line width=1pt,blue] (00) to[bend right=25] (13);

    \draw[-latex,line width=1pt,blue] (00) to[bend left=35] (0);
    \draw[-latex,line width=1pt,blue] (00) to[bend left=35] (1);
    \draw[-latex,line width=1pt,blue] (00) to[bend left=35] (2);
    \draw[-latex,line width=1pt,blue] (00) to[bend left=25] node[right,pos=0.8]{$k$} (3);

    \draw[-latex,line width=1pt,color=red,anchor=north] plot [smooth, tension=1] coordinates { (12.-70) (1,0.85) (13.-110) };
    \node at (1,0.6) {\textcolor{red}{$\infty$}};
    \draw[-latex,line width=1pt,red,anchor=north] (12) to[bend right=38] node {$\infty$} (00);
    \draw[-latex,line width=1pt,red,anchor=west] (1) to[bend right=20] node {$\infty$} (20);
\end{tikzpicture}
}

%% file: figures/topo_split3.tex
\scalebox{0.65}{
\begin{tikzpicture}[node/.style={rectangle,draw=black,minimum size=7mm,align=center}]
    \node[node]	(10)	at (-1.5,2.5-1) {$c_{1,1}$};
    \node[node]	(11)	at (-0.5,2.5-1) {$c_{1,2}$};
    \node[node]	(12)	at (0.5,2.5-1) {$c_{1,3}$};
    \node[node]	(13)	at (1.5,2.5-1) {$c_{1,4}$};

    \node[node,text width=3.5cm,draw=none]	(14)	at (0,0.5+2.5) {};

    \draw[-latex,line width=1pt] (10.north) to[out=70, in=110] (11.north);
    \draw[-latex,line width=1pt] (11.north) to[out=70, in=110] (12.north);
    \draw[-latex,line width=1pt] (12.north) to[out=70, in=110] (13.north);
    \draw[-latex, anchor=south,line width=1pt] (13.70) to[out=110, in=70] node {$10$} (10.110);
    
    \node[node]	(0)	at (-1.5,1-2.5) {$c_{2,1}$};
    \node[node]	(1)	at (-0.5,1-2.5) {$c_{2,2}$};
    \node[node]	(2)	at (0.5,1-2.5) {$c_{2,3}$};
    \node[node]	(3)	at (1.5,1-2.5) {$c_{2,4}$};
    
    \node[node,text width=3.5cm,draw=none]	(4)	at (0,-0.5-2.5) {};
    
    \draw[-latex,line width=1pt] (0.south) to[out=-70, in=-110] (1.south);
    \draw[-latex,line width=1pt] (1.south) to[out=-70, in=-110] (2.south);
    \draw[-latex,line width=1pt] (2.south) to[out=-70, in=-110] (3.south);
    \draw[-latex, anchor=north,line width=1pt] (3.-70) to[out=-110, in=-70] node {$10$} (0.-110);
    
    \path[latex-latex, anchor=east] (0.north) edge node {$1$} (0.north|-10.south);
    \path[latex-latex] (1.north) edge (1.north|-11.south);
    \path[latex-latex] (2.north) edge (2.north|-12.south);
    \path[latex-latex] (3.north) edge (3.north|-13.south);
\end{tikzpicture}
}

%% file: tree_construction.tex
\subsection{Spanning Tree Construction}\label{sec:treeconstruct}

Given the switch-free logical topology like Figure~\ref{fig:splitting}(d), ForestColl applies the spanning out-tree packing algorithm~\cite{berczi2010packing,schrijver}. Our earlier efforts have ensured that in the logical topology, there exist $k$ spanning out-trees rooted at each node, with respect to the scaled link capacities. With all switch nodes removed, every node is now a compute node, and the out-trees simply span all nodes in the topology. Because $k$ can potentially be large, constructing spanning trees one by one may be intractable and not within polynomial time. It turns out that these $k$ out-trees are often not distinct. For example, we may have a batch of $\frac{k}{2}\!-\!1$ identical out-trees and another batch of $\frac{k}{2}\!+\!1$ identical out-trees rooted at the same node. In the algorithm, we construct the out-trees in batches, or rather, trees with capacities. For each node $v$, the algorithm starts by initializing a $k$-capacity out-tree containing only a root node $\{v\}$. Then, it iteratively adds edges to each tree, expanding the tree until it spans all nodes in the graph. When adding an edge to an out-tree, the algorithm calculates the maximum capacity $\mu$ of the edge that can be added to the out-tree while maintaining the feasibility of constructing the remaining trees. If $\mu$ is less than the tree's capacity $m$, the algorithm splits the tree into two: one with capacity $m\!-\!\mu$ and another with capacity $\mu$, adding the edge to the latter. Appendix~\ref{app:sec:constructtree} describes the complete details of the algorithm.

Figure~\ref{fig:treeconstruct}(a) shows one example of the spanning out-trees constructed by applying the algorithm to Figure~\ref{fig:splitting}(d). Thanks to the equivalence guarantee of the logical topology, we can map the out-tree back to the original topology, resulting in the tree shown in Figure~\ref{fig:treeconstruct}(b).

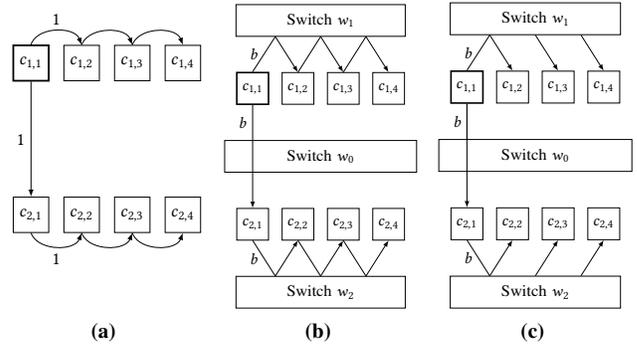
\begin{figure}[tb]
    \vspace{-1em}
    \centering
    \begin{subfigure}{0.327\columnwidth}
        \centering
        \resizebox{1.05\textwidth}{!}{
        \input{figures/topo_split_tree}
        }
        \vspace{1.2em}
        \caption{}
    \end{subfigure}
    \hfill
    \begin{subfigure}{0.327\columnwidth}
        \centering
        \resizebox{1.05\textwidth}{!}{
        \input{figures/topo_tree2}
        }
        \vspace{1.2em}
        \caption{}
    \end{subfigure}
    \hfill
    \begin{subfigure}{0.327\columnwidth}
        \centering
        \resizebox{1.05\textwidth}{!}{
        \input{figures/topo_tree_ina}
        }
        \vspace{1.2em}
        \caption{}
    \end{subfigure}
    \caption{The constructed spanning out-tree. \textnormal{(a) shows one example of the spanning out-trees generated by applying the spanning tree packing algorithm to Figure~\ref{fig:splitting}(d). (b) shows the corresponding tree after mapping (a) back to the original topology. (c) shows the post-processed tree utilizing the in-network multicast/aggregation capabilities of switches $w_1,w_2$.}}
    \label{fig:treeconstruct}
\end{figure}

%% file: figures/topo_split_tree.tex
\scalebox{0.7}{
\begin{tikzpicture}[node/.style={rectangle,draw=black,minimum size=7mm,align=center}]
    \node[node, line width=1pt]	(10)	at (-1.5,2.5-1) {$c_{1,1}$};
    \node[node]	(11)	at (-0.5,2.5-1) {$c_{1,2}$};
    \node[node]	(12)	at (0.5,2.5-1) {$c_{1,3}$};
    \node[node]	(13)	at (1.5,2.5-1) {$c_{1,4}$};

    \node[node,text width=3.5cm,draw=none]	(14)	at (0,0.5+2.5) {};

    \draw[-latex,anchor=south] (10.north) to[out=70, in=110] node {$1$} (11.north);
    \draw[-latex] (11.north) to[out=70, in=110] (12.north);
    \draw[-latex] (12.north) to[out=70, in=110] (13.north);
    \draw[-latex, draw=none] (13.70) to[out=110, in=70] (10.110);
    
    \node[node]	(0)	at (-1.5,1-2.5) {$c_{2,1}$};
    \node[node]	(1)	at (-0.5,1-2.5) {$c_{2,2}$};
    \node[node]	(2)	at (0.5,1-2.5) {$c_{2,3}$};
    \node[node]	(3)	at (1.5,1-2.5) {$c_{2,4}$};

    \node[node,text width=3.5cm,draw=none]	(4)	at (0,-0.5-2.5) {};
    
    \draw[-latex,anchor=north] (0.south) to[out=-70, in=-110] node {$1$} (1.south);
    \draw[-latex] (1.south) to[out=-70, in=-110] (2.south);
    \draw[-latex] (2.south) to[out=-70, in=-110] (3.south);
    \draw[-latex, draw=none] (3.-70) to[out=-110, in=-70] (0.-110);
    
    \path[latex-,anchor=east] (0.north) edge node {$1$} (0.north|-10.south);
\end{tikzpicture}
}

%% file: figures/topo_tree_ina.tex
\scalebox{0.7}{
\begin{tikzpicture}[node/.style={rectangle,draw=black,minimum size=7mm,align=center}]
    \node[node, line width=1pt]	(10)	at (-1.5,2.5-1) {$c_{1,1}$};
    \node[node]	(11)	at (-0.5,2.5-1) {$c_{1,2}$};
    \node[node]	(12)	at (0.5,2.5-1) {$c_{1,3}$};
    \node[node]	(13)	at (1.5,2.5-1) {$c_{1,4}$};
    \node[node,text width=3.5cm]	(14)	at (0,0.5+2.5) {Switch $w_1$};
    
    \draw (10.north) -- node [anchor=east] {$b$} ($(14.south -| -1, 0)$);
    \draw[-latex] ($(14.south -| -1, 0)$) -- (11.north);
    
    \draw[-latex] ($(14.south -| 0, 0)$) -- (12.north);
    
    \draw[-latex] ($(14.south -| 1, 0)$) -- (13.north);

    \node[node]	(0)	at (-1.5,1-2.5) {$c_{2,1}$};
    \node[node]	(1)	at (-0.5,1-2.5) {$c_{2,2}$};
    \node[node]	(2)	at (0.5,1-2.5) {$c_{2,3}$};
    \node[node]	(3)	at (1.5,1-2.5) {$c_{2,4}$};
    \node[node,text width=3.5cm]	(4)	at (0,-0.5-2.5) {Switch $w_2$};
            
    \node[node,text width=4cm]	(20)	at (0,0) {Switch $w_0$};
    
    \path[-latex] (10.south) edge (0.north);
    
    \draw (0.south) -- node [anchor=east] {$b$} ($(4.north -| -1, 0)$);
    \draw[-latex] ($(4.north -| -1, 0)$) -- (1.south);
    
    \draw[-latex] ($(4.north -| 0, 0)$) -- (2.south);
    
    \draw[-latex] ($(4.north -| 1, 0)$) -- (3.south);

    \path[anchor=east,line width=0pt] (10.south) edge node {$b$} (10.south|-20.north);
\end{tikzpicture}
}

%% file: fixed_k.tex
\subsection[Fixed-k Schedule Generation]{Fixed-$k$ Schedule Generation}

In \S\ref{sec:binarysearch}, the optimality binary search automatically determines $k$ (the number of trees rooted at each compute node) and $y$ (the bandwidth utilized by each tree) to achieve theoretically throughput-optimal allgather. However, the $k$ required by optimality can sometimes be a large number. Although the time complexity of ForestColl does not depend on $k$, a large $k$ may complicate the implementation of the schedule. To address this, ForestColl provides an option to generate the highest-throughput schedule given any fixed $k$. The method uses a binary search, similar to \S\ref{sec:binarysearch}, to determine the optimality for the fixed $k$, followed by the usual switch node removal and spanning tree construction to create the out-trees. Appendix~\ref{app:sec:fixedk} provides further details on the algorithm.

Fixed-$k$ schedule generation can significantly simplify the schedule when the optimal $k$ is large. A small $k$---much smaller than what is required for exact optimality---can still achieve performance very close to the optimal. For example, in the 2-box AMD MI250 topology, the theoretically optimal algorithmic bandwidth (algbw) of 354GB/s is attained with $k\!=\!83$ trees rooted at each compute node. Yet, with just $k\!=\!2$, we can already achieve 341GB/s theoretical algbw. In practice, if the optimality binary search gives a too large $k$, we opt to scan $k$ values within a much smaller range ($<\!10$) and pick the best $k$ for schedule construction.

%% file: in_network.tex
\subsection{In-Network Multicast \& Aggregation}\label{sec:innetwork}

On the constructed trees, ForestColl applies a post-processing step to utilize the in-network multicast/aggregation of some switches. Counterintuitively, \textit{in-network multicast/aggregation does not affect allgather/reduce-scatter optimality}, as their optimality is determined by the throughput bottleneck cut in \S\ref{sec:throughputoptimality}, which is unaffected by the switches' capabilities. The intuition is that, in allgather, while in-network multicast can save GPUs from repeatedly sending the same data, each GPU still must receive $N\!-\!1$ distinct data shards, making ingress bandwidth the true bottleneck. Nonetheless, in-network multicast/aggregation is effective for offloading work from GPUs to switches and for reducing overall network traffic.

For each constructed tree, we traverse it from the root, removing traffic that becomes redundant due to in-network multicast of switches. Figure~\ref{fig:treeconstruct} gives an example. In Figure~\ref{fig:treeconstruct}(b), starting from the root, when we reach a node ($c_{2,1}$) sending data to a switch ($w_2$) capable of in-network multicast, we check if other nodes in the tree also send data to the same switch ($c_{2,2},c_{2,3}\!\to\! w_2$). As the data being sent is the same throughout the tree, such traffic can be deleted, resulting in the tree in Figure~\ref{fig:treeconstruct}(c). The same approach also applies to reduce-scatter using in-network aggregation, with everything in the reversed direction. ForestColl is, thus, fully compatible with switches w/ or w/o in-network multicast/aggregation, maximizing performance in all cases.

%% file: discussion.tex
\subsection{Other Collective Operations}\label{sec:othercollectives}

While introduced in the context of allgather, ForestColl can be easily adapted for other collectives. For reduce-scatter, we reverse the allgather out-trees to create in-trees for aggregation. For allreduce, the in-trees and out-trees can be combined to first aggregate to the roots and then broadcast. However, simply combining reduce-scatter and allgather trees does not guarantee allreduce optimality, since (i) allreduce allows each root to reduce/broadcast variable amount of data, and (ii) congestion between in-trees and out-trees can be further optimized. To compute allreduce optimality, a linear program is detailed in Appendix~\ref{app:sec:allreducelp}. Nevertheless, in practice, directly combining reduce-scatter and allgather trees has been sufficient to achieve optimality in all topologies we have evaluated, and we hypothesize that this holds for any topology with equal bandwidth per compute node. For non-uniform allgather/reduce-scatter, where compute nodes broadcast/reduce varying amounts of data, the link capacities from source node $s$ to compute nodes in the auxiliary networks can be adjusted to accommodate such variations.

%% file: evaluation.tex
\section{Evaluation}\label{sec:evaluation}

We present a comprehensive evaluation of ForestColl. \S\ref{sec:implementation} describes our implementations of ForestColl’s schedules. \S\ref{sec:schedule_expt} compares the performance of various schedules on both AMD and NVIDIA hardware. \S\ref{sec:nvidia_expt} evaluates ForestColl against NCCL~\cite{nccl} on a large-scale GPU cluster. \S\ref{sec:evalfsdp} presents results from LLM training with PyTorch FSDP. Finally, \S\ref{sec:geneval} compares different methods for large-scale schedule generation.

\input{implementation}
\input{schedule_expt}
\input{nvidia_expt}
\input{fsdp}
\input{generation_comp}

%% file: implementation.tex
\subsection{Schedule Implementation}\label{sec:implementation}

Given the generated trees, we adopted two implementations: (i) expressing the schedules in XMLs to be executed by the MSCCL runtime, and (ii) using the MSCCL++ library to implement the trees in customized CUDA kernels. MSCCL~\cite{msccl} is built on top of NCCL, sharing the same communication primitives. It is widely used by schedule generation methods, integrates seamlessly with PyTorch, but suffers from scalability limitations. We use it to compare schedule performance, eliminating any differences due to schedule implementation. MSCCL++~\cite{mscclpp} is a CUDA library that provides send/recv channels over NVLink and IB networks, supporting zero-copy communication and NVLink SHARP. For large-scale experiments, we use it to build our own customized CUDA kernels that scale effectively and deliver the best performance with ForestColl’s schedules.

%% file: schedule_expt.tex
\subsection{Schedule Performance Comparison}\label{sec:schedule_expt}

\textbf{Setup:} We evaluated ForestColl’s schedules against vendor libraries and other schedule generation methods on 2-box AMD MI250 and 2-box NVIDIA DGX A100 systems. To eliminate performance differences due to implementation, we uniformly use MSCCL to execute schedules from both ForestColl and the baselines. Thus, any observed performance difference can be attributed solely to the quality of the schedules.

\textbf{Baselines:} We evaluated ForestColl's schedules against TACCL, Blink, and NCCL/RCCL. Due to a runtime error in TACCL's code, we were only able to generate and compare its allgather schedules. For Blink, which lacks publicly available code, we implemented an optimal single-root spanning tree packing based on its paper. Since Blink does not support switch topology, we applied Blink's tree packing to ForestColl's switch-free logical topology to create the ``Blink+Switch'' baseline. Furthermore, Blink's tree packing is limited to single-root reduce+broadcast for allreduce, and it suggests performing allgather as allreduce without reduction, so we only evaluated Blink's allreduce. Both TACCL and Blink use MSCCL in our experiments. While TACCL's code generates MSCCL schedule XMLs, we used ForestColl's compiler to generate Blink's XMLs, as both are tree-flow schedules. Finally, on AMD hardware, we compared against RCCL~\cite{rccl} (ROCm Collective Communication Library), AMD's library optimized for its GPUs, instead of NCCL. Because TE-CCL's code lacks executable schedules and SyCCL's was released only shortly before submission, we evaluate them on theoretical performance (\S\ref{sec:geneval}).

\input{mi250_comparison}
\input{a100_comparison}

%% file: mi250_comparison.tex
\subsubsection{AMD MI250 Experiments}\label{sec:mi250eval}

\input{amd-fig}

\begin{figure}[tb]
    \centering
    \includegraphics[width=0.98\columnwidth]{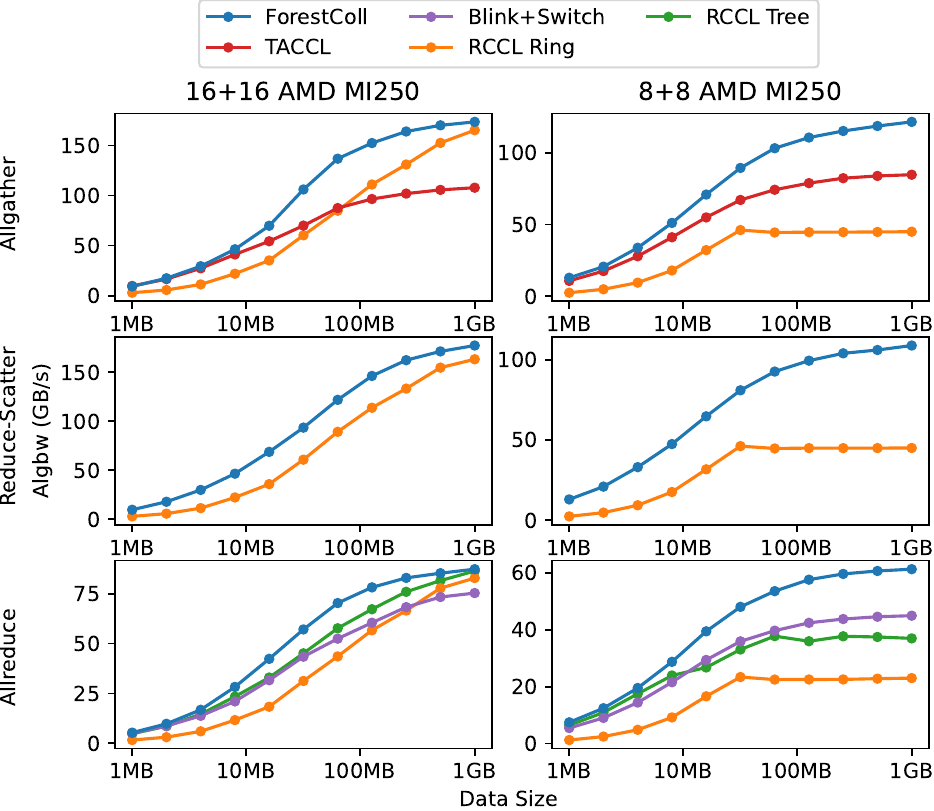}
    \caption{Comparing collective communication performance of TACCL, Blink+Switch, RCCL, and ForestColl in 16+16 and 8+8 settings on 2-box AMD MI250. \textnormal{The columns and rows correspond to different settings and collectives, respectively. ``Blink+Switch'' represents Blink augmented with our switch removal technique, enabling it to support switches.}}
    \label{fig:mi250perf}
\end{figure}

\textbf{Testbed Setup:} The AMD MI250’s complex topology presents significant challenges for schedule generation. Figure~\ref{fig:mi250topo}(a) shows the 2-box topology, characterized by a hybrid of direct intra-box connections and an inter-box switch network.
Each box contains 16 GPUs, each directly connected to three or four other GPUs through $7\times$ AMD Infinity Fabric links at 50GB/s per link. Each box also has $8\times$ IB NICs, offering 256GB/s total inter-box bandwidth and connected to GPUs via PCIe switches. Note that although ForestColl models PCIe switches and IB NICs as switch nodes in schedule generation, for simplicity of illustration, we omit these components and assume each GPU has 16GB/s bandwidth to the IB switch in Figure~\ref{fig:mi250topo}(a).

\textbf{Experiment Setup:} We tested allgather, reduce-scatter, and allreduce performance of ForestColl and the baselines in two settings: one involving all 32 GPUs (16+16) and another with 8 GPUs per box (8+8). In the 8+8 setting, we only enable GPUs $0\!\sim\! 7$ in each box, which corresponds to the left half of Figure~\ref{fig:mi250topo}(a). The 8+8 setting can result from hybrid training parallelism or bin-packing jobs in a cloud environment. Figure~\ref{fig:mi250topo}(b) and (c) show two examples of the trees ForestColl generated for the 16+16 and 8+8, respectively.

\textbf{16+16 Results:} The left column of Figure~\ref{fig:mi250perf} shows our experiment results in 16+16 setting. We compare performance by algorithmic bandwidth (algbw), calculated as data size divided by runtime. ForestColl consistently outperforms baselines. In allgather comparison with TACCL, ForestColl shows a 61\% higher algbw at 1GB data size and an average\footnote{The arithmetic mean percentage improvement across data sizes.} 36\% higher algbw from 1MB to 1GB. Against Blink+Switch in allreduce, ForestColl is 16\% faster at 1GB and 23\% faster on average. In Figure~\ref{fig:mi250perf}, allgather is generally twice as fast as allreduce, contradicting Blink's suggestion to perform allgather as allreduce. RCCL performs comparably to ForestColl at 1GB. However, in allgather and reduce-scatter, RCCL relies solely on the RCCL ring, which is the worst case for hop latency. Thus, ForestColl is much faster at smaller data sizes, outperforming RCCL by 91\% and 87\% on average in allgather and reduce-scatter, respectively. For allreduce, where RCCL tree is available, ForestColl still outperforms RCCL by 15\% on average.

\textbf{8+8 Results:} In 8+8 setting, ForestColl also outperforms the baselines. The comparison between ForestColl and TACCL remains similar, with ForestColl being 43\% faster at 1GB and 32\% faster on average from 1MB to 1GB. Against Blink+Switch, ForestColl is 36\% faster both at 1GB and on average. RCCL's performance, however, drops significantly in the 8+8 setting, with ForestColl being on average 2.98x, 2.86x, and 1.40x fast in allgather, reduce-scatter, and allreduce, respectively. At 1GB data size, ForestColl has 2.7x, 2.42x, and 1.66x algbws compared to RCCL's best-performing algorithm. RCCL's performance drops because it is hand-tuned for fixed topologies with full 16 GPUs per box. In contrast, ForestColl, TACCL, and Blink+Switch can dynamically generate schedules for the new 8+8 topology and have stable performance.

%% file: amd-fig.tex
\begin{figure*}
    \centering
    \begin{subfigure}{0.35\textwidth}
        \centering
        \include{figures/mi250_topo}
        \caption{MI250 Topology}
    \end{subfigure}
    \begin{subfigure}{0.32\textwidth}
        \centering
        \include{figures/mi250_tree}
        \caption{ForestColl 16+16}
    \end{subfigure}
    \begin{subfigure}{0.32\textwidth}
        \centering
        \include{figures/mi250_half_tree}
        \caption{ForestColl 8+8}
    \end{subfigure}
    \caption{2-box AMD MI250 topology and examples of ForestColl's spanning out-trees in 16+16 and 8+8 settings. \textnormal{PCIe switches and IB NICs are omitted for simplicity. (b) and (c) showcase ForestColl's trees rooted at the bold GPU. The complete schedules have at least one tree rooted at each GPU.}}
    \label{fig:mi250topo}
\end{figure*}

%% file: figures/mi250_topo.tex
\scalebox{0.58}{
\begin{tikzpicture}[node/.style={rectangle,draw=black,minimum size=7mm,align=center}]
\draw (-5,-0.35) rectangle (5,0.35);
\node at (0,0) {InfiniBand Switch};

\node[node] (0) at (-1.4, 4.890000000000001) {\scalebox{0.8}[1]{GPU}};
\node[node] (1) at (-1.4, 3.63) {\scalebox{0.8}[1]{GPU}};
\node[node] (2) at (-1.4, 2.37) {\scalebox{0.8}[1]{GPU}};
\node[node] (3) at (-1.4, 1.1099999999999999) {\scalebox{0.8}[1]{GPU}};
\node[node] (4) at (-2.8, 4.890000000000001) {\scalebox{0.8}[1]{GPU}};
\node[node] (5) at (-2.8, 3.63) {\scalebox{0.8}[1]{GPU}};
\node[node] (6) at (-2.8, 2.37) {\scalebox{0.8}[1]{GPU}};
\node[node] (7) at (-2.8, 1.1099999999999999) {\scalebox{0.8}[1]{GPU}};
\node[node] (8) at (1.4, 4.890000000000001) {\scalebox{0.8}[1]{GPU}};
\node[node] (9) at (1.4, 3.63) {\scalebox{0.8}[1]{GPU}};
\node[node] (10) at (1.4, 2.37) {\scalebox{0.8}[1]{GPU}};
\node[node] (11) at (1.4, 1.1099999999999999) {\scalebox{0.8}[1]{GPU}};
\node[node] (12) at (2.8, 4.890000000000001) {\scalebox{0.8}[1]{GPU}};
\node[node] (13) at (2.8, 3.63) {\scalebox{0.8}[1]{GPU}};
\node[node] (14) at (2.8, 2.37) {\scalebox{0.8}[1]{GPU}};
\node[node] (15) at (2.8, 1.1099999999999999) {\scalebox{0.8}[1]{GPU}};

\path[latex-latex,line width=0.75pt] (0.240) edge (0.240|-1.north);
\path[latex-latex,line width=0.75pt] (0.260) edge (0.260|-1.north);
\path[latex-latex,line width=0.75pt] (0.280) edge (0.280|-1.north);
\path[latex-latex,line width=0.75pt] (0.300) edge (0.300|-1.north);
\path[latex-latex,line width=0.75pt] (2.240) edge (2.240|-3.north);
\path[latex-latex,line width=0.75pt] (2.260) edge (2.260|-3.north);
\path[latex-latex,line width=0.75pt] (2.280) edge (2.280|-3.north);
\path[latex-latex,line width=0.75pt] (2.300) edge (2.300|-3.north);
\path[latex-latex,line width=0.75pt] (4.240) edge (4.240|-5.north);
\path[latex-latex,line width=0.75pt] (4.260) edge (4.260|-5.north);
\path[latex-latex,line width=0.75pt] (4.280) edge (4.280|-5.north);
\path[latex-latex,line width=0.75pt] (4.300) edge (4.300|-5.north);
\path[latex-latex,line width=0.75pt] (6.240) edge (6.240|-7.north);
\path[latex-latex,line width=0.75pt] (6.260) edge (6.260|-7.north);
\path[latex-latex,line width=0.75pt] (6.280) edge (6.280|-7.north);
\path[latex-latex,line width=0.75pt] (6.300) edge (6.300|-7.north);
\path[latex-latex,line width=0.75pt] (8.240) edge (8.240|-9.north);
\path[latex-latex,line width=0.75pt] (8.260) edge (8.260|-9.north);
\path[latex-latex,line width=0.75pt] (8.280) edge (8.280|-9.north);
\path[latex-latex,line width=0.75pt] (8.300) edge (8.300|-9.north);
\path[latex-latex,line width=0.75pt] (10.240) edge (10.240|-11.north);
\path[latex-latex,line width=0.75pt] (10.260) edge (10.260|-11.north);
\path[latex-latex,line width=0.75pt] (10.280) edge (10.280|-11.north);
\path[latex-latex,line width=0.75pt] (10.300) edge (10.300|-11.north);
\path[latex-latex,line width=0.75pt] (12.240) edge (12.240|-13.north);
\path[latex-latex,line width=0.75pt] (12.260) edge (12.260|-13.north);
\path[latex-latex,line width=0.75pt] (12.280) edge (12.280|-13.north);
\path[latex-latex,line width=0.75pt] (12.300) edge (12.300|-13.north);
\path[latex-latex,line width=0.75pt] (14.240) edge (14.240|-15.north);
\path[latex-latex,line width=0.75pt] (14.260) edge (14.260|-15.north);
\path[latex-latex,line width=0.75pt] (14.280) edge (14.280|-15.north);
\path[latex-latex,line width=0.75pt] (14.300) edge (14.300|-15.north);

\path[latex-latex,line width=0.75pt] (4.10) edge (4.10-|0.west);
\path[latex-latex,line width=0.75pt] (4.350) edge (4.350-|0.west);
\path[latex-latex,line width=0.75pt] (7.10) edge (7.10-|3.west);
\path[latex-latex,line width=0.75pt] (7.350) edge (7.350-|3.west);
\path[latex-latex,line width=0.75pt] (8.10) edge (8.10-|12.west);
\path[latex-latex,line width=0.75pt] (8.350) edge (8.350-|12.west);
\path[latex-latex,line width=0.75pt] (11.10) edge (11.10-|15.west);
\path[latex-latex,line width=0.75pt] (11.350) edge (11.350-|15.west);

\path[latex-latex,line width=0.75pt] (5.10) edge (1.170);
\path[latex-latex,line width=0.75pt] (6.10) edge (2.170);
\path[latex-latex,line width=0.75pt] (0.10) edge (8.170);
\path[latex-latex,line width=0.75pt] (1.10) edge (9.170);
\path[latex-latex,line width=0.75pt] (9) edge (2);
\path[latex-latex,line width=0.75pt] (10) edge (1);
\path[latex-latex,line width=0.75pt] (2.10) edge (10.170);
\path[latex-latex,line width=0.75pt] (3.10) edge (11.170);
\path[latex-latex,line width=0.75pt] (9.10) edge (13.170);
\path[latex-latex,line width=0.75pt] (10.10) edge (14.170);
\path[latex-latex,line width=0.75pt] (5) edge (6);
\path[latex-latex,line width=0.75pt] (13) edge (14);

\draw[latex-latex,line width=0.75pt] (4.170) -| (-3.5, 2.43) |- (6.170);
\draw[latex-latex,line width=0.75pt] (5.170) -| (-3.6999999999999997, 2.43) |- (7.170);
\draw[latex-latex,line width=0.75pt] (12.10) -| (3.5, 2.43) |- (14.10);
\draw[latex-latex,line width=0.75pt] (13.10) -| (3.6999999999999997, 2.43) |- (15.10);

\draw[latex-latex,gray] (0.350) -| (-0.1,0.35);
\draw[latex-latex,gray] (1.350) -| (-0.30000000000000004,0.35);
\draw[latex-latex,gray] (2.350) -| (-0.5,0.35);
\draw[latex-latex,gray] (3.350) -| (-0.7000000000000001,0.35);
\draw[latex-latex,gray] (4.190) -| (-4.5,0.35);
\draw[latex-latex,gray] (5.190) -| (-4.3,0.35);
\draw[latex-latex,gray] (6.190) -| (-4.1,0.35);
\draw[latex-latex,gray] (7.190) -| (-3.8999999999999995,0.35);
\draw[latex-latex,gray] (8.190) -| (0.1,0.35);
\draw[latex-latex,gray] (9.190) -| (0.30000000000000004,0.35);
\draw[latex-latex,gray] (10.190) -| (0.5,0.35);
\draw[latex-latex,gray] (11.190) -| (0.7000000000000001,0.35);
\draw[latex-latex,gray] (12.350) -| (4.5,0.35);
\draw[latex-latex,gray] (13.350) -| (4.3,0.35);
\draw[latex-latex,gray] (14.350) -| (4.1,0.35);
\draw[latex-latex,gray] (15.350) -| (3.8999999999999995,0.35);

\node[node] (16) at (-1.4, -1.1099999999999999) {\scalebox{0.8}[1]{GPU}};
\node[node] (17) at (-1.4, -2.37) {\scalebox{0.8}[1]{GPU}};
\node[node] (18) at (-1.4, -3.63) {\scalebox{0.8}[1]{GPU}};
\node[node] (19) at (-1.4, -4.890000000000001) {\scalebox{0.8}[1]{GPU}};
\node[node] (20) at (-2.8, -1.1099999999999999) {\scalebox{0.8}[1]{GPU}};
\node[node] (21) at (-2.8, -2.37) {\scalebox{0.8}[1]{GPU}};
\node[node] (22) at (-2.8, -3.63) {\scalebox{0.8}[1]{GPU}};
\node[node] (23) at (-2.8, -4.890000000000001) {\scalebox{0.8}[1]{GPU}};
\node[node] (24) at (1.4, -1.1099999999999999) {\scalebox{0.8}[1]{GPU}};
\node[node] (25) at (1.4, -2.37) {\scalebox{0.8}[1]{GPU}};
\node[node] (26) at (1.4, -3.63) {\scalebox{0.8}[1]{GPU}};
\node[node] (27) at (1.4, -4.890000000000001) {\scalebox{0.8}[1]{GPU}};
\node[node] (28) at (2.8, -1.1099999999999999) {\scalebox{0.8}[1]{GPU}};
\node[node] (29) at (2.8, -2.37) {\scalebox{0.8}[1]{GPU}};
\node[node] (30) at (2.8, -3.63) {\scalebox{0.8}[1]{GPU}};
\node[node] (31) at (2.8, -4.890000000000001) {\scalebox{0.8}[1]{GPU}};

\path[latex-latex,line width=0.75pt] (16.240) edge (16.240|-17.north);
\path[latex-latex,line width=0.75pt] (16.260) edge (16.260|-17.north);
\path[latex-latex,line width=0.75pt] (16.280) edge (16.280|-17.north);
\path[latex-latex,line width=0.75pt] (16.300) edge (16.300|-17.north);
\path[latex-latex,line width=0.75pt] (18.240) edge (18.240|-19.north);
\path[latex-latex,line width=0.75pt] (18.260) edge (18.260|-19.north);
\path[latex-latex,line width=0.75pt] (18.280) edge (18.280|-19.north);
\path[latex-latex,line width=0.75pt] (18.300) edge (18.300|-19.north);
\path[latex-latex,line width=0.75pt] (20.240) edge (20.240|-21.north);
\path[latex-latex,line width=0.75pt] (20.260) edge (20.260|-21.north);
\path[latex-latex,line width=0.75pt] (20.280) edge (20.280|-21.north);
\path[latex-latex,line width=0.75pt] (20.300) edge (20.300|-21.north);
\path[latex-latex,line width=0.75pt] (22.240) edge (22.240|-23.north);
\path[latex-latex,line width=0.75pt] (22.260) edge (22.260|-23.north);
\path[latex-latex,line width=0.75pt] (22.280) edge (22.280|-23.north);
\path[latex-latex,line width=0.75pt] (22.300) edge (22.300|-23.north);
\path[latex-latex,line width=0.75pt] (24.240) edge (24.240|-25.north);
\path[latex-latex,line width=0.75pt] (24.260) edge (24.260|-25.north);
\path[latex-latex,line width=0.75pt] (24.280) edge (24.280|-25.north);
\path[latex-latex,line width=0.75pt] (24.300) edge (24.300|-25.north);
\path[latex-latex,line width=0.75pt] (26.240) edge (26.240|-27.north);
\path[latex-latex,line width=0.75pt] (26.260) edge (26.260|-27.north);
\path[latex-latex,line width=0.75pt] (26.280) edge (26.280|-27.north);
\path[latex-latex,line width=0.75pt] (26.300) edge (26.300|-27.north);
\path[latex-latex,line width=0.75pt] (28.240) edge (28.240|-29.north);
\path[latex-latex,line width=0.75pt] (28.260) edge (28.260|-29.north);
\path[latex-latex,line width=0.75pt] (28.280) edge (28.280|-29.north);
\path[latex-latex,line width=0.75pt] (28.300) edge (28.300|-29.north);
\path[latex-latex,line width=0.75pt] (30.240) edge (30.240|-31.north);
\path[latex-latex,line width=0.75pt] (30.260) edge (30.260|-31.north);
\path[latex-latex,line width=0.75pt] (30.280) edge (30.280|-31.north);
\path[latex-latex,line width=0.75pt] (30.300) edge (30.300|-31.north);

\path[latex-latex,line width=0.75pt] (20.10) edge (20.10-|16.west);
\path[latex-latex,line width=0.75pt] (20.350) edge (20.350-|16.west);
\path[latex-latex,line width=0.75pt] (23.10) edge (23.10-|19.west);
\path[latex-latex,line width=0.75pt] (23.350) edge (23.350-|19.west);
\path[latex-latex,line width=0.75pt] (24.10) edge (24.10-|28.west);
\path[latex-latex,line width=0.75pt] (24.350) edge (24.350-|28.west);
\path[latex-latex,line width=0.75pt] (27.10) edge (27.10-|31.west);
\path[latex-latex,line width=0.75pt] (27.350) edge (27.350-|31.west);

\path[latex-latex,line width=0.75pt] (21.350) edge (17.190);
\path[latex-latex,line width=0.75pt] (22.350) edge (18.190);
\path[latex-latex,line width=0.75pt] (16.350) edge (24.190);
\path[latex-latex,line width=0.75pt] (17.350) edge (25.190);
\path[latex-latex,line width=0.75pt] (25) edge (18);
\path[latex-latex,line width=0.75pt] (26) edge (17);
\path[latex-latex,line width=0.75pt] (18.350) edge (26.190);
\path[latex-latex,line width=0.75pt] (19.350) edge (27.190);
\path[latex-latex,line width=0.75pt] (25.350) edge (29.190);
\path[latex-latex,line width=0.75pt] (26.350) edge (30.190);
\path[latex-latex,line width=0.75pt] (21) edge (22);
\path[latex-latex,line width=0.75pt] (29) edge (30);

\draw[latex-latex,line width=0.75pt] (20.190) -| (-3.5, -2.43) |- (22.190);
\draw[latex-latex,line width=0.75pt] (21.190) -| (-3.6999999999999997, -2.43) |- (23.190);
\draw[latex-latex,line width=0.75pt] (28.350) -| (3.5, -2.43) |- (30.350);
\draw[latex-latex,line width=0.75pt] (29.350) -| (3.6999999999999997, -2.43) |- (31.350);

\draw[latex-latex,gray] (16.10) -| (-0.7,-0.35);
\draw[latex-latex,gray] (17.10) -| (-0.49999999999999994,-0.35);
\draw[latex-latex,gray] (18.10) -| (-0.29999999999999993,-0.35);
\draw[latex-latex,gray] (19.10) -| (-0.09999999999999987,-0.35);
\draw[latex-latex,gray] (20.170) -| (-3.9,-0.35);
\draw[latex-latex,gray] (21.170) -| (-4.1000000000000005,-0.35);
\draw[latex-latex,gray] (22.170) -| (-4.3,-0.35);
\draw[latex-latex,gray] (23.170) -| (-4.5,-0.35);
\draw[latex-latex,gray] (24.170) -| (0.7,-0.35);
\draw[latex-latex,gray] (25.170) -| (0.49999999999999994,-0.35);
\draw[latex-latex,gray] (26.170) -| (0.29999999999999993,-0.35);
\draw[latex-latex,gray] (27.170) -| (0.09999999999999987,-0.35);
\draw[latex-latex,gray] (28.10) -| (3.9,-0.35);
\draw[latex-latex,gray] (29.10) -| (4.1000000000000005,-0.35);
\draw[latex-latex,gray] (30.10) -| (4.3,-0.35);
\draw[latex-latex,gray] (31.10) -| (4.5,-0.35);

\node[gray, rotate=90] at (-4.7, 2.6) {16GB/s};
\node[rotate=90] at (-3.7, 4.26) {50GB/s};
\end{tikzpicture}
}

%% file: figures/mi250_tree.tex
\scalebox{0.58}{
\begin{tikzpicture}[node/.style={rectangle,draw=black,minimum size=7mm,align=center}]
\draw (-4,-0.35) rectangle (4,0.35);

\node[node,line width=1.5pt] (0) at (-1.4, 4.890000000000001) {\scalebox{0.8}[1]{GPU}};
\node[node] (1) at (-1.4, 3.63) {\scalebox{0.8}[1]{GPU}};
\node[node] (2) at (-1.4, 2.37) {\scalebox{0.8}[1]{GPU}};
\node[node] (3) at (-1.4, 1.1099999999999999) {\scalebox{0.8}[1]{GPU}};
\node[node] (4) at (-2.8, 4.890000000000001) {\scalebox{0.8}[1]{GPU}};
\node[node] (5) at (-2.8, 3.63) {\scalebox{0.8}[1]{GPU}};
\node[node] (6) at (-2.8, 2.37) {\scalebox{0.8}[1]{GPU}};
\node[node] (7) at (-2.8, 1.1099999999999999) {\scalebox{0.8}[1]{GPU}};
\node[node] (8) at (1.4, 4.890000000000001) {\scalebox{0.8}[1]{GPU}};
\node[node] (9) at (1.4, 3.63) {\scalebox{0.8}[1]{GPU}};
\node[node] (10) at (1.4, 2.37) {\scalebox{0.8}[1]{GPU}};
\node[node] (11) at (1.4, 1.1099999999999999) {\scalebox{0.8}[1]{GPU}};
\node[node] (12) at (2.8, 4.890000000000001) {\scalebox{0.8}[1]{GPU}};
\node[node] (13) at (2.8, 3.63) {\scalebox{0.8}[1]{GPU}};
\node[node] (14) at (2.8, 2.37) {\scalebox{0.8}[1]{GPU}};
\node[node] (15) at (2.8, 1.1099999999999999) {\scalebox{0.8}[1]{GPU}};

\path[-latex] (0) edge (1);
\path[-latex] (0) edge (4);
\path[-latex] (0.10) edge (8.170);
\path[-latex] (1) edge (5);
\path[-latex] (1) edge (9);
\path[-latex] (1) edge (10);
\draw[-latex] (4.west) -| (-3.5, 2.43) |- (6.west);
\path[-latex] (8) edge (12);
\draw[-latex] (5.west) -| (-3.6999999999999997, 2.43) |- (7.west);
\path[-latex] (9) edge (2);
\path[-latex] (10) edge (11);
\path[-latex] (10) edge (14);
\path[-latex] (12) edge (13);
\path[-latex] (2) edge (3);
\path[-latex] (14) edge (15);

\node[node] (16) at (-1.4, -1.1099999999999999) {\scalebox{0.8}[1]{GPU}};
\node[node] (17) at (-1.4, -2.37) {\scalebox{0.8}[1]{GPU}};
\node[node] (18) at (-1.4, -3.63) {\scalebox{0.8}[1]{GPU}};
\node[node] (19) at (-1.4, -4.890000000000001) {\scalebox{0.8}[1]{GPU}};
\node[node] (20) at (-2.8, -1.1099999999999999) {\scalebox{0.8}[1]{GPU}};
\node[node] (21) at (-2.8, -2.37) {\scalebox{0.8}[1]{GPU}};
\node[node] (22) at (-2.8, -3.63) {\scalebox{0.8}[1]{GPU}};
\node[node] (23) at (-2.8, -4.890000000000001) {\scalebox{0.8}[1]{GPU}};
\node[node] (24) at (1.4, -1.1099999999999999) {\scalebox{0.8}[1]{GPU}};
\node[node] (25) at (1.4, -2.37) {\scalebox{0.8}[1]{GPU}};
\node[node] (26) at (1.4, -3.63) {\scalebox{0.8}[1]{GPU}};
\node[node] (27) at (1.4, -4.890000000000001) {\scalebox{0.8}[1]{GPU}};
\node[node] (28) at (2.8, -1.1099999999999999) {\scalebox{0.8}[1]{GPU}};
\node[node] (29) at (2.8, -2.37) {\scalebox{0.8}[1]{GPU}};
\node[node] (30) at (2.8, -3.63) {\scalebox{0.8}[1]{GPU}};
\node[node] (31) at (2.8, -4.890000000000001) {\scalebox{0.8}[1]{GPU}};

\path[-latex] (16) edge (17);
\path[-latex] (16) edge (20);
\path[-latex] (16.350) edge (24.190);
\path[-latex] (17) edge (25);
\path[-latex] (17) edge (26);
\path[-latex] (20) edge (21);
\draw[-latex] (20.west) -| (-3.5, -2.43) |- (22.west);
\path[-latex] (24) edge (28);
\path[-latex] (25) edge (18);
\path[-latex] (26) edge (27);
\path[-latex] (26) edge (30);
\path[-latex] (22) edge (23);
\path[-latex] (28) edge (29);
\path[-latex] (18) edge (19);
\path[-latex] (30) edge (31);

\draw[-latex] (0.350) -| (-0.5, 0) |- (16.10);
\end{tikzpicture}
}

%% file: figures/mi250_half_tree.tex
\scalebox{0.58}{
\begin{tikzpicture}[node/.style={rectangle,draw=black,minimum size=7mm,align=center}]
\draw (-4,-0.35) rectangle (4,0.35);

\node[node,line width=1.5pt] (0) at (-1.4, 4.890000000000001) {\scalebox{0.8}[1]{GPU}};
\node[node] (1) at (-1.4, 3.63) {\scalebox{0.8}[1]{GPU}};
\node[node] (2) at (-1.4, 2.37) {\scalebox{0.8}[1]{GPU}};
\node[node] (3) at (-1.4, 1.1099999999999999) {\scalebox{0.8}[1]{GPU}};
\node[node] (4) at (-2.8, 4.890000000000001) {\scalebox{0.8}[1]{GPU}};
\node[node] (5) at (-2.8, 3.63) {\scalebox{0.8}[1]{GPU}};
\node[node] (6) at (-2.8, 2.37) {\scalebox{0.8}[1]{GPU}};
\node[node] (7) at (-2.8, 1.1099999999999999) {\scalebox{0.8}[1]{GPU}};
\node[node,gray] (8) at (1.4, 4.890000000000001) {\scalebox{0.8}[1]{GPU}};
\node[node,gray] (9) at (1.4, 3.63) {\scalebox{0.8}[1]{GPU}};
\node[node,gray] (10) at (1.4, 2.37) {\scalebox{0.8}[1]{GPU}};
\node[node,gray] (11) at (1.4, 1.1099999999999999) {\scalebox{0.8}[1]{GPU}};
\node[node,gray] (12) at (2.8, 4.890000000000001) {\scalebox{0.8}[1]{GPU}};
\node[node,gray] (13) at (2.8, 3.63) {\scalebox{0.8}[1]{GPU}};
\node[node,gray] (14) at (2.8, 2.37) {\scalebox{0.8}[1]{GPU}};
\node[node,gray] (15) at (2.8, 1.1099999999999999) {\scalebox{0.8}[1]{GPU}};

\path[-latex] (0) edge (1);
\path[-latex] (0) edge (4);
\path[-latex] (1) edge (5);
\draw[-latex] (4.west) -| (-3.5, 2.43) |- (6.west);
\draw[-latex] (5.west) -| (-3.6999999999999997, 2.43) |- (7.west);
\path[-latex] (6) edge (2);
\path[-latex] (2) edge (3);

\node[node] (16) at (-1.4, -1.1099999999999999) {\scalebox{0.8}[1]{GPU}};
\node[node] (17) at (-1.4, -2.37) {\scalebox{0.8}[1]{GPU}};
\node[node] (18) at (-1.4, -3.63) {\scalebox{0.8}[1]{GPU}};
\node[node] (19) at (-1.4, -4.890000000000001) {\scalebox{0.8}[1]{GPU}};
\node[node] (20) at (-2.8, -1.1099999999999999) {\scalebox{0.8}[1]{GPU}};
\node[node] (21) at (-2.8, -2.37) {\scalebox{0.8}[1]{GPU}};
\node[node] (22) at (-2.8, -3.63) {\scalebox{0.8}[1]{GPU}};
\node[node] (23) at (-2.8, -4.890000000000001) {\scalebox{0.8}[1]{GPU}};
\node[node,gray] (24) at (1.4, -1.1099999999999999) {\scalebox{0.8}[1]{GPU}};
\node[node,gray] (25) at (1.4, -2.37) {\scalebox{0.8}[1]{GPU}};
\node[node,gray] (26) at (1.4, -3.63) {\scalebox{0.8}[1]{GPU}};
\node[node,gray] (27) at (1.4, -4.890000000000001) {\scalebox{0.8}[1]{GPU}};
\node[node,gray] (28) at (2.8, -1.1099999999999999) {\scalebox{0.8}[1]{GPU}};
\node[node,gray] (29) at (2.8, -2.37) {\scalebox{0.8}[1]{GPU}};
\node[node,gray] (30) at (2.8, -3.63) {\scalebox{0.8}[1]{GPU}};
\node[node,gray] (31) at (2.8, -4.890000000000001) {\scalebox{0.8}[1]{GPU}};

\path[-latex] (16) edge (17);
\path[-latex] (16) edge (20);
\path[-latex] (17) edge (21);
\draw[-latex] (20.west) -| (-3.5, -2.43) |- (22.west);
\draw[-latex] (21.west) -| (-3.6999999999999997, -2.43) |- (23.west);
\path[-latex] (22) edge (18);
\path[-latex] (18) edge (19);

\draw[-latex] (0.east) -| (-0.5, 0) |- (16.east);
\end{tikzpicture}
}

%% file: a100_comparison.tex
\subsubsection{NVIDIA DGX A100 Experiments}\label{sec:a100eval}

\textbf{Testbed Setup:} On a 2-box NVIDIA DGX A100 testbed, we evaluated ForestColl's schedules against TACCL and NCCL.
Each box has $8\times$ NVIDIA A100 GPUs, interconnected by an NVSwitch with 300GB/s intra-box bandwidth per GPU. Additionally, every two GPUs are connected to two IB NICs via a PCIe switch. Each NIC offers 25GB/s inter-box bandwidth.

\begin{figure}[tb]
    \centering
    \includegraphics[width=0.98\columnwidth]{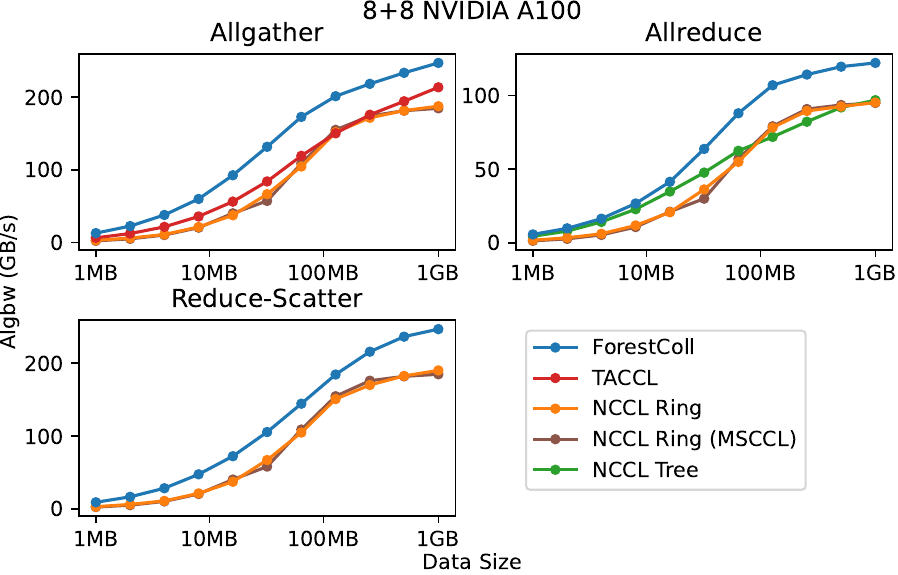}
    \caption{Comparing collective communication performance of TACCL, NCCL, and ForestColl on 2-box NVIDIA DGX A100. \textnormal{The ``NCCL Ring (MSCCL)'' results are obtained by implementing the NCCL ring in MSCCL XMLs to confirm no inherent performance difference between NCCL and MSCCL.}}
    \label{fig:a100perf}
\end{figure}

\textbf{Results:} Figure~\ref{fig:a100perf} presents our experiment results. ForestColl leads in all three collectives by a considerable margin over the closest baseline. While TACCL's performance improves in a switch-only topology, ForestColl still outperforms it by 16\% at 1GB and by an average of 53\% for sizes from 1MB to 1GB. The improvement of ForestColl over NCCL is even more pronounced. At 1GB, ForestColl achieves 32\%, 30\%, and 26\% higher algbws for allgather, reduce-scatter, and allreduce, respectively. Averaged across 1MB$\sim$1GB, ForestColl is 130\%, 85\%, and 27\% faster than NCCL.

In addition to comparing NCCL and ForestColl, we implemented the NCCL ring as MSCCL XML and tested its performance. In Figure~\ref{fig:a100perf}, the NCCL ring in MSCCL shows identical performance to the default NCCL ring in all collectives, showing that \ul{ForestColl's improvements stem solely from scheduling optimization, not runtime tuning or inherent performance difference between NCCL and MSCCL.}

%% file: nvidia_expt.tex
\begin{figure*}[t]
    \centering
    \begin{subfigure}{0.376\textwidth}
        \centering
        \includegraphics[width=\textwidth]{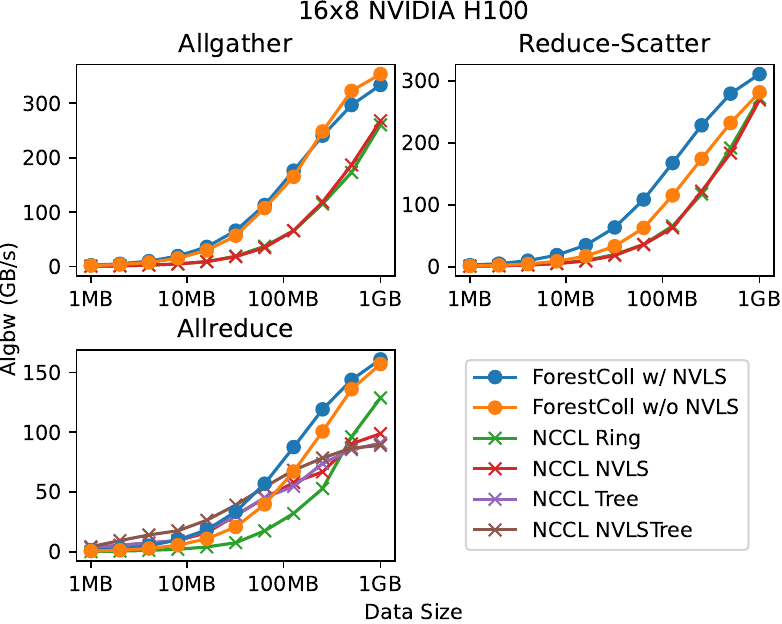}
        \vspace{1em}
        \caption{Allgather, reduce-scatter, allreduce algbws at 16x8 H100}
    \end{subfigure}
    \hfill
    \begin{subfigure}{0.604\textwidth}
        \centering
        \includegraphics[width=\textwidth]{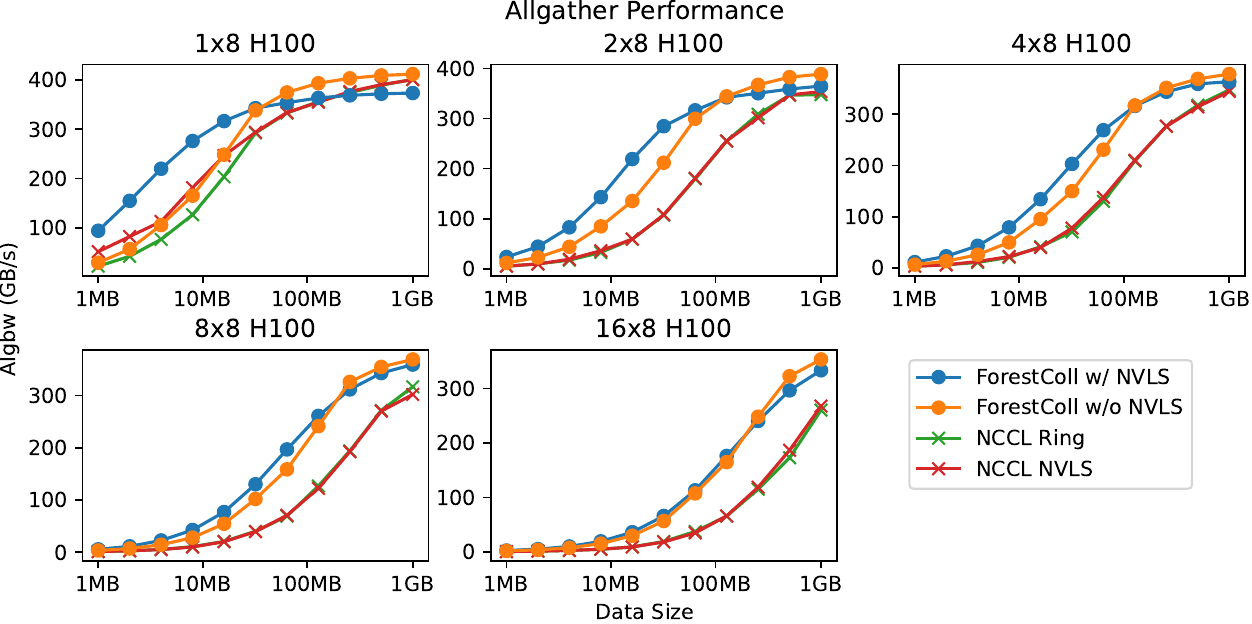}
        \vspace{1em}
        \caption{Allgather algbws at $\{1,2,4,8,16\}$x8 H100}
    \end{subfigure}
    \caption{Comparing NCCL and ForestColl collective communication performance on NVIDIA DGX H100.}
    \label{fig:h100results}
\end{figure*}

\subsection{Large-Scale GPU Cluster Experiments}\label{sec:nvidia_expt}

\textbf{Setup}: We evaluated ForestColl against NCCL on a testbed of $16\times$ DGX H100 boxes ($128\times$ H100 GPUs). Each DGX box is equipped with an NVSwitch, providing 450GB/s of intra-box bandwidth per GPU, and $8\times$ IB NICs, each offering 50GB/s of inter-box bandwidth. Due to scalability limitations of MSCCL---specifically, each SM can send/recv data from only one peer---we implemented customized CUDA kernels using MSCCL++ based on ForestColl's generated trees.

\textbf{16x8 Results:} Figure~\ref{fig:h100results}(a) shows allgather, reduce-scatter and allreduce performance on $16\times$ DGX H100 boxes ($128\times$ GPUs). ForestColl achieves substantially higher throughput than NCCL in all three collectives, benefiting from both more efficient scheduling and optimized implementation. For allgather and reduce-scatter, ForestColl delivers 32\% and 14\% higher throughput at 1GB data size. In allreduce, while NCCL's tree algorithms perform better at smaller, latency-sensitive data sizes, ForestColl still dominates at large data sizes, achieving 25\% higher throughput at 1GB. In production, existing runtime systems can seamlessly switch between low-latency schedules and ForestColl depending on the input data size, allowing the two to complement each other.

\textbf{1-16x8 Results:} Figure~\ref{fig:h100results}(b) presents allgather performance from 1 to $16\times$ DGX H100 boxes. At the 1x8 scale, where communication is intra-box only, ForestColl and NCCL have similar throughput, with ForestColl NVLS performing better at smaller data sizes due to the lower latency of zero-copy implementation. At larger scales, where inter-box bandwidth becomes the bottleneck, ForestColl's schedules have less cross-box traffic (as explained in Figure~\ref{fig:ringexample}) and outperform NCCL by larger margins. The results show that \ul{collective communication implementations guided by ForestColl's scheduling can achieve higher throughput than state-of-the-art communication libraries.}

%% file: fsdp.tex
\begin{figure}[tb]
    \centering
    \includegraphics[width=0.98\columnwidth]{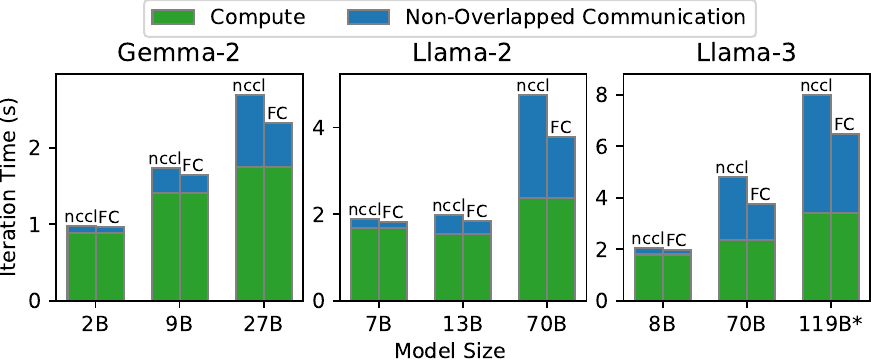}
    \caption{Comparing NCCL and ForestColl in Fully Sharded Data Parallel (FSDP) training. \textnormal{The training is on 2x DGX A100 with 16 GPUs using PyTorch FSDP~\cite{pytorchfsdp}. The compute times are measured by training with communications skipped. Given the limited scale of our testbed, context lengths are set to 2048 for Gemma and 1024 for Llama models, with batch sizes set to the maximum allowed by GPU memory (80GB per GPU). Models are from Hugging Face~\cite{hftransformers} and use FlashAttention~\cite{flashattention,flashattention2} with BFloat16.}}
    \label{fig:fsdptraining}
\end{figure}

\subsection{FSDP Training Experiments}\label{sec:evalfsdp}

To show that the communication speedup provided by ForestColl accelerates LLM training, we run Fully Sharded Data Parallel (FSDP) training~\cite{pytorchfsdp,zero} with open-source LLMs: Gemma-2~\cite{gemma2} from Google and Llama-2~\cite{touvron2023llama} \& 3~\cite{llama3} from Meta. FSDP is widely used for training large models that far exceed the memory capacity of a single GPU~\cite{touvron2023llama,llama3,olmo}. It shards model parameters across GPUs and allgathers them as needed. In LLM training, FSDP typically allgathers the weights at each layer, performs the computation, and discards the weights to free up memory for the next layer in the forward and backward passes. A reduce-scatter is also needed in the backward pass to aggregate the gradients of each layer.

With MSCCL's seamless integration with PyTorch, we use the same setup as in \S\ref{sec:a100eval}. Figure~\ref{fig:fsdptraining} shows our training results, comparing iteration times (forward+backward) using NCCL vs ForestColl. The iteration times are broken down into compute (comp) time\footnote{The comp time is measured by skipping comm operations in an iteration.} and communication (comm) time not overlapped by comp. For smaller models, such as Gemma-2-2b, Llama-2-7b, and Llama-3-8b, the improvements with ForestColl are minimal, showing reductions in iteration time of less than 5\%. With comp accounting for over 88\% of the iteration time, these small models are comp-bound, with speedup in comm having little effect on overall performance. However, as model size increases, the trend shifts toward becoming comm-bound. For Gemma-2-27b, Llama-2-70b, and Llama-3-119b\footnote{Due to the limited scale of our testbed, we reduce num\_hidden\_layers in Llama-3-405B to 36, creating the 119B model for our experiments.}, comp accounts for only 65\%, 50\%, and 43\% of the iteration times. As a result, compared to NCCL, ForestColl reduces iteration times by 14\% for Gemma-2-27b and 20\% for both Llama-2-70b and Llama-3-119b.

Large models are more comm-bound for two reasons. First, large models cannot be trained with large batch sizes due to higher GPU memory usage. In our experiments, while a small model like Llama-3-8b can be trained with a batch size of 8, Llama-3-70b is limited to a batch size of 1 to avoid GPU out-of-memory, even with memory-efficient techniques like FlashAttention~\cite{flashattention,flashattention2} and BF16 parameters. Second, large models have poor comp-comm overlap due to contention between comp and comm kernels for GPU resources. For example, the comp kernel in FlashAttention uses a number of Streaming Multiprocessors (SM) proportional to the number of attention heads, while comm kernel requires more SMs to saturate bandwidth for large data transfers. For large models, the combined demands of comp and comm kernels exceed a GPU's total SMs, forcing sequential execution.

%% file: generation_comp.tex
\subsection{Large-Scale Schedule Generation Comparison}\label{sec:geneval}

In large-scale schedule generation, we compare ForestColl against MultiTree, TACCL, TE-CCL, and SyCCL. While Blink, TACOS, and BFB also conduct schedule generation, Blink and BFB lack support for switch topologies, and the released TACOS implementation does not support switches (as of submission). In Figure~\ref{fig:taccl_gen}, we compare MultiTree, TACCL, TE-CCL, and ForestColl in generating allgather schedules for NVIDIA A100 and AMD MI250 topologies, with SyCCL included for A100 only due to failures in schedule generation for MI250. Appendix~\ref{sec:gen_opt} details our implementation and parallelization of ForestColl's scheduling algorithm.

\begin{figure}[tb]
    \centering
    \includegraphics[width=0.98\columnwidth]{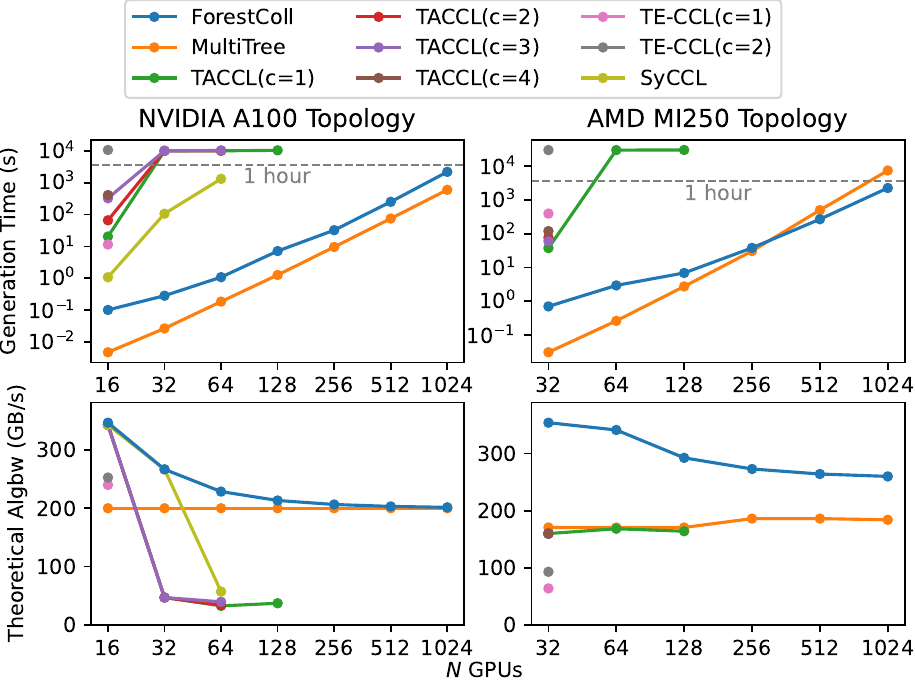}
    \caption{Large-scale schedule generation comparison between MultiTree, TACCL, TE-CCL, SyCCL, and ForestColl on NVIDIA A100 and AMD MI250 topologies. \textnormal{The top row compares the time spent on generation, and the bottom row compares the theoretical algorithmic bandwidth of the generated schedules. TACCL and TE-CCL are run with varying numbers of chunks. The time limit is set to $10^4$s for A100 and $3\!\times\! 10^4$s for MI250.}}
    \label{fig:taccl_gen}
\end{figure}

\textbf{Setup:} TACCL, TE-CCL, and SyCCL use mixed integer linear programming (MILP) to generate schedules. Since solving MILP to optimality is NP-hard and often extremely time-consuming, these methods support setting a time limit to stop early and return the best solution found so far. However, for large topologies, the solver may not find any solution within the time limit. We set a $10^4$s time limit for A100 topologies and $3\!\times\! 10^4$s (8.3 hours) for the more complicated MI250. MultiTree briefly mentions handling heterogeneity by creating multiedges with unit bandwidth but does not specify how to determine the unit bandwidth. If too small, all trees could be scheduled to use the same congested link. Here, we set the unit bandwidth to the bandwidth of the slowest link.

\textbf{Generation Runtime:} In Figure~\ref{fig:taccl_gen}, ForestColl is orders of magnitude faster than TACCL, TE-CCL, and SyCCL in schedule generation. On A100, while TACCL hits the time limit at 4 boxes (32 GPUs), ForestColl generates a schedule in under a second. For larger topologies, TACCL fails to generate any solution within the time limits for $>\!128$ GPUs, and TE-CCL cannot generate a schedule beyond two A100 or MI250 boxes. We also tried TE-CCL's $A^*$ technique, but it also failed to scale further. SyCCL runs faster by exploiting topology symmetry, but still fails to scale beyond 64 GPUs. Its paper reports scaling to 512 GPUs (37min) with heuristic tuning, yet ForestColl (4min) is still order of magnitude faster at the same scale. Compared to MultiTree, ForestColl is slower on the A100 topology but faster on the more complex MI250. Despite MultiTree's use of a much simpler greedy algorithm, ForestColl still achieves a similar scalability curve. Notably, ForestColl is the \textit{only} method able to generate both 1024-GPU schedules within 1 hour (A100: 37min, MI250: 38min). Although not generated in seconds, these runtimes are far more practical than MILP approaches and acceptable given that schedules are only precomputed once per topology.

\textbf{Theoretical Schedule Throughput:} ForestColl is always theoretically optimal, outperforming all other methods in Figure~\ref{fig:taccl_gen}. On A100, TACCL and SyCCL initially match ForestColl but their throughput drops significantly at larger scales. MultiTree starts with considerably lower throughput than ForestColl but asymptotically matches it as topology size scales, likely due to the simplicity of A100 topologies. On the more complex MI250, ForestColl outperforms MultiTree by 50\%+.
Meanwhile, TE-CCL lags behind all other methods.

Due to the scalability limitations of MILP, TACCL, TE-CCL, and SyCCL rely on extensive heuristic tuning with tens of parameters in their configs/sketches. Despite our tuning efforts, we still observed instability in scalability and schedule throughput. In contrast, ForestColl needs only the input topology as a capacitated graph to ensure both scalability and optimality. \ul{ForestColl makes finding optimal schedules for large-scale topologies mathematically feasible and achievable within tractable runtime bounds.}

%% file: conclusion.tex
\section{Concluding Remarks}

Collective communication has become a performance bottleneck in distributed ML. The heterogeneity and diversity of the network topologies pose significant challenges to designing efficient communication algorithms. In this paper, we proposed ForestColl, which \emph{efficiently} generates \emph{throughput-optimal} schedules for \emph{any type} of network topology. Experiments on popular ML hardware platforms have demonstrated our significant performance improvements over both the platforms' own optimized communication libraries and other state-of-the-art schedule generation techniques.

%% file: appendix/appendix_main.tex
\tolerance=3000 

\clearpage

\noindent {\Large \textbf{Appendix}}

\vspace{0.1in}

\noindent In this appendix, we provide detailed mathematical analysis and proofs to supplement the main text. To summarize,
\squishlist
    \item \S\ref{app:sec:notations} provides a summary of notations used in the main text and appendix.
    \item \S\ref{sec:otherrelated} discusses related works beyond schedule generation.
    \item \S\ref{sec:gen_opt} describes the implementation and parallelization of schedule generation algorithm.
    \item \S\ref{app:sec:dilemma} shows a dilemma for step schedules to reach optimality.
    \item \S\ref{app:sec:algodesign} elaborates on the mathematical details of algorithm.
    \item \S\ref{app:sec:runtime} proves that every part of our algorithm runs in polynomial time.
    \item \S\ref{app:sec:allreducelp} describes a linear program to construct optimal allreduce schedule.
    \item \S\ref{app:sec:proofs} provides proofs of all theorems in this paper.
\squishend

\section{Notations}\label{app:sec:notations}

To ensure rigorous mathematical reasoning, we introduce the following notations:
\squishlist
    \item $G\!=\!(V\!=\!V_s\cup V_c,E)$: the input topology as a directed graph. $V_s$ and $V_c$ represent the switch nodes and compute nodes, respectively.
    \item $\vec{G}_x$: the auxiliary network constructed for optimality binary search. Defined in \S\ref{app:sec:binarysearch}.
    \item $G(\{Ub_e\})$: a graph obtained by multiplying each link bandwidth $b_e$ of $G$ by $U$. Defined in \S\ref{app:sec:binarysearch}.
    \item $G^{ef}$: a graph obtained by splitting off edge $e$ and $f$ of $G$. Defined in \S\ref{app:sec:edgesplit}.
    \item $G^*\!=\!(V_c,E^*)$: the graph after removing all switch nodes from $G$ using edge splitting technique. Defined in \S\ref{app:sec:edgesplit}.
    \item $\widehat{D}_{(u,w),v},\widehat{D}_{(w,t),v}$: the auxiliary networks for edge splitting (computing $\gamma$). Defined in \S\ref{app:sec:edgesplit}.
    \item $\overline{D}$: the auxiliary network for spanning tree construction (computing $\mu$). Defined in \S\ref{app:sec:constructtree}.
    \item $M$: total size of the data across all nodes.
    \item $N$: the number of compute nodes, i.e., $N=|V_c|$.
    \item $b_e$: the bandwidth of link $e$.
    \item $k$: the number of out-trees rooted at each compute node.
    \item $x$: the total bandwidth of the out-trees rooted at each compute node.
    \item $y$: the bandwidth utilized by each out-tree.
    \item $U$: equal to $1/y$. Used to scale edge capacities.
    \item $\gamma$: the maximum capacity we can safely replace (or split off) $(u,w),(w,t)$ with $(u,t)$. Defined in Theorem~\ref{app:thm:multicapasplit}.
    \item $\mu$: the maximum capacity we can add an edge into a tree. Defined in (\ref{app:eq:definemu}) of \S\ref{app:sec:constructtree}.
    \item $S,S^*$: a cut represented as a vertex subset. $S^*$ denotes the bottleneck cut, where $\frac{|S^*\cap V_c|}{B^+_G(S^*)}\!\geq\!\frac{|S\cap V_c|}{B^+_G(S)}$ for all $S\!\subset\! V,S\!\not\supseteq\! V_c$.
    \item $B^+_G(S),B^-_G(S)$: the exiting bandwidth and entering bandwidth of a vertex set $S$ on a graph $G$, i.e., sum of the bandwidths of all links exiting/entering $S$.
    \item $F(x,y;G)$: the maxflow from node $x$ to $y$ in graph $G$.
    \item $c(A,B;G)$: the total capacity from vertex set $A$ to $B$ in graph $G$, i.e., the sum of the capacities of directed edges going from $A$ to $B$.
    \item $\lambda(x,y;G)$: the edge connectivity from $x$ to $y$ in graph $G$. In integer-capacity graph, $\lambda(x,y;G)\!=\!F(x,y;G)$.
    \item $d^+(v),d^-(v)$: the in-degree and out-degree of node $v$. In integer-capacity graph, $d^+(v),d^-(v)$ are total ingress and egress capacity of $v$.
    \item $T_{u,i}$: the $i$-th out-tree rooted at node $u$.
    \item $R_{u,i},\cV(T_{u,i})$: the vertex set of the out-tree $T_{u,i}$.
    \item $\cE(T_{u,i})$: the edge set of the out-tree $T_{u,i}$.
    \item $m(R_{u,i}),g(x,y)$: notations for spanning tree construction. Defined in Theorem~\ref{app:thm:polynomialpacking}.
\squishend

\input{related_work}
\input{gen_opt}
\input{appendix/intro}
\input{appendix/algorithm}
\input{appendix/runtime}
\input{appendix/allreduce_lp}
\input{appendix/proofs}

%% file: related_work.tex
\section{Other Related Work}\label{sec:otherrelated}

\noindent\textbf{Other Optimizations of Collective Communications:} Beyond schedule generation, other works optimize collective comms in ways orthogonal to ForestColl. TopoOpt~\cite{topoopt}, Rail-only~\cite{wang2024rail}, and BFB~\cite{bfb} optimize the underlying network topology, where ForestColl's optimality for any topology can be beneficial. C-Cube~\cite{c-cube} explores efficient comm on logical trees, such as overlapping reduction and broadcast, but does not mention tree construction. BlueConnect~\cite{blueconnect} proposes a collective algorithm for single hierarchical switching fabrics but is otherwise inapplicable. Recent works~\cite{inadt,roar,flare,inaatp,switchml,rina} explore prototype hardware/protocol designs for in-network multicast/aggregation under simplistic topologies. ForestColl is compatible with these designs, extending them to more complex topologies. Other works also optimize ML comm through network transport~\cite{mlt}, packet scheduling~\cite{tictac,xinjin}, and multi-tenancy~\cite{mccs,cassini}.

\noindent\textbf{Hybrid Parallelism:} A major line of work focuses on designing hybrid parallelism strategies to co-optimize the use of comp, comm, and memory resources. Megatron-LM~\cite{megatronlm,megatrondeepak} showcases how to combine data, tensor, and pipeline parallelisms to speed up LLM training. FlexFlow~\cite{flexflow} and Alpa~\cite{alpa} propose automated search for hybrid parallelism strategies, with Alpa specifically aiming to minimize comm cost. nnScaler~\cite{nnscaler} incorporates domain experts into the search process for more parallelization opportunities.

\noindent\textbf{Comp-Comm Overlap:} Another line of research aims to enhance the overlap between comp and comm in ML training. Some implementations of parallelisms, such as PyTorch DDP~\cite{pytorchddp}, FSDP~\cite{pytorchfsdp}, and Domino~\cite{domino}, exploit overlap opportunities within a single parallelism. More complicated approaches, including CoCoNet~\cite{coconet}, Syndicate~\cite{syndicate}, and Centauri~\cite{centauri}, optimize the scheduling of comp and comm operations to achieve overlap in hybrid parallelism. Recent works like CoCoNet~\cite{coconet}, \cite{overlapdecomp}, T3~\cite{T3}, and Flux~\cite{flux} overlap at the finest scale by fusing comp and comm kernels.

ForestColl complements both hybrid parallelism and comp-comm overlap. In hybrid parallelism, comm cost plays a pivotal role in overall performance~\cite{flexflow,alpa,nnscaler}. By speeding up comm, ForestColl alleviates comm bottlenecks and enables more optimizations for comp and memory access. Its adaptability to any topology---including subsets of a topology (\S\ref{sec:mi250eval})---enables more possibilities for hybrid parallelism. While comp-comm overlap aims to hide comm cost, the massive traffic required by LLMs means that comm cost remains substantial~\cite{domino,flux,overlapdecomp,T3}. Besides, resource contention on GPUs (e.g., comp units, memory bandwidth) between comp and comm operations underscores that overlap is not cost-free~\cite{overlapdecomp,compcommdma}. ForestColl helps reduce the non-overlapped comm cost and enhance comm efficiency in overlap (\S\ref{sec:evalfsdp}).

%% file: gen_opt.tex
\section{Implementation of Schedule Generation}\label{sec:gen_opt}

In this section, we describe our implementation of ForestColl's schedule generation algorithm, focusing on how we parallelize key components to leverage multicore processors. The algorithm is implemented in Java with $\sim\!1100$ LOC. For maxflow, we use the push–relabel algorithm~\cite{preflowpush} provided by JGraphT library~\cite{jgrapht}. Table~\ref{tab:genbreakdown} shows the breakdown of schedule generation time for 1024-GPU topologies in \S\ref{sec:geneval}.

\begin{table}[h]
    \centering
    \resizebox{\columnwidth}{!}{
    \begin{tabular}{|c|c|c|c|c|}
        \hline
        \multirow{2}{*}{Topology} & \multirow{2}{*}{\shortstack[c]{Optimality \\ Binary Search}} & \multirow{2}{*}{\shortstack[c]{Switch Node \\ Removal}} & \multirow{2}{*}{\shortstack[c]{Spanning Tree \\ Construction}} & \multirow{2}{*}{Total Time} \\
        & & & & \\
        \hline
        1024-GPU A100 & 2.2s & 979s & 1209s & 36.5min \\
        \hline
        1024-GPU MI250 & 3.8s & 550s & 1708s & 37.7min \\
        \hline
    \end{tabular}
    }
    \caption{Breakdown of schedule generation time for 1024-GPU topologies in \S\ref{sec:geneval}. \textnormal{Runtimes were measured on a 128-core 2.2GHz CPU.}}
    \label{tab:genbreakdown}
    \vspace{-0.5em}
\end{table}

Optimality binary search is the fastest stage of the schedule generation algorithm. In Algorithm~\ref{algo:binarysearch}, we parallelize the computation of maxflows from node $s$ to each compute node $c$. With this parallelization, ForestColl can derive the optimal throughput of 1024-GPU topologies within seconds. For switch node removal, we similarly parallelize the computation of $\gamma$ in Algorithm~\ref{algo:switchremove}, which also requires independent maxflow computations for each compute node. The runtime of this stage depends on the amount of switches in the network; for example, switch node removal on the A100 topology takes longer than on the MI250 due to the additional NVSwitches.

Parallelizing spanning tree construction is more challenging. The main bottleneck lies in computing $\mu$ in Algorithm~\ref{app:algo:construct}. This step is difficult to parallelize because it requires only a single maxflow computation, and the push–relabel algorithm in JGraphT is not parallelized. The challenge is compounded by the sequential nature of edge additions: for each edge $e_i$, we must compute $\mu$, decide whether to add or skip $e_i$, and then move on to $e_{i+1}$. As a result, constructing $kN$ spanning trees requires computing $\mu$ for $\Omega(kN^2)$ times. To address this, we take inspiration from branch prediction. Specifically, thread $j$ speculatively computes $\mu$ for edge $e_{i+j}$ under the assumption that edges $e_i,\dots,e_{i+j-1}$ are all added. If all consecutive edges can indeed be added, they are incorporated in the wall-clock time of a single $\mu$ computation. To further accelerate the process, we also assign threads to test subsequent edges under the assumption that $e_i$ is rejected, ensuring that the next eligible edge is immediately available.

%% file: appendix/intro.tex
\section{Minimality-or-Saturation Dilemma}\label{app:sec:dilemma}

\begin{figure*}[t]
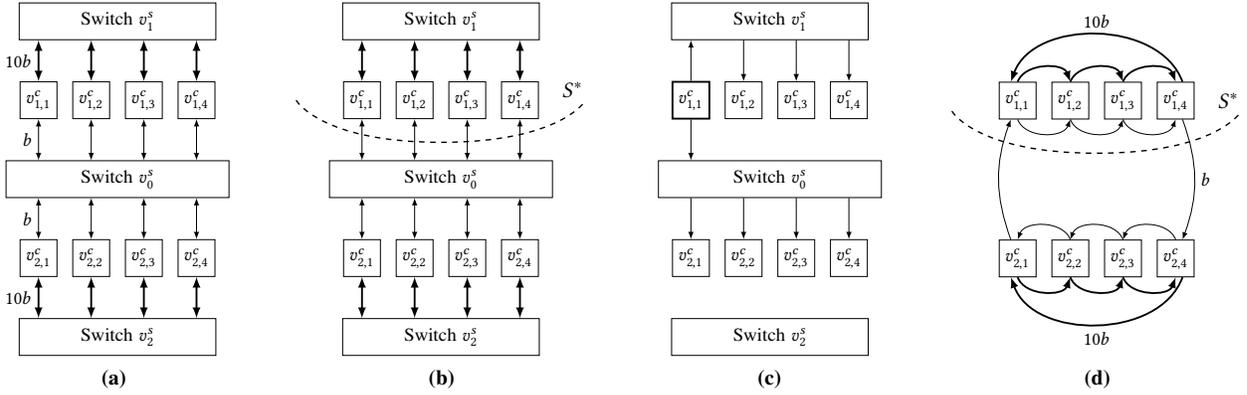

    \centering
    \begin{subfigure}{0.24\textwidth}
        \centering
        \include{appendix/figure/topo}
        \caption{}
        \label{app:fig:topo}
    \end{subfigure}
    \begin{subfigure}{0.24\textwidth}
        \centering
        \include{appendix/figure/topo_partition}
        \caption{}
        \label{app:fig:topo_partition}
    \end{subfigure}
    \begin{subfigure}{0.24\textwidth}
        \centering
        \include{appendix/figure/topo_tree}
        \caption{}
        \label{app:fig:topo_tree}
    \end{subfigure}
    \begin{subfigure}{0.24\textwidth}
        \centering
        \include{appendix/figure/topo_split2}
        \caption{}
        \label{app:fig:topo_split2}
    \end{subfigure}
    \caption{An 8-compute-node switch topology in 2-box setting. \textnormal{The thick links have 10x the bandwidth of the thin ones. Figure (a) shows the original switch topology. Figure (b) shows the bottleneck cut in this topology. Figure (c) shows a spanning tree rooted at $v_{1,1}^c$ with switch-node broadcast. Figure (d) shows a suboptimal way of transforming the switch topology into a direct-connect logical topology (resulting in 4x worse optimal performance).}}
\end{figure*}

In this section, we discuss why we need a tree-flow schedule instead of an ordinary step schedule to achieve optimality. We show that in certain situations, tree-flow schedule is \textit{the only possible way} to achieve optimality. As shown in optimality~(\ref{eq:optimality}), the performance of a topology is bounded by a bottleneck cut $(S^*,\overline{S^*})$. Suppose we want to achieve the performance bound given by the bottleneck cut, i.e., $(M/N)|S^*\cap V_c|/B_G^+(S^*)$, then the schedule must satisfy two requirements: (a) the bandwidth of the bottleneck cut, i.e., $B_G^+(S^*)$, must be saturated at all times, and (b) only the minimum amount of data required, i.e., $(M/N)|S^*\cap V_c|$, is transmitted through the bottleneck cut.

Consider the switch topology in Figure~\ref{app:fig:topo}. The topology has 8 compute nodes and 3 switch nodes. The eight compute nodes are in two boxes. Each box has a switch $v_1^s$ or $v_2^s$ providing $10b$ egress/ingress bandwidth for each compute node in the box. The 8 compute nodes are also connected to a global switch $v_0^s$, providing $b$ egress/ingress bandwidth for each compute node. It is easy to check that the bottleneck cut in this topology is a box cut $S^*=\{v_1^s,v_{1,1}^c,v_{1,2}^c,v_{1,3}^c,v_{1,4}^c\}$ shown in Figure~\ref{app:fig:topo_partition}. The cut provides a communication time lower bound of $(M/N)(4/4b)$. In comparison, a single-compute-node cut provides a much lower communication time lower bound $(M/N)(1/11b)$.

Suppose we want to achieve the lower bound by bottleneck cut $S^*$. Let $C$ be the last chunk sent through the cut to box 2, and suppose it is sent to $v_{2,1}^c$. The first thing to try is to saturate the bandwidth. It means that the schedule terminates right after $C$ is sent, leaving no idle time for $B_G^+(S^*)$. Then, at least one of $v_{2,2}^c,v_{2,3}^c,v_{2,4}^c$ must get $C$ directly from box 1 because they have no time to get it from $v_{2,1}^c$. This violates minimality, however, because chunk $C$ got sent through the bottleneck cut at least twice.

Suppose we want to achieve minimality. Then, $v_{2,1}^c$ has to broadcast $C$ to $v_{2,2}^c,v_{2,3}^c,v_{2,4}^c$ within the box. However, because $C$ is the last chunk sent through the cut by assumption, the cut bandwidth $B_G^+(S^*)$ is idle during the broadcast. The saturation requirement is violated. Thus, we are in a minimality-or-saturation dilemma that we cannot achieve both at the same time. However, we can do infinitely close by making chunk $C$ infinitesimally small. By doing so, we transmit minimum data required, and we also make the idle time of bottleneck cut close to 0. In step schedule, one always needs to specify $C$ as a fixed fraction of the total data, so it is impossible to achieve optimality in such a case. In contrast, the size of one send/recv can be arbitrarily small in tree-flow schedule. Therefore, tree-flow schedule is the only way to achieve optimality.

%% file: appendix/figure/topo.tex
\scalebox{0.7}{
\begin{tikzpicture}[node/.style={rectangle,draw=black,minimum size=7mm,align=center}]
    \node[node]	(10)	at (-1.5,2.5-1) {$v^c_{1,1}$};
    \node[node]	(11)	at (-0.5,2.5-1) {$v^c_{1,2}$};
    \node[node]	(12)	at (0.5,2.5-1) {$v^c_{1,3}$};
    \node[node]	(13)	at (1.5,2.5-1) {$v^c_{1,4}$};
    \node[node,text width=3.5cm]	(14)	at (0,0.5+2.5) {Switch $v^s_1$};
    
    \path[latex-latex,anchor=east,line width=1pt] (10.north) edge node {$10b$} (10.north|-14.south);
    \path[latex-latex,line width=1pt] (11.north) edge (11.north|-14.south);
    \path[latex-latex,line width=1pt] (12.north) edge (12.north|-14.south);
    \path[latex-latex,line width=1pt] (13.north) edge (13.north|-14.south);
    
    \node[node]	(0)	at (-1.5,1-2.5) {$v^c_{2,1}$};
    \node[node]	(1)	at (-0.5,1-2.5) {$v^c_{2,2}$};
    \node[node]	(2)	at (0.5,1-2.5) {$v^c_{2,3}$};
    \node[node]	(3)	at (1.5,1-2.5) {$v^c_{2,4}$};
    \node[node,text width=3.5cm]	(4)	at (0,-0.5-2.5) {Switch $v^s_2$};
    
    \path[latex-latex,anchor=east,line width=1pt] (0.south) edge node {$10b$} (0.south|-4.north);
    \path[latex-latex,line width=1pt] (1.south) edge (1.south|-4.north);
    \path[latex-latex,line width=1pt] (2.south) edge (2.south|-4.north);
    \path[latex-latex,line width=1pt] (3.south) edge (3.south|-4.north);
            
    \node[node,text width=4cm]	(20)	at (0,0) {Switch $v^s_0$};
    
    \path[latex-latex,anchor=east] (0.north) edge node {$b$} (0.north|-20.south);
    \path[latex-latex] (1.north) edge (1.north|-20.south);
    \path[latex-latex] (2.north) edge (2.north|-20.south);
    \path[latex-latex] (3.north) edge (3.north|-20.south);
    
    \path[latex-latex,anchor=east] (10.south) edge node {$b$} (10.south|-20.north);
    \path[latex-latex] (11.south) edge (11.south|-20.north);
    \path[latex-latex] (12.south) edge (12.south|-20.north);
    \path[latex-latex] (13.south) edge (13.south|-20.north);
\end{tikzpicture}
}

%% file: appendix/figure/topo_partition.tex
\scalebox{0.7}{
\begin{tikzpicture}[node/.style={rectangle,draw=black,minimum size=7mm,align=center}]
    \node[node]	(10)	at (-1.5,2.5-1) {$v^c_{1,1}$};
    \node[node]	(11)	at (-0.5,2.5-1) {$v^c_{1,2}$};
    \node[node]	(12)	at (0.5,2.5-1) {$v^c_{1,3}$};
    \node[node]	(13)	at (1.5,2.5-1) {$v^c_{1,4}$};
    \node[node,text width=3.5cm]	(14)	at (0,0.5+2.5) {Switch $v^s_1$};
    
    \path[latex-latex,line width=1pt] (10.north) edge (10.north|-14.south);
    \path[latex-latex,line width=1pt] (11.north) edge (11.north|-14.south);
    \path[latex-latex,line width=1pt] (12.north) edge (12.north|-14.south);
    \path[latex-latex,line width=1pt] (13.north) edge (13.north|-14.south);
    
    \node[node]	(0)	at (-1.5,1-2.5) {$v^c_{2,1}$};
    \node[node]	(1)	at (-0.5,1-2.5) {$v^c_{2,2}$};
    \node[node]	(2)	at (0.5,1-2.5) {$v^c_{2,3}$};
    \node[node]	(3)	at (1.5,1-2.5) {$v^c_{2,4}$};
    \node[node,text width=3.5cm]	(4)	at (0,-0.5-2.5) {Switch $v^s_2$};
    
    \path[latex-latex,line width=1pt] (0.south) edge (0.south|-4.north);
    \path[latex-latex,line width=1pt] (1.south) edge (1.south|-4.north);
    \path[latex-latex,line width=1pt] (2.south) edge (2.south|-4.north);
    \path[latex-latex,line width=1pt] (3.south) edge (3.south|-4.north);
            
    \node[node,text width=4cm]	(20)	at (0,0) {Switch $v^s_0$};
    
    \path[latex-latex] (0.north) edge (0.north|-20.south);
    \path[latex-latex] (1.north) edge (1.north|-20.south);
    \path[latex-latex] (2.north) edge (2.north|-20.south);
    \path[latex-latex] (3.north) edge (3.north|-20.south);
    
    \draw[line width=0.75pt, dashed] plot [smooth, tension=1.5] coordinates { (-2.75,1.5) (0,0.7) (2.75,1.5) };
    \node[font=\Large] at (2.5, 1.7) {$S^*$};
    
    \path[latex-latex] (10.south) edge (10.south|-20.north);
    \path[latex-latex] (11.south) edge (11.south|-20.north);
    \path[latex-latex] (12.south) edge (12.south|-20.north);
    \path[latex-latex] (13.south) edge (13.south|-20.north);
\end{tikzpicture}
}

%% file: appendix/figure/topo_tree.tex
\scalebox{0.7}{
\begin{tikzpicture}[node/.style={rectangle,draw=black,minimum size=7mm,align=center}]
    \node[node, line width=1pt]	(10)	at (-1.5,2.5-1) {$v^c_{1,1}$};
    \node[node]	(11)	at (-0.5,2.5-1) {$v^c_{1,2}$};
    \node[node]	(12)	at (0.5,2.5-1) {$v^c_{1,3}$};
    \node[node]	(13)	at (1.5,2.5-1) {$v^c_{1,4}$};
    \node[node,text width=3.5cm]	(14)	at (0,0.5+2.5) {Switch $v^s_1$};
    
    \path[-latex] (10.north) edge (10.north|-14.south);
    \path[latex-] (11.north) edge (11.north|-14.south);
    \path[latex-] (12.north) edge (12.north|-14.south);
    \path[latex-] (13.north) edge (13.north|-14.south);
    
    \node[node]	(0)	at (-1.5,1-2.5) {$v^c_{2,1}$};
    \node[node]	(1)	at (-0.5,1-2.5) {$v^c_{2,2}$};
    \node[node]	(2)	at (0.5,1-2.5) {$v^c_{2,3}$};
    \node[node]	(3)	at (1.5,1-2.5) {$v^c_{2,4}$};
    \node[node,text width=3.5cm]	(4)	at (0,-0.5-2.5) {Switch $v^s_2$};
            
    \node[node,text width=4cm]	(20)	at (0,0) {Switch $v^s_0$};
    
    \path[latex-] (0.north) edge (0.north|-20.south);
    \path[latex-] (1.north) edge (1.north|-20.south);
    \path[latex-] (2.north) edge (2.north|-20.south);
    \path[latex-] (3.north) edge (3.north|-20.south);
    
    \path[-latex] (10.south) edge (10.south|-20.north);
\end{tikzpicture}
}

%% file: appendix/figure/topo_split2.tex
\scalebox{0.7}{
\begin{tikzpicture}[node/.style={rectangle,draw=black,minimum size=7mm,align=center}]
    \node[node]	(10)	at (-1.5,2.5-1) {$v^c_{1,1}$};
    \node[node]	(11)	at (-0.5,2.5-1) {$v^c_{1,2}$};
    \node[node]	(12)	at (0.5,2.5-1) {$v^c_{1,3}$};
    \node[node]	(13)	at (1.5,2.5-1) {$v^c_{1,4}$};

    \node[node,text width=3.5cm,draw=none]	(14)	at (0,0.5+2.5) {};

    \draw[-latex,line width=1pt] (10.north) to[out=70, in=110] (11.north);
    \draw[-latex,line width=1pt] (11.north) to[out=70, in=110] (12.north);
    \draw[-latex,line width=1pt] (12.north) to[out=70, in=110] (13.north);
    \draw[-latex, anchor=south,line width=1pt] (13.70) to[out=110, in=70] node {$10b$} (10.110);
    
    \node[node]	(0)	at (-1.5,1-2.5) {$v^c_{2,1}$};
    \node[node]	(1)	at (-0.5,1-2.5) {$v^c_{2,2}$};
    \node[node]	(2)	at (0.5,1-2.5) {$v^c_{2,3}$};
    \node[node]	(3)	at (1.5,1-2.5) {$v^c_{2,4}$};
    
    \node[node,text width=3.5cm,draw=none]	(4)	at (0,-0.5-2.5) {};
    
    \draw[-latex,line width=1pt] (0.south) to[out=-70, in=-110] (1.south);
    \draw[-latex,line width=1pt] (1.south) to[out=-70, in=-110] (2.south);
    \draw[-latex,line width=1pt] (2.south) to[out=-70, in=-110] (3.south);
    \draw[-latex, anchor=north,line width=1pt] (3.-70) to[out=-110, in=-70] node {$10b$} (0.-110);

    \draw[latex-] (0.north) to[out=70, in=110] (1.north);
    \draw[latex-] (1.north) to[out=70, in=110] (2.north);
    \draw[latex-] (2.north) to[out=70, in=110] (3.north);
    \draw[latex-] (10.-110) to[out=-110, in=110] (0.110);
    \draw[-latex] (10.south) to[out=-70, in=-110] (11.south);
    \draw[-latex] (11.south) to[out=-70, in=-110] (12.south);
    \draw[-latex] (12.south) to[out=-70, in=-110] (13.south);
    \draw[-latex,anchor=west] (13.-70) to[out=-70, in=70] node {$b$} (3.70);

    \draw[line width=0.75pt, dashed] plot [smooth, tension=1.5] coordinates { (-2.75,1.3) (0,0.5) (2.75,1.3) };
    \node[font=\Large] at (2.5, 1.5) {$S^*$};
\end{tikzpicture}
}

%% file: appendix/algorithm.tex
\section{Algorithm Design}\label{app:sec:algodesign}

Let $G\!=\!(V\!=\!V_s\cup V_c,E)$ be an arbitrary network topology. We will compute an allgather schedule that reaches the optimal communication time~(\ref{eq:optimality}). We make two trivial assumptions about the topology: (a) all link bandwidths are integers and (b) $G$ is \textit{Eulerian}, i.e., the total egress bandwidth equals the total ingress bandwidth for any node. For (a), when bandwidths are rational numbers, one can always scale them up to become integers. For (b), we use $B^+_G(v)$ and $B^-_G(v)$ to denote the total egress and ingress bandwidth of node $v$ respectively. Since $G$ is Eulerian, we have $B^+_G(v)=B^-_G(v)$ for all $v\in V$ and, consequently, $B^+_G(S)=B^-_G(S)$ for any $S\subseteq V$.

In summary, the algorithm contains three parts:
\squishlist
    \item \S\ref{app:sec:binarysearch}: Conduct a binary search to compute the optimal communication time~(\ref{eq:optimality}). The binary search uses a network flow based oracle to test if a certain value is $\geq$ or $<$ than the true value of optimality~(\ref{eq:optimality}).
    \item \S\ref{app:sec:edgesplit}: Transform the switch topology into a direct-connect logical topology by using \textit{edge splitting} to remove switch nodes. The transformation is done without compromising optimal performance. This part can be skipped if the input topology is already direct-connect.
    \item \S\ref{app:sec:constructtree}: Construct spanning trees in direct-connect topology to achieve optimal performance. These spanning trees can then be mapped back to the original topology by reversing edge splitting, which determines the routing of communications between compute nodes.
\squishend
The algorithm design is centered on earlier graph theoretical results on constructing edge-disjoint out-trees in directed graph~\cite{bang-jensen, tarjan, edmonds, berczi2010packing, schrijver}. A key observation leading to this algorithm is that \textit{given a set of out-trees, there are at most $U$ out-trees congested on any edge of $G$, if and only if, the set of out-trees is edge-disjoint in a multigraph topology obtained by duplicating each of $G$'s edges $U$ times.} 

Another core design of our algorithm relies on \textit{edge splitting}, also a technique from graph theory~\cite{bang-jensen,frank,jackson}. It is used to transform the switch topology into a direct-connect topology so that one can construct compute-node-only spanning trees. Previous works such as TACCL~\cite{taccl} and TACOS~\cite{tacos} attempt to do this by ``unwinding'' switch topologies into predefined logical topologies, such as rings. However, their transformations often result in a loss of performance compared to the original switch topology. For example, the previous works may unwind all switches in Figure~\ref{app:fig:topo} into rings, resulting in Figure~\ref{app:fig:topo_split2}. However, it makes the bottleneck cut $S^*$ worse that the egress bandwidth of $S^*$ becomes $b$ instead of $4b$, causing optimality~(\ref{eq:optimality}) being $(M/N)(4/b)$ (4x worse). In contrast, our \textit{edge splitting} strategically removes switch nodes without sacrificing any overall performance. Our transformation generates direct-connect topology in Figure~\ref{app:fig:topo_split}, which has the same optimality as Figure~\ref{app:fig:topo}. 

In this paper, we make extensive use of network flow between different pairs of nodes. For any flow network $D$, we use $F(x,y;D)$ to denote the value of maxflow from $x$ to $y$ in $D$. For disjoint $A,B$, let $c(A,B;D)$ be the total capacity from $A$ to $B$ in $D$. By min-cut theorem, $F(x,y;D)\leq c(A,\bar{A};D)$ if $x\in A,y\in\bar{A}$, and there exists an $x$-$y$ cut $(A^*,\overline{A^*})$ that $F(x,y;D)=c(A^*,\overline{A^*};D)$.

\subsection{Optimality Binary Search}\label{app:sec:binarysearch}

In this section, we will show a way to compute the optimality~(\ref{eq:optimality}). Let $\{b_e\}_{e\in E}$ be the link bandwidths of $G$. By assumption, $\{b_e\}_{e\in E}$ are in $\Z_+$ and represented as capacities of edges in $G$. For any $x\in\Q$, we define $\vec{G}_x$ to be the flow network that (a) a source node $s$ is added and (b) an edge $(s,u)$ is added with capacity $x$ for every vertex $u\in V_c$. Now, we have the following theorem:

\begin{restatable}{theorem}{thmbinarysearch}\label{app:thm:binarysearch}
    $\min_{v\in V_c} F(s,v;\vec{G}_x)\geq|V_c|x$ if and only if $1/x\geq\max_{S\subset V,S\not\supseteq V_c}|S\cap V_c|/B^+_G(S)$.
\end{restatable}

The implication of Theorem~\ref{app:thm:binarysearch} is that we can do a binary search to get $1/x^*=\max_{S\subset V,S\not\supseteq V_c}|S\cap V_c|/B^+_G(S)$. The following initial range is trivial
\[
    \frac{N-1}{\min_{v\in V_c}B^-_G(v)}\leq\max_{S\subset V,S\not\supseteq V_c}\frac{|S\cap V_c|}{B^+_G(S)}\leq N-1.
\]
The lower bound corresponds to a partition containing all nodes except the compute node with minimum ingress bandwidth. The upper bound is due to the fact that $|S\cap V_c|\leq N-1$ and $B^+_G(S)\geq 1$. Starting with the initial range, one can then continuously test if $\min_{v\in V_c} F(s,v;\vec{G}_x)\geq|V_c|x$ for some midpoint $x$ to do a binary search. To find the exact $1/x^*$, let $S^*=\argmax_{S\subset V,S\not\supseteq V_c}|S\cap V_c|/B^+_G(S)$, then $1/x^*$ equals a fractional number with $B^+_G(S^*)$ as its denominator. Observe that $|S^*\cap V_c|\leq N-1$ and $|S^*\cap V_c|/B^+_G(S^*)\geq(N-1)/\min_{v\in V_c}B^-_G(v)$, so $B^+_G(S^*)\leq\min_{v\in V_c}B^-_G(v)$. Therefore, the denominator of $1/x^*$ is bounded by $\min_{v\in V_c}B^-_G(v)$.  Now, we use the following proposition:
\begin{restatable}{proposition}{lmfraction}\label{app:lm:fraction}
    Given two unequal fractional numbers $a/b$ and $c/d$ with $a,b,c,d\in\Z_+$, if denominators $b,d\leq X$ for some $X\in\Z_+$, then $|a/b-c/d|\geq 1/X^2$.
\end{restatable}
The proposition implies that if $1/x^*=a/b$ for some $b\leq\min_{v\in V_c}B^-_G(v)$, then any $c/d\neq 1/x^*$ with $d\leq\min_{v\in V_c}B^-_G(v)$ satisfies $|c/d-1/x^*|\geq 1/\min_{v\in V_c}B^-_G(v)^2$. Thus, one can run binary search until the range is smaller than $1/\min_{v\in V_c}B^-_G(v)^2$. Then, $1/x^*$ can be computed exactly by finding the fractional number closest to the midpoint with a denominator not exceeding $\min_{v\in V_c}B^-_G(v)$. The latter can be done with the continued fraction algorithm or brute force search if $\min_{v\in V_c}B^-_G(v)$ is small.

At the point, we have already known the optimality of communication time given a topology $G$. For the remainder of this section, we will show that there exists a family of spanning trees that achieves this optimality. First of all, we have assumed that $G$'s links have the set of bandwidths $\{b_e\}_{e\in E}$. For the simplicity of notation, we use $G(\{c_e\})$ to denote the same topology as $G$ but with the set of bandwidths $\{c_e\}_{e\in E}$ instead. $\vec{G}_x(\{c_e\})$ is also defined accordingly. When $\{c_e\}_{e\in E}$ are integers, we say a family of out-trees $\cF$ is \textit{edge-disjoint} in $G(\{c_e\})$ if the number of trees using any edge $e\in E$ is less than or equal to $c_e$, i.e., $\sum_{T\in\cF}\I[e\in T]\leq c_e$ for all $e\in E$. The intuition behind this edge-disjointness is that \textit{the integer capacity $c_e$ represents the number of multiedges from the tail to the head of $e$.}

Now, we find $U\in\Q,k\in\N$ such that $U/k=1/x^*$ and $Ub_e\in\Z_+$ for all $e\in E$. For simplicity of schedule, we want $k$ to be as small as possible. The following proposition shows how to find such $U,k$:
\begin{restatable}{proposition}{lmsmallestUk}\label{app:lm:smallestUk}
     Given $\{b_e\}_{e\in E}\subset\Z_+$ and $1/x^*\in\Q$, let $p/q$ be the simplest fractional representation of $1/x^*$, i.e., $p/q=1/x^*$ and $\gcd(p,q)=1$. Suppose $k\in\N$ is the smallest such that there exists $U\in\Q$ satisfying $U/k=1/x^*$ and $Ub_e\in\Z_+$ for all $e\in E$, then $U=p/\gcd(q,\{b_e\}_{e\in E})$ and $k=Ux^*$.
\end{restatable}
In Figure~\ref{app:fig:topo}'s example, we have $1/x^*\!=\!|S^*\cap V_c|/B_G^+(S^*)\!=\!4/4b\!=\!1/b$ and thus $U=1/b,k=1$.

Consider the digraph $G(\{Ub_e\})$. Each edge of $G(\{Ub_e\})$ has integer capacity. We will show that there exists a family of edge-disjoint out-trees $\{T_{u,i}\}_{u\in V_c,i\in[k]}$ in $G(\{Ub_e\})$ with $T_{u,i}$ rooted at $u$ and $\cV(T_{u,i})\supseteq V_c$. Here, $[k]=\{1,2,\dots,k\}$ and $\cV(T_{u,i})$ denotes the vertex set of $T_{u,i}$. We use the following theorem proven by Bang-Jensen et al.~\cite{bang-jensen}:
\begin{restatable}[Bang-Jensen et al.~\cite{bang-jensen}]{theorem}{thmbangjensentrees}\label{app:thm:bang-jensen}
    Let $n\geq 1$ and $D=(V,E)$ be a digraph with a special node $s$. Let $T'=\{v\ |\ v\in V-s,d^-(v)<d^+(v)\}$. If $\lambda(s,v;D)\geq n$ for all $v\in T'$, then there is a family $\cF$ of edge-disjoint out-trees rooted at $s$ such that every $v\in V$ belongs to at least $\min(n,\lambda(s,v;D))$ number of out-trees.
\end{restatable}
Because we see integer capacity as the number of multiedges, here, the total in-degree $d^-(v)$ and out-degree $d^+(v)$ are simply the total ingress and egress capacity of $v$ in $G(\{Ub_e\})$. $\lambda(x,y;D)$ denotes the edge-connectivity from $x$ to $y$ in $D$, i.e., $\lambda(x,y;D)=\min_{x\in A,y\in\bar{A}}c(A,\bar{A};D)$. By min-cut theorem, $\lambda(x,y;D)$ is also equal to the maxflow from $x$ to $y$. Theorem~\ref{app:thm:bang-jensen} leads to the following:
\begin{restatable}{theorem}{thmexisttrees}\label{app:thm:existtrees}
    Given integer-capacity digraph $D=(V_s\cup V_c,E)$ and $k\in\N$, there exists a family of edge-disjoint out-trees $\{T_{u,i}\}_{u\in V_c,i\in[k]}$ in $D$ with $T_{u,i}$ rooted at $u$ and $\cV(T_{u,i})\supseteq V_c$ if and only if \mbox{$\min_{v\in V_c}F(s,v;\vec{D}_k)\geq|V_c|k$.}
\end{restatable}
Consider the flow network $\vec{G}_k(\{Ub_e\})$. It is trivial to see that each edge in $\vec{G}_k(\{Ub_e\})$ has exactly $U$ times the capacity as in $\vec{G}_{x^*}$, including the edges incident from $s$. Thus, we have
\begin{align*}
    \min_{v\in V_c} F(s,v;\vec{G}_k(\{Ub_e\}))
    &=U\cdot\min_{v\in V_c}F(s,v;\vec{G}_{x^*})\\
    &\geq U\cdot|V_c|x^*\\
    &=|V_c|k.
\end{align*}
By Theorem~\ref{app:thm:existtrees}, there exists a family of edge-disjoint out-trees $\{T_{u,i}\}_{u\in V_c,i\in[k]}$ in $G(\{Ub_e\})$ with $T_{u,i}$ rooted at $u$ and $\cV(T_{u,i})\supseteq V_c$. Observe that for any edge $e \in E$, at most $Ub_e$ number of trees from $\{T_{u,i}\}_{u\in V_c,i\in[k]}$ use edge $e$. For allgather, we make each tree broadcast $1/k$ of the root's data shard, then the communication time is
\begin{multline*}
    T_{\text{comm}}\leq\max_{e\in E}\frac{M}{Nk}\cdot\frac{Ub_e}{b_e}=\frac{M}{N}\cdot\frac{U}{k}\\
    =\frac{M}{N}\cdot\frac{1}{x^*}=\frac{M}{N}\max_{S\subset V,S\not\supseteq V_c}\frac{|S\cap V_c|}{B^+_G(S)}
\end{multline*}
reaching the optimality~(\ref{eq:optimality}) given topology $G$.

At this point, one may be tempted to construct and use $\{T_{u,i}\}_{u\in V_c,i\in[k]}$ to perform allgather. However, because $T_{u,i}$ can be arbitrary tree in $G(\{Ub_e\})$, it may force switch nodes to broadcast like $v_0^s,v_1^s$ in Figure~\ref{app:fig:topo_tree}. In the following section, we introduce a way to remove switch nodes from $G(\{Ub_e\})$, while preserving the existence of out-trees with the same communication time. Afterward, we construct out-trees in the compute-node-only topology and map the communications back to $G(\{Ub_e\})$. Thus, we are able to construct a schedule with the same optimal performance but without switch-node broadcast.

\begin{figure*}[t]
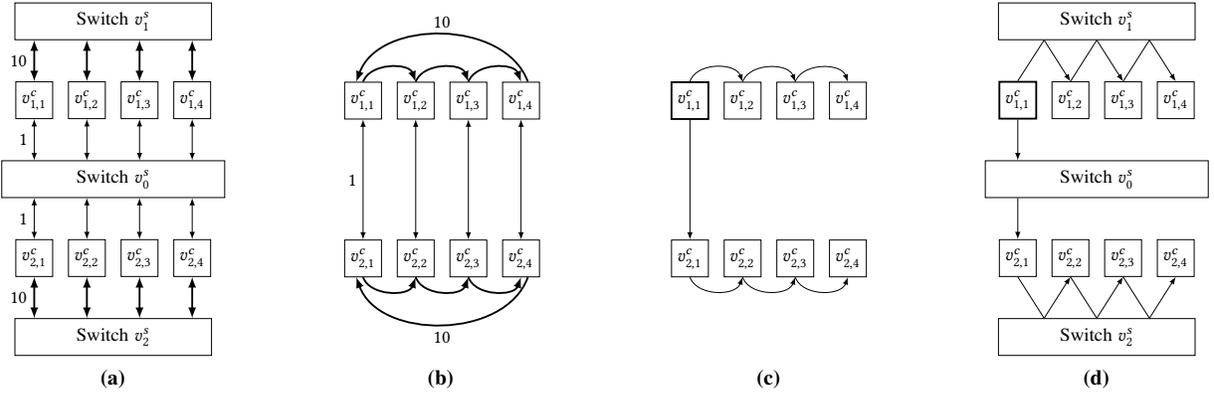

    \centering
    \begin{subfigure}{0.24\textwidth}
        \centering
        \include{appendix/figure/topo2}
        \caption{}
        \label{app:fig:topo2}
    \end{subfigure}
    \begin{subfigure}{0.24\textwidth}
        \centering
        \include{appendix/figure/topo_split}
        \caption{}
        \label{app:fig:topo_split}
    \end{subfigure}
    \begin{subfigure}{0.24\textwidth}
        \centering
        \include{appendix/figure/topo_split_tree}
        \caption{}
        \label{app:fig:topo_split_tree}
    \end{subfigure}
    \begin{subfigure}{0.24\textwidth}
        \centering
        \include{appendix/figure/topo_tree2}
        \caption{}
        \label{app:fig:topo_tree2}
    \end{subfigure}
    \caption{Different stages of the topology in schedule construction. \textnormal{Figure (a) shows the topology of $G(\{Ub_e\})$. Note that the link capacities no longer have $b$ as a multiplier. Figure (b) shows the topology $G^*$ after edge splitting removes all switch nodes. Figure (c) shows a spanning tree constructed in $G^*$. Figure (d) shows the routings in $G$ corresponding to the spanning tree.}}
    \label{app:fig:states}
\end{figure*}

\subsection{Edge Splitting}\label{app:sec:edgesplit}

{\small
\begin{algorithm}[tb]
    \small
    \SetInd{0.4em}{0.5em}
    \caption{Remove Switch Nodes}
    \label{app:algo:removeswitch}
    \SetAlgoLined
    \linespread{0.9}\selectfont
    \DontPrintSemicolon
    \KwIn{Integer-capacity Eulerian digraph $D=(V_s\cup V_c,E)$ and $k\in\N$.}
    \KwOut{Direct-connect digraph $D^*=(V_c,E^*)$ and path recovery table \emap.}
    \Begin{
        Initialize table \emap\;
        \ForEach{\normalfont switch node $w\in V_s$}{
            \ForEach{\normalfont egress edge $f=(w,t)\in E$}{
                \ForEach{\normalfont ingress edge $e=(u,w)\in E$}{
                    Compute $\gamma$ as in (\ref{app:eq:maxsplitoff}).\;
                    \If{$\gamma>0$}{
                        Decrease $f$'s and $e$'s capacity by $\gamma$. Remove $e$ if its capacity reaches 0.\;
                        Increase capacity of $(u,t)$ by $\gamma$. Add the edge if $(u,t)\notin E$.\;
                        $\emap[(u,t)][w]\leftarrow\emap[(u,t)][w]+\gamma$\;
                        \lIf{\normalfont $f$'s capacity reaches 0}{\textbf{break}}
                    }
                }
                \tcp{\normalfont\textit{Edge $f$ should have 0 capacity at this point.}}
                Remove edge $f$ from $D$.\;
            }
            \tcp{\normalfont\textit{Node $w$ should be isolated at this point.}}
            Remove node $w$ from $D$.\;
        }
        \Return{\normalfont the latest $D$ as $D^*$ and table \emap}
    }
\end{algorithm}
}

To remove the switch nodes from $G(\{Ub_e\})$, we apply a technique called \textit{edge splitting}. Consider a vertex $w$ and two incident edges $(u,w),(w,t)$. The operation of edge splitting is to replace $(u,w),(w,t)$ by a direct edge $(u,t)$ while maintaining edge-connectivities in the graph. In our context, $w$ is a switch node. We continuously split off one capacity of an incoming edge to $w$ and one capacity of an outgoing edge from $w$ until $w$ is isolated and can be removed from the graph. Because the edge-connectivities are maintained, we are able to show that $\min_{v\in V_c} F(s,v;\vec{G}_k(\{Ub_e\}))\geq|V_c|k$ is maintained in the process. Thus, by Theorem~\ref{app:thm:existtrees}, the existence of spanning trees with the same optimal performance is also preserved.

We start with the following theorem from Bang-Jensen et al.~\cite{bang-jensen}. The theorem was originally proven by Frank~\cite{frank} and Jackson~\cite{jackson}.
\begin{restatable}[Bang-Jensen et al.~\cite{bang-jensen}]{theorem}{thmbangjensenedgesplit}\label{app:thm:BJedgesplitting}
    Let $D=(V+w,E)$ be a directed Eulerian graph, that is, $d^-(x)=d^+(x)$ for every node $x$ of $D$. Then, for every edge $f=(w,t)$ there is an edge $e=(u,w)$ such that $\lambda(x,y;D^{ef})=\lambda(x,y;D)$ for every $x,y\in V$, where $D^{ef}$ is the resulting graph obtained by splitting off $e$ and $f$ in $D$.
\end{restatable}
In our case, we are not concerned with any edge-connectivity other than from $s$. In other words, we allow $\lambda(x,y;D^{ef})\neq\lambda(x,y;D)$ as long as $\min_{v\in V_c}F(s,v;\vec{D}^{ef}_{k})=\min_{v\in V_c}\lambda(s,v;\vec{D}^{ef}_{k})\geq|V_c|k$ holds after splitting. Theorem~\ref{app:thm:BJedgesplitting} is used to derive the following theorem:
\begin{restatable}{theorem}{thmedgesplit}\label{app:thm:edgesplit}
    Given integer-capacity Eulerian digraph $D=(V_s\cup V_c,E)$ and $k\in\N$ with $\min_{v\in V_c}F(s,v;\vec{D}_k)\geq|V_c|k$, for every edge $f=(w,t)$ $(w\in V_s)$ there is an edge $e=(u,w)$ such that $\min_{v\in V_c}F(s,v;\vec{D}^{ef}_{k})\geq|V_c|k$.
\end{restatable}
Note that here, $f$ and $e$ each represent one of the multiedges (or one capacity) between $w,t$ and $u,w$, respectively. Observe that edge splitting does not affect a graph being Eulerian. Thus, in $G(\{Ub_e\})$, we can iteratively replace edges $e=(u,w),f=(w,t)$ by $(u,t)$ for each switch node $w\in V_s$, while maintaining $\min_{v\in V_c}F(s,v;\vec{G}^{ef}_{k}(\{Ub_e\}))\geq|V_c|k$. The resulting graph will have all nodes in $V_s$ isolated. By removing $V_s$, we get a graph $G^*=(V_c,E^*)$ having compute nodes only. Because of Theorem~\ref{app:thm:existtrees}, there exists a family of edge-disjoint out-trees in $G^*=(V_c,E^*)$ that achieves the optimal performance.

While one can split off one capacity of $(u,w),(w,t)$ at a time, this becomes inefficient if the capacities of edges are large. Here, we introduce a way to split off $(u,w),(w,t)$ by maximum capacity at once. Given edges $(u,w),(w,t)\in E$, we construct a flow network $\widehat{D}_{(u,w),v}$ from $\vec{D}_k$ for each $v\in V_c$ that $\widehat{D}_{(u,w),v}$ connects $(u,s),(u,t),(v,w)$ with $\infty$ capacity. Similarly, we construct a flow network $\widehat{D}_{(w,t),v}$ that connects $(w,s),(u,t),(v,t)$ with $\infty$ capacity.
\begin{restatable}{theorem}{thmmulticapasplit}\label{app:thm:multicapasplit}
    Given integer-capacity Eulerian digraph $D\!=\!(V_s\cup V_c,E)$ and $k\!\in\!\N$ with $\min_{v\in V_c}F(s,v;\vec{D}_k)\!\geq\!|V_c|k$, the maximum capacity that $e\!=\!(u,w),f\!=\!(w,t)$ can be split off with the resulting graph $D^{ef}$ satisfying $\min_{v\in V_c}F(s,v;\vec{D}^{ef}_{k})\!\geq\!|V_c|k$ is
    \begin{align}\label{app:eq:maxsplitoff}
    \begin{split}
        \gamma=\min\bigg\{&c(u,w;D)\ ,\ c(w,t;D)\ ,\\
        &\min_{v\in V_c} F(u,w;\widehat{D}_{(u,w),v})-|V_c|k\ ,\\
        &\min_{v\in V_c} F(w,t;\widehat{D}_{(w,t),v})-|V_c|k\bigg\}.
    \end{split}
    \end{align}
\end{restatable}
Based on Theorem~\ref{app:thm:multicapasplit}, we are able to develop Algorithm~\ref{app:algo:removeswitch}. What is remarkable about Algorithm~\ref{app:algo:removeswitch} is that its runtime does not depend on the capacities of the digraph. One should also note that we update a table \emap while splitting. After edge splitting, we are ready to construct spanning trees that only use compute nodes for broadcast. \emap is then used to convert the spanning trees back to paths in $G$ that use switch nodes for send/receive between compute nodes.

Figure~\ref{app:fig:states} gives an example of edge splitting. In Figure~\ref{app:fig:topo2}, within each box $i\in\{1,2\}$, we split off 10 capacity of $(v^c_{i,j},v^s_{i}),(v^s_{i},v^c_{i,(j\bmod 4)+1})$ for $j=1,2,3,4$ to form a ring topology. Across boxes, we split off $1$ capacity of $(v^c_{i,j},v^s_0),(v^s_0,v^c_{(i\bmod 2)+1,j})$ for $j=1,2,3,4$. The resulting topology Figure~\ref{app:fig:topo_split} has compute nodes only, and the optimal communication time is still $(M/N)(4/4b)$ if bandwidth multiplier $b$ is added. Note that for this example, in the innermost foreach loop of Algorithm~\ref{app:algo:removeswitch}, we adjusted the order of iterating through $e$s to prioritize splitting off $(u,v_0^s),(v_0^s,t)$ pairs with $u,t$ in different boxes. The adjustment is not for performance-related reasons, but rather to simplify routing by scheduling all intra-box traffic through the in-box switch. We successfully achieved this goal: the capacity of each $f$ reaches 0 before we iterate to an $e$ with $u,t$ in the same box.

\input{appendix/spanning_tree}

\begin{figure}[tb]
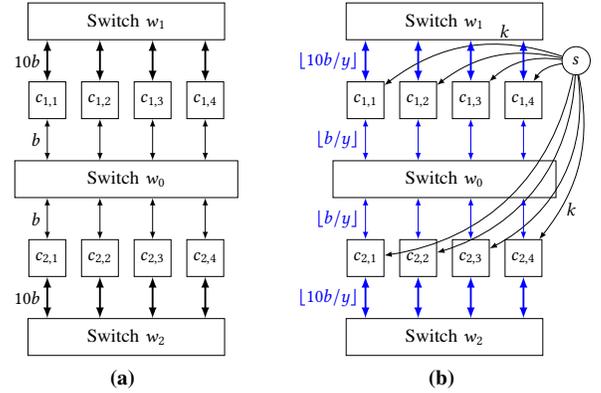

    \centering
    \begin{subfigure}{0.49\columnwidth}
        \centering
        \include{figures/topo_fixed_k}
        \caption{}
    \end{subfigure}
    \begin{subfigure}{0.49\columnwidth}
        \centering
        \include{figures/flow_fixed_k}
        \caption{}
    \end{subfigure}
    \caption{The auxiliary network for fixed-$k$ binary search. \textnormal{(a) shows the original topology. (b) shows the auxiliary network ForestColl uses to binary search for the optimal $y$ (the bandwidth utilized by each tree) given a fixed $k$ (the number of trees rooted at each compute node).}}
    \label{fig:fixedk}
\end{figure}

{\small
\begin{algorithm}[tb]
    \small
    \caption{Fixed-$k$ Binary Search}
    \label{algo:fixedk}
    \SetAlgoLined
    \linespread{0.9}\selectfont
    \DontPrintSemicolon
    \KwIn{A directed graph $G=(V_s\cup V_c,E)$ and $k$, the number of trees rooted at each compute node.}
    \KwOut{$y^*$, the maximum bandwidth of each tree.}
    \Begin{
        $l\leftarrow\frac{(N-1)k}{\min_{v\in V_c}B^-_G(v)}$\tcp*{\normalfont\textit{a lower bound of $\frac{1}{y^*}$}}
        $r\leftarrow (N-1)k$\tcp*{\normalfont\textit{an upper bound of $\frac{1}{y^*}$}}
        \While{$r-l\geq 1/\max_{e\in E}b_e^2$}{
            $\frac{1}{y}\leftarrow(l+r)/2$\;
            Add node $s$ to $G$.\;
            \ForEach{\normalfont compute node $c\in V_c$}{
                Add an edge from $s$ to $c$ with capacity $k$.\;
            }
            \ForEach{\normalfont link $e\in E$ with bandwidth $b_e$}{
                Adjust the capacity of $e$ to $\lfloor b_e/y\rfloor$.\;
            }
            \eIf{\normalfont the maxflow from $s$ to each $c\in V_c$ is $Nk$}{
                $r\leftarrow\frac{1}{y}$\tcp*{\normalfont\textit{case $\frac{1}{y}\!\geq\!\frac{1}{y^*}$}}
            }{
                $l\leftarrow\frac{1}{y}$\tcp*{\normalfont\textit{case $\frac{1}{y}\!<\!\frac{1}{y^*}$}}
            }
        }
        Find the unique fractional number $\frac{p}{q}\in[l,r]$ such that denominator $q\leq\max_{e\in E}b_e$.\;
        \Return{\normalfont $\frac{q}{p}$ as $y^*$}
    }
\end{algorithm}
}

\subsection[Fixed-k Optimality]{Fixed-$k$ Optimality}\label{app:sec:fixedk}

A potential problem of our schedule is that $k$, the number of spanning trees per root, depends linearly on link bandwidths, potentially reaching up to $\min_{v\in V_c}B^-_G(v)/\gcd(\{b_e\}_{e\in E})$. Although the runtime of spanning tree construction does not depend on $k$, in practice, one may want to reduce $k$ to simplify the schedule. In this section, we offer a way to construct a schedule with the best possible performance for a fixed $k$. We start with the following theorem:

\begin{restatable}{theorem}{thmfixedk}\label{app:thm:fixedk}
    Given $U\in\R_{+}$ and $k\in\N$, a family of out-trees $\{T_{u,i}\}_{u\in V_c,i\in[k]}$ with $T_{u,i}$ rooted at $u$ and $\cV(T_{u,i})\supseteq V_c$ achieves $\frac{M}{Nk}\cdot U$ communication time if and only if it is edge-disjoint in $G(\{\lfloor Ub_e\rfloor\}_{e\in E})$.
\end{restatable}

To test the existence of edge disjoint $\{T_{u,i}\}_{u\in V_c,i\in[k]}$ in $G(\{\lfloor Ub_e\rfloor\}_{e\in E})$, by Theorem~\ref{app:thm:existtrees}, we can simply test whether $\min_{v\in V_c}F(s,v;\vec{G}_k(\{\lfloor Ub_e\rfloor\}))\geq|V_c|k$ holds. The following theorem provides a method for binary search to find the lowest communication time for the given $k$.

\begin{restatable}{theorem}{thmfixedkbinarysearch}\label{app:thm:fixedkbinarysearch}
    Let $\frac{M}{Nk}\cdot U^*$ be the lowest communication time that can be achieved with $k$ out-trees per $v\in V_c$. Then, there exists a family of edge-disjoint out-trees $\{T_{u,i}\}_{u\in V_c,i\in[k]}$ in $G(\{\lfloor Ub_e\rfloor\}_{e\in E})$ with $T_{u,i}$ rooted at $u$ and $\cV(T_{u,i})\supseteq V_c$ if and only if $U\geq U^*$.
\end{restatable}
The initial range is:
\[
    \frac{(N-1)k}{\min_{v\in V_c}B^-_G(v)}\leq U^*\leq (N-1)k.
\]
Observe that there must exists $b_e\in E$ such that $U^*b_e\in\Z_+$; otherwise, $U^*$ can be further decreased. Thus, the denominator of $U^*$ must be less than or equal to $\max_{e\in E}b_e$. Similar to optimality binary search, by Proposition~\ref{app:lm:fraction}, one can run binary search until the range is smaller than $1/\max_{e\in E}b_e^2$. Then, $U^*$ can be determined exactly by computing the fractional number that is closest to the midpoint, while having a denominator less than or equal to $\max_{e\in E}b_e$. Algorithm~\ref{algo:fixedk} and Figure~\ref{fig:fixedk} show the pseudocode and auxiliary network for binary search, respectively, with $y\!=\!1/U$. After having $U^*$, one can simply apply edge splitting and spanning tree construction to $G(\{\lfloor U^*b_e\rfloor\}_{e\in E})$ to derive the pipeline schedule. \textit{Note that $G(\{\lfloor U^*b_e\rfloor\}_{e\in E})$ is not necessarily Eulerian. If $G(\{\lfloor U^*b_e\rfloor\}_{e\in E})$ is not Eulerian, then edge splitting cannot be applied. However, in cases where $G$ is bidirectional, $G(\{\lfloor U^*b_e\rfloor\}_{e\in E})$ is guaranteed to be Eulerian.}

The following theorem gives a bound on how close $\frac{M}{Nk}\cdot U^*$ is to optimality~(\ref{eq:optimality}):
\begin{restatable}{theorem}{thmfixedkapprox}\label{app:thm:fixedkapprox}
    Let $\frac{M}{Nk}\cdot U^*$ be the lowest communication time that can be achieved with $k$ out-trees per $v\in V_c$. Then,
    \[
        \frac{M}{Nk}\cdot U^*\leq\frac{M}{N}\max_{S\subset V,S\not\supseteq V_c}\frac{|S\cap V_c|}{B^+_G(S)}+\frac{M}{Nk}\cdot\frac{1}{\min_{e\in E}b_e}.
    \]
\end{restatable}

%% file: appendix/figure/topo2.tex
\scalebox{0.7}{
\begin{tikzpicture}[node/.style={rectangle,draw=black,minimum size=7mm,align=center}]
    \node[node]	(10)	at (-1.5,2.5-1) {$v^c_{1,1}$};
    \node[node]	(11)	at (-0.5,2.5-1) {$v^c_{1,2}$};
    \node[node]	(12)	at (0.5,2.5-1) {$v^c_{1,3}$};
    \node[node]	(13)	at (1.5,2.5-1) {$v^c_{1,4}$};
    \node[node,text width=3.5cm]	(14)	at (0,0.5+2.5) {Switch $v^s_1$};
    
    \path[latex-latex,anchor=east,line width=1pt] (10.north) edge node {$10$} (10.north|-14.south);
    \path[latex-latex,line width=1pt] (11.north) edge (11.north|-14.south);
    \path[latex-latex,line width=1pt] (12.north) edge (12.north|-14.south);
    \path[latex-latex,line width=1pt] (13.north) edge (13.north|-14.south);
    
    \node[node]	(0)	at (-1.5,1-2.5) {$v^c_{2,1}$};
    \node[node]	(1)	at (-0.5,1-2.5) {$v^c_{2,2}$};
    \node[node]	(2)	at (0.5,1-2.5) {$v^c_{2,3}$};
    \node[node]	(3)	at (1.5,1-2.5) {$v^c_{2,4}$};
    \node[node,text width=3.5cm]	(4)	at (0,-0.5-2.5) {Switch $v^s_2$};
    
    \path[latex-latex,anchor=east,line width=1pt] (0.south) edge node {$10$} (0.south|-4.north);
    \path[latex-latex,line width=1pt] (1.south) edge (1.south|-4.north);
    \path[latex-latex,line width=1pt] (2.south) edge (2.south|-4.north);
    \path[latex-latex,line width=1pt] (3.south) edge (3.south|-4.north);
            
    \node[node,text width=4cm]	(20)	at (0,0) {Switch $v^s_0$};
    
    \path[latex-latex,anchor=east] (0.north) edge node {$1$} (0.north|-20.south);
    \path[latex-latex] (1.north) edge (1.north|-20.south);
    \path[latex-latex] (2.north) edge (2.north|-20.south);
    \path[latex-latex] (3.north) edge (3.north|-20.south);
    
    \path[latex-latex,anchor=east] (10.south) edge node {$1$} (10.south|-20.north);
    \path[latex-latex] (11.south) edge (11.south|-20.north);
    \path[latex-latex] (12.south) edge (12.south|-20.north);
    \path[latex-latex] (13.south) edge (13.south|-20.north);
\end{tikzpicture}
}

%% file: appendix/figure/topo_split.tex
\scalebox{0.7}{
\begin{tikzpicture}[node/.style={rectangle,draw=black,minimum size=7mm,align=center}]
    \node[node]	(10)	at (-1.5,2.5-1) {$v^c_{1,1}$};
    \node[node]	(11)	at (-0.5,2.5-1) {$v^c_{1,2}$};
    \node[node]	(12)	at (0.5,2.5-1) {$v^c_{1,3}$};
    \node[node]	(13)	at (1.5,2.5-1) {$v^c_{1,4}$};

    \node[node,text width=3.5cm,draw=none]	(14)	at (0,0.5+2.5) {};

    \draw[-latex,line width=1pt] (10.north) to[out=70, in=110] (11.north);
    \draw[-latex,line width=1pt] (11.north) to[out=70, in=110] (12.north);
    \draw[-latex,line width=1pt] (12.north) to[out=70, in=110] (13.north);
    \draw[-latex, anchor=south,line width=1pt] (13.70) to[out=110, in=70] node {$10$} (10.110);
    
    \node[node]	(0)	at (-1.5,1-2.5) {$v^c_{2,1}$};
    \node[node]	(1)	at (-0.5,1-2.5) {$v^c_{2,2}$};
    \node[node]	(2)	at (0.5,1-2.5) {$v^c_{2,3}$};
    \node[node]	(3)	at (1.5,1-2.5) {$v^c_{2,4}$};
    
    \node[node,text width=3.5cm,draw=none]	(4)	at (0,-0.5-2.5) {};
    
    \draw[-latex,line width=1pt] (0.south) to[out=-70, in=-110] (1.south);
    \draw[-latex,line width=1pt] (1.south) to[out=-70, in=-110] (2.south);
    \draw[-latex,line width=1pt] (2.south) to[out=-70, in=-110] (3.south);
    \draw[-latex, anchor=north,line width=1pt] (3.-70) to[out=-110, in=-70] node {$10$} (0.-110);
    
    \path[latex-latex, anchor=east] (0.north) edge node {$1$} (0.north|-10.south);
    \path[latex-latex] (1.north) edge (1.north|-11.south);
    \path[latex-latex] (2.north) edge (2.north|-12.south);
    \path[latex-latex] (3.north) edge (3.north|-13.south);
\end{tikzpicture}
}

%% file: appendix/figure/topo_split_tree.tex
\scalebox{0.7}{
\begin{tikzpicture}[node/.style={rectangle,draw=black,minimum size=7mm,align=center}]
    \node[node, line width=1pt]	(10)	at (-1.5,2.5-1) {$v^c_{1,1}$};
    \node[node]	(11)	at (-0.5,2.5-1) {$v^c_{1,2}$};
    \node[node]	(12)	at (0.5,2.5-1) {$v^c_{1,3}$};
    \node[node]	(13)	at (1.5,2.5-1) {$v^c_{1,4}$};

    \node[node,text width=3.5cm,draw=none]	(14)	at (0,0.5+2.5) {};

    \draw[-latex] (10.north) to[out=70, in=110] (11.north);
    \draw[-latex] (11.north) to[out=70, in=110] (12.north);
    \draw[-latex] (12.north) to[out=70, in=110] (13.north);
    \draw[-latex, draw=none] (13.70) to[out=110, in=70] (10.110);
    
    \node[node]	(0)	at (-1.5,1-2.5) {$v^c_{2,1}$};
    \node[node]	(1)	at (-0.5,1-2.5) {$v^c_{2,2}$};
    \node[node]	(2)	at (0.5,1-2.5) {$v^c_{2,3}$};
    \node[node]	(3)	at (1.5,1-2.5) {$v^c_{2,4}$};

    \node[node,text width=3.5cm,draw=none]	(4)	at (0,-0.5-2.5) {};
    
    \draw[-latex] (0.south) to[out=-70, in=-110] (1.south);
    \draw[-latex] (1.south) to[out=-70, in=-110] (2.south);
    \draw[-latex] (2.south) to[out=-70, in=-110] (3.south);
    \draw[-latex, draw=none] (3.-70) to[out=-110, in=-70] (0.-110);
    
    \path[latex-] (0.north) edge (0.north|-10.south);
\end{tikzpicture}
}

%% file: appendix/figure/topo_tree2.tex
\scalebox{0.7}{
\begin{tikzpicture}[node/.style={rectangle,draw=black,minimum size=7mm,align=center}]
    \node[node, line width=1pt]	(10)	at (-1.5,2.5-1) {$v^c_{1,1}$};
    \node[node]	(11)	at (-0.5,2.5-1) {$v^c_{1,2}$};
    \node[node]	(12)	at (0.5,2.5-1) {$v^c_{1,3}$};
    \node[node]	(13)	at (1.5,2.5-1) {$v^c_{1,4}$};
    \node[node,text width=3.5cm]	(14)	at (0,0.5+2.5) {Switch $v^s_1$};
    
    \draw (10.north) -- ($(14.south -| -1, 0)$);
    \draw[-latex] ($(14.south -| -1, 0)$) -- (11.north);
    
    \draw (11.north) -- ($(14.south -| 0, 0)$);
    \draw[-latex] ($(14.south -| 0, 0)$) -- (12.north);
    
    \draw (12.north) -- ($(14.south -| 1, 0)$);
    \draw[-latex] ($(14.south -| 1, 0)$) -- (13.north);

    \node[node]	(0)	at (-1.5,1-2.5) {$v^c_{2,1}$};
    \node[node]	(1)	at (-0.5,1-2.5) {$v^c_{2,2}$};
    \node[node]	(2)	at (0.5,1-2.5) {$v^c_{2,3}$};
    \node[node]	(3)	at (1.5,1-2.5) {$v^c_{2,4}$};
    \node[node,text width=3.5cm]	(4)	at (0,-0.5-2.5) {Switch $v^s_2$};
            
    \node[node,text width=4cm]	(20)	at (0,0) {Switch $v^s_0$};
    
    \path[latex-] (0.north) edge (0.north|-20.south);
    \path[-latex] (10.south) edge (10.south|-20.north);
    
    \draw (0.south) -- ($(4.north -| -1, 0)$);
    \draw[-latex] ($(4.north -| -1, 0)$) -- (1.south);
    
    \draw (1.south) -- ($(4.north -| 0, 0)$);
    \draw[-latex] ($(4.north -| 0, 0)$) -- (2.south);
    
    \draw (2.south) -- ($(4.north -| 1, 0)$);
    \draw[-latex] ($(4.north -| 1, 0)$) -- (3.south);
\end{tikzpicture}
}

%% file: appendix/spanning_tree.tex
\subsection{Spanning Tree Construction}\label{app:sec:constructtree}

{\small
\begin{algorithm}[tb]
    \small
    \SetInd{0.4em}{0.5em}
    \caption{Spanning Tree Construction}
    \label{app:algo:construct}
    \SetAlgoLined
    \linespread{0.9}\selectfont
    \DontPrintSemicolon
    \KwIn{Integer-capacity digraph $D^*=(V_c,E^*)$ and $k\in\N$.}
    \KwOut{Spanning tree $(R_{u,i},\cE(R_{u,i}))$ for each $u\in V_c,i\in[n_u]$. Subgraph $(R_{u,i},\cE(R_{u,i}))$s satisfy $\forall u\in V_c:\sum_{i=1}^{n_u}m(R_{u,i})=k$ and $\forall e\in E^*:\sum\{m(R_{u,i})\ |\ e\in \cE(R_{u,i})\}\leq c(e;D^*)$.}
    \Begin{
        Initialize $R_{u,1}=\{u\},\cE(R_{u,1})=\emptyset,m(R_{u,1})=k,n_u=1$ for all $u\in V_c$.\;
        Initialize $g(e)=c(e;D^*)$ for all $e\in E^*$.\;
        \While{\normalfont there exists $R_{u,i}\neq V_c$}{
            \While{\normalfont $R_{u,i}\neq V_c$}{
                Pick an edge $(x,y)$ in $D^*$ that $x\in R_{u,i},y\notin R_{u,i}$.\;
                Compute $\mu$ as in (\ref{app:eq:computemu}).\;
                \lIf{\normalfont $\mu=0$}{\textbf{continue}}
                \If{\normalfont $\mu<m(R_{u,i})$}{
                    $n_{u}\leftarrow n_{u}+1$\;
                    Create a new copy $R_{u,n_{u}}=R_{u,i},\cE(R_{u,n_{u}})=\cE(R_{u,i}),m(R_{u,n_{u}})=m(R_{u,i})-\mu$.\;
                    $m(R_{u,i})\leftarrow\mu$\;
                }
                $\cE(R_{u,i})\leftarrow\cE(R_{u,i})+(x,y)$\;
                $R_{u,i}\leftarrow R_{u,i}+y$\;
                $g(x,y)\leftarrow g(x,y)-\mu$.\;
                Remove $(x,y)$ if $g(x,y)$ reaches $0$.\;
            }
        }
    }
\end{algorithm}
}

At this point, we have a digraph $G^*=(V_c,E^*)$ with only compute nodes. In this section, we construct $k$ out-trees from every node that span all nodes $V_c$ in $G^*$. We start by showing the existence of spanning trees with the following theorem in Tarjan~\cite{tarjan}. The theorem was originally proven by Edmonds~\cite{edmonds}.
\begin{restatable}[Tarjan~\cite{tarjan}]{theorem}{thmtarjan}\label{app:thm:tarjan}
    For any integer-capacity digraph $D=(V,E)$ and any sets $R_i\subseteq V$, $i\in[k]$, there exist $k$ edge-disjoint spanning out-trees $T_i$, $i\in[k]$, rooted respectively at $R_i$, if and only if for every $S\neq V$,   
    \begin{equation}\label{app:eq:tjcondition}
        c(S,\bar{S};D)\geq|\{i\ |\ R_i\subseteq S\}|.
    \end{equation}
\end{restatable}
A spanning out-tree is \textit{rooted at $R_i$} if for every $v\in V-R_i$, there is exactly one directed path from a vertex in $R_i$ to $v$ within the acyclic subgraph of out-tree. To see there exists a family of edge-disjoint spanning out-trees $\{T_{u,i}\}_{u\in V_c,i\in[k]}$ in $G^*$, observe that each $T_{u,i}$ is rooted at $R_{u,i}=\{u\}$, so $|\{(u,i)\ |\ R_{u,i}\subseteq S\}|=|S|k$ for any $S\subset V_c,S\neq V_c$. We show the following theorem:
\begin{restatable}{theorem}{thminitialexist}\label{app:thm:initialexist}
    Given integer-capacity digraph $D=(V_c,E)$ and $k\in\N$, $c(S,\bar{S};D)\geq|S|k$ for all $S\subset V_c,S\neq V_c$ if and only if $\min_{v\in V_c}F(s,v;\vec{D}_k)\geq|V_c|k$.
\end{restatable}
Since we ensured $\min_{v\in V_c}F(s,v;\vec{G}^*_k)\geq|V_c|k$, condition~(\ref{app:eq:tjcondition}) is satisfied. Spanning tree construction essentially involves iteratively expanding each $R_{u,i}=\cV(T_{u,i})$ from $\{u\}$ to $V_c$ by adding edges to $T_{u,i}$, while maintaining condition~(\ref{app:eq:tjcondition}). Tarjan~\cite{tarjan} has proposed such an algorithm. For each $T_{u,i}$, the algorithm continuously finds an edge $(x,y)$ with $x\in R_{u,i},y\notin R_{u,i}$ that adding this edge to $T_{u,i}$ does not violate (\ref{app:eq:tjcondition}). It is proven that such an edge is guaranteed to exist. However, the runtime of the algorithm quadratically depends on the total number of spanning trees, i.e., $Nk$ in our case. This becomes problematic when $k$ is large, as $k$ can get up to $\min_{v\in V_c}B^-_G(v)/\gcd(\{b_e\}_{e\in E})$. Fortunately, B{\'e}rczi \& Frank~\cite{berczi2010packing} has proposed a strongly polynomial time algorithm based on Schrijver~\cite{schrijver}. The runtime of the algorithm does not depend on $k$ at all. In particular, the following theorem has been shown:
\begin{restatable}[B{\'e}rczi \& Frank~\cite{berczi2010packing}]{theorem}{thmpolynomialpacking}\label{app:thm:polynomialpacking}
    Let $D\!=\!(V,E)$ be a digraph, $g:E\!\to\!\Z_+$ a capacity function, $\mathcal{R}\!=\!\{R_1,\dots,R_n\}$ a list of root-sets, $\mathcal{U}\!=\!\{U_1,\dots,U_n\}$ a set of convex sets with $R_i\!\subseteq\! U_i$, and $m:\mathcal{R}\!\to\!\Z_+$ a demand function. There is a strongly polynomial time algorithm that finds (if there exist) $m(\mathcal{R})$ out-trees so that $m(R_i)$ of them are spanning $U_i$ with root-set $R_i$ and each edge $e\in E$ is contained in at most $g(e)$ out-trees.
\end{restatable}
In our context, we start with $\mathcal{R}=\{R_u\ |\ u\in V_c\}$ and $R_u=\{u\},U_u=V_c,m(R_u)=k$. We define $\cE(R_i)$ to be the edge set of the $m(R_i)$ out-trees corresponding to $R_i$, so $\cE(R_u)=\emptyset$ is initialized. Given $\mathcal{R}=\{R_1,\dots,R_n\}$, we pick an $R_i\neq V_c$, say $R_1$. Then, we find an edge $(x,y)$ such that $x\in R_1,y\notin R_1$ and $(x,y)$ can be added to $\mu:0<\mu\leq\min\{g(x,y),m(R_1)\}$ copies of the $m(R_1)$ out-trees without violating (\ref{app:eq:tjcondition}). If $\mu=m(R_1)$, then we directly add $(x,y)$ to $\cE(R_1)$ and $R_1=R_1+y$. If $\mu<m(R_1)$, then we add a copy $R_{n+1}$ of $R_1$ that $\cE(R_{n+1})=\cE(R_1),m(R_{n+1})=m(R_1)-\mu$. We revise $m(R_1)$ to $\mu$, add $(x,y)$ to $\cE(R_1)$, and $R_1=R_1+y$. Finally, we update $g(x,y)=g(x,y)-\mu$. Now, given $\mathcal{R}=\{R_1,\dots,R_{n+1}\}$, we can apply the step repeatedly until $R_i=V_c$ for all $R_i\in\mathcal{R}$. According to B{\'e}rczi \& Frank~\cite{berczi2010packing}, $\mu$ is defined as followed:
\begin{align}\label{app:eq:definemu}
\begin{split}
    \textstyle\mu=&\min\big\{g(x,y)\ ,\ m(R_1)\ ,\\
    &\textstyle\quad\min\{c(S,\bar{S};D)-p(S;D):x\in S,y\in\bar{S},R_1\not\subseteq S\}\big\}
\end{split}
\end{align}
where $p(S;D)=\sum\{m(R_i)\ |\ R_i\subseteq S\}$. Neither B{\'e}rczi \& Frank~\cite{berczi2010packing} nor Schrijver~\cite{schrijver} explicitly mentioned how to compute $\mu$ in polynomial time. Therefore, we describe a method for doing so. We construct a flow network $\overline{D}$ such that (a) a node $s_i$ is added for each $R_i$ except $i=1$, (b) connect $x$ to each $s_i$ with capacity $m(R_i)$, and (c) connect each $s_i$ to every vertex in $R_i$ with $\infty$ capacity. We then show the following result:
\begin{restatable}{theorem}{thmcomputemu}\label{app:thm:computemu}
    For any edge $(x,y)$ in $D$ with $x\in R_1,y\notin R_1$,
    \begin{equation}\label{app:eq:computemu}
        \textstyle\mu=\min\left\{g(x,y)\ ,\ m(R_1)\ ,\ F(x,y;\overline{D})-\sum_{i\neq 1}m(R_i)\right\}.
    \end{equation}
\end{restatable}
Thus, $\mu$ can be calculated by computing a single maxflow from $x$ to $y$ in $\overline{D}$. The complete algorithm is described in Algorithm~\ref{app:algo:construct}. The resulting $\mathcal{R}$ can be indexed as $\mathcal{R}\!=\!\bigcup_{u\in V_c}\{R_{u,1},\dots,R_{u,n_{u}}\}$, where $R_{u,i}$ corresponds to $m(R_{u,i})$ number of identical out-trees rooted at $u$ and specified by edge set $\cE(R_{u,i})$. We have $\sum_{i=1}^{n_u}m(R_{u,i})=k$ for all $u$. Thus, $\mathcal{R}$ can be decomposed into $\{T_{u,i}\}_{u\in V_c,i\in[k]}$. However, since all spanning trees within $R_{u,i}$ are identical, the allgather schedule can simply be specified in terms of $\cE(R_{u,i})$ and $m(R_{u,i})$.

After construction, we have edge-disjoint spanning trees $\{T_{u,i}\}_{u\in V_c,i\in [k]}$ in $G^*$. Each of the edge $(u,v)$ in $T_{u,i}$ may correspond to a path $u\!\to\! w_1\!\to\!\dots\!\to\! w_n\!\to\! v$ in $G$ with $w_1,\dots,w_n$ being switch nodes. In other words, edges in $T_{u,i}$ only specify the source and destination of send/recv between compute nodes. Thus, one needs to use the \emap in Algorithm~\ref{app:algo:removeswitch} to recover the paths in $G$. For any edge $(u,t)$ in $G^*$, $\emap[(u,t)][w]$ denotes the amount of capacity from $u$ to $t$ that is going through $(u,w),(w,t)$. It should be noted that \emap may be recursive, meaning that $(u,w),(w,t)$ may also go through some other switches. Because each capacity of $(u,t)$ corresponds to one capacity of a path from $u$ to $t$ in $G$, the resulting schedule in $G$ has the same performance in $G^*$, achieving the optimal performance~(\ref{eq:optimality}). 

In Figure~\ref{app:fig:states}'s example, we construct a spanning tree like \ref{app:fig:topo_split_tree} for each of the compute node. By reversing the edge splitting with \emap, the spanning tree becomes the schedule in \ref{app:fig:topo_tree2}. Note that the corresponding schedule of a spanning tree in $G^*$ is not necessarily a tree in $G$. For example, the schedule in \ref{app:fig:topo_tree2} visits switches $v_1^s,v_2^s$ multiple times. To obtain a complete allgather schedule with optimal communication time $(M/N)(4/4b)$, one can apply the similar schedule for each of the compute nodes in \ref{app:fig:topo_tree2}.

One may be tempted to devise a way to construct spanning trees with low heights. This has numerous benefits such as lower latency at small data sizes and better convergence towards optimality. Although there is indeed potential progress to be made in this direction, constructing edge-disjoint spanning trees of minimum height has been proven to be NP-complete~\cite{npcomplete}.

%% file: figures/topo_fixed_k.tex
\scalebox{0.7}{
\begin{tikzpicture}[node/.style={rectangle,draw=black,minimum size=7mm,align=center}]
    \node[node]	(10)	at (-1.5,2.5-1) {$c_{1,1}$};
    \node[node]	(11)	at (-0.5,2.5-1) {$c_{1,2}$};
    \node[node]	(12)	at (0.5,2.5-1) {$c_{1,3}$};
    \node[node]	(13)	at (1.5,2.5-1) {$c_{1,4}$};
    \node[node,text width=3.5cm]	(14)	at (0,0.5+2.5) {Switch $w_1$};
    
    \path[latex-latex,anchor=east,line width=1pt] (10.north) edge node {$10b$} (10.north|-14.south);
    \path[latex-latex,line width=1pt] (11.north) edge (11.north|-14.south);
    \path[latex-latex,line width=1pt] (12.north) edge (12.north|-14.south);
    \path[latex-latex,line width=1pt] (13.north) edge (13.north|-14.south);
    
    \node[node]	(0)	at (-1.5,1-2.5) {$c_{2,1}$};
    \node[node]	(1)	at (-0.5,1-2.5) {$c_{2,2}$};
    \node[node]	(2)	at (0.5,1-2.5) {$c_{2,3}$};
    \node[node]	(3)	at (1.5,1-2.5) {$c_{2,4}$};
    \node[node,text width=3.5cm]	(4)	at (0,-0.5-2.5) {Switch $w_2$};
    
    \path[latex-latex,anchor=east,line width=1pt] (0.south) edge node {$10b$} (0.south|-4.north);
    \path[latex-latex,line width=1pt] (1.south) edge (1.south|-4.north);
    \path[latex-latex,line width=1pt] (2.south) edge (2.south|-4.north);
    \path[latex-latex,line width=1pt] (3.south) edge (3.south|-4.north);
            
    \node[node,text width=4cm]	(20)	at (0,0) {Switch $w_0$};
    
    \path[latex-latex,anchor=east] (0.north) edge node {$b$} (0.north|-20.south);
    \path[latex-latex] (1.north) edge (1.north|-20.south);
    \path[latex-latex] (2.north) edge (2.north|-20.south);
    \path[latex-latex] (3.north) edge (3.north|-20.south);
    
    \path[latex-latex,anchor=east] (10.south) edge node {$b$} (10.south|-20.north);
    \path[latex-latex] (11.south) edge (11.south|-20.north);
    \path[latex-latex] (12.south) edge (12.south|-20.north);
    \path[latex-latex] (13.south) edge (13.south|-20.north);
\end{tikzpicture}
}

%% file: figures/flow_fixed_k.tex
\scalebox{0.7}{
\begin{tikzpicture}[node/.style={rectangle,draw=black,minimum size=7mm,align=center}]
    \node[node]	(10)	at (-1.5,2.5-1) {$c_{1,1}$};
    \node[node]	(11)	at (-0.5,2.5-1) {$c_{1,2}$};
    \node[node]	(12)	at (0.5,2.5-1) {$c_{1,3}$};
    \node[node]	(13)	at (1.5,2.5-1) {$c_{1,4}$};
    \node[node,text width=3.5cm]	(14)	at (0,0.5+2.5) {Switch $w_1$};

    \node[circle,draw=black] (00)    at (2.5, 2.25) {$s$};
    \draw[-latex] (00) to[bend right=28] node[above,pos=0.3]{$k$} (10);
    \draw[-latex] (00) to[bend right=25] (11);
    \draw[-latex] (00) to[bend right=25] (12);
    \draw[-latex] (00) to[bend right=25] (13);
    
    \path[latex-latex,anchor=east,line width=1pt,blue] (10.north) edge node {$\lfloor 10b/y\rfloor$} (10.north|-14.south);
    \path[latex-latex,line width=1pt,blue] (11.north) edge (11.north|-14.south);
    \path[latex-latex,line width=1pt,blue] (12.north) edge (12.north|-14.south);
    \path[latex-latex,line width=1pt,blue] (13.north) edge (13.north|-14.south);
    
    \node[node]	(0)	at (-1.5,1-2.5) {$c_{2,1}$};
    \node[node]	(1)	at (-0.5,1-2.5) {$c_{2,2}$};
    \node[node]	(2)	at (0.5,1-2.5) {$c_{2,3}$};
    \node[node]	(3)	at (1.5,1-2.5) {$c_{2,4}$};
    \node[node,text width=3.5cm]	(4)	at (0,-0.5-2.5) {Switch $w_2$};

    \draw[-latex] (00) to[bend left=35] (0);
    \draw[-latex] (00) to[bend left=35] (1);
    \draw[-latex] (00) to[bend left=35] (2);
    \draw[-latex] (00) to[bend left=25] node[right,pos=0.8]{$k$} (3);
    
    \path[latex-latex,anchor=east,line width=1pt,blue] (0.south) edge node {$\lfloor 10b/y\rfloor$} (0.south|-4.north);
    \path[latex-latex,line width=1pt,blue] (1.south) edge (1.south|-4.north);
    \path[latex-latex,line width=1pt,blue] (2.south) edge (2.south|-4.north);
    \path[latex-latex,line width=1pt,blue] (3.south) edge (3.south|-4.north);
            
    \node[node,text width=4cm]	(20)	at (0,0) {Switch $w_0$};
    
    \path[latex-latex,anchor=east,blue] (0.north) edge node {$\lfloor b/y\rfloor$} (0.north|-20.south);
    \path[latex-latex,blue] (1.north) edge (1.north|-20.south);
    \path[latex-latex,blue] (2.north) edge (2.north|-20.south);
    \path[latex-latex,blue] (3.north) edge (3.north|-20.south);
    
    \path[latex-latex,anchor=east,blue] (10.south) edge node {$\lfloor b/y\rfloor$} (10.south|-20.north);
    \path[latex-latex,blue] (11.south) edge (11.south|-20.north);
    \path[latex-latex,blue] (12.south) edge (12.south|-20.north);
    \path[latex-latex,blue] (13.south) edge (13.south|-20.north);

\end{tikzpicture}
}

%% file: appendix/runtime.tex
\section{Time Complexity Analysis}\label{app:sec:runtime}

In this section, we analyze the runtime complexity of different parts of the algorithm. To summarize, all parts run in polynomial time. One might be concerned by the time complexity like $N^8$. However, the runtime bounds are very loose---for example, using $\mathcal{O}(N^2)$ to bound the number of edges, whereas realistic network topologies are typically much sparser. The bounds also assume single-core execution, while \S\ref{sec:gen_opt} describes various techniques to extensively parallelize the algorithm in practice. \textit{\ul{The purpose of this analysis is to show that the algorithm runs in polynomial time, rather than to provide exact tight runtime bounds.}} We leave proofs of tighter runtime bounds for future work. The empirical runtime performance of ForestColl is evaluated in \S\ref{sec:geneval}.

\textbf{Optimality Binary Search}\quad The key part of optimality binary search is to compute $\min_{v\in V_c} F(s,v;\vec{G}_x)$, which involves computing maxflow from $s$ to every compute node in $V_c$. Assuming the use of preflow-push Algorithm~\cite{preflowpush} to solve network flow, the time complexity to compute $\min_{v\in V_c} F(s,v;\vec{G}_x)$ is $\cO(N|V|^2|E|)$. Note that in practice, one can compute the maxflow from $s$ to each $v\in V_c$ in parallel to significantly speed up the computation. As for how many times $\min_{v\in V_c} F(s,v;\vec{G}_x)$ is computed, observe that the binary search terminates when range is smaller than $1/\min_{v\in V_c}B^-_G(v)^2$. The initial range of binary search is bounded by interval $(0,N)$, so the binary search takes at most $\lceil\log_2(N\min_{v\in V_c}B^-_G(v)^2)\rceil$ iterations. Because $\min_{v\in V_c}B^-_G(v)<|V|\max_{e\in E}b_e$, the total runtime complexity is $\cO(N|V|^2|E|(\log|V|+\log\max_{e\in E}b_e))$.

\textbf{Edge Splitting}\quad In Algorithm~\ref{app:algo:removeswitch}, while we possibly add more edges to the topology, the number of edges is loosely bounded by $\cO(|V|^2)$. Thus, computing $\gamma$ in Theorem~\ref{app:thm:multicapasplit} takes $\cO(N|V|^4)$, and $\gamma$ is computed at most $\cO(|V_s||V|^4)$ times in the nested foreach loop. The total runtime can be loosely bounded by $\cO(N|V_s||V|^8)$.

\textbf{Spanning Tree Construction}\quad Upon completion of Algorithm~\ref{app:algo:removeswitch}, $G^*$ has $N$ vertices and hence $\cO(N^2)$ number of edges. In Algorithm~\ref{app:algo:construct}, $\mu$ only needs one maxflow to be computed. The runtime is thus $\cO(N^4)$. B{\'e}rczi \& Frank~\cite{berczi2010packing} proved that $\mu$ only needs to be computed $\cO(mn^2)$ times, where $m$ and $n$ are the number of edges and vertices respectively. Thus, the runtime of Algorithm~\ref{app:algo:construct} can be loosely bounded by $\cO(N^8)$.

\textbf{Fixed-$k$ Optimality}\quad The runtime of this part is similar to optimality binary search, with the exception that the binary search takes at most $\lceil\log_2(Nk\max_{e\in E}b_e^2)\rceil$ iterations instead. The total runtime complexity is $\cO(N|V|^2|E|(\log N+\log k+\log\max_{e\in E}b_e))$.

%% file: appendix/allreduce_lp.tex
\section{Allreduce Linear Program}\label{app:sec:allreducelp}

Generating an Allreduce schedule is similar to generating an Allgather schedule, as we can also use spanning tree packing. For Allreduce, data flows through spanning in-trees to be reduced at root nodes and then is broadcast through spanning out-trees. One can apply the algorithm introduced in the main text to generate optimal out-trees and then reverse them for the in-trees. While this approach always yields the optimal Allreduce schedule in our work so far, theoretically, Allreduce can be further optimized from two perspectives:
\squishenum
    \item Each node can be the root of a variable number of spanning trees instead of equal number in allgather.
    \item Congestion between spanning in-trees and out-trees may be further optimized.
\squishenumend
In this section, we introduce a linear program designed to optimize allreduce schedules, addressing both perspectives. This linear program formulation automatically determine: for (i), the number of trees rooted at each node, and for (ii), the bandwidth allocation of each edge for the reduce in-trees and broadcast out-trees, respectively.

The linear program works by formulating maxflow and edge splitting as linear program constraints. Given a graph $G$, suppose we want the maxflow from $s$ to $t$ in $G$ being $\geq L$, i.e., $F(s,t;G)\geq L$. This constraint can be expressed in terms of linear program constraints:
\[
	\begin{array}{lr@{}ll}
	\text{s.t.}& \displaystyle\sum_{u\in N^-_G(v)} f^{s,t}_{(u,v)} & \geq\displaystyle\sum_{w\in N^+_G(v)} f^{s,t}_{(v,w)} , & \forall v\!\in\! V(G),v\!\notin\!\{s,t\}\\
	                 & \displaystyle\sum_{u\in N^-_G(t)} f^{s,t}_{(u,t)} & \geq\displaystyle\sum_{w\in N^+_G(t)} f^{s,t}_{(t,w)}+L , & \\
	                 & \displaystyle 0\leq f^{s,t}_{(u,v)} & \displaystyle\leq c_{(u,v)}. & \forall(u,v)\!\in\! E(G)
	\end{array}
\]
As long as a solution exists for the set of decision variables $\{f^{s,t}_{(u,v)}\ |\ (u,v)\in E_G\}$, the maxflow from $s$ to $t$ is $\geq L$. Recall from \S\ref{sec:binarysearch} that our objective is to maximize $x$ while ensuring that the maxflow from $s$ to any $v\!\in\! V_c$ is $\geq Nx$. Here, because of (i), we assign a distinct $x_v$ for each $v\in V_c$. Consequently, the optimization problem shifts to maximize $\sum_{v\in V_c}x_v$ without causing any $F(s,v;\vec{G})\!<\!\sum_{v\in V_c}x_v$. The resulting allreduce communication time is
\[
    T_{\text{comm}}=M\big/\sum_{v\in V_c}x_v.
\]
To optimize (ii), we introduce variables $c^\text{RE}_{(u,v)}$ and $c^\text{BC}_{(u,v)}$ to reserve the bandwidth of each link $(u,v)$ for reduce in-trees and broadcast out-trees, respectively, with $c^\text{RE}_{(u,v)}\!+\!c^\text{BC}_{(u,v)}\!=\!b_{(u,v)}$. Thus, $c^\text{RE}_{(u,v)}$s and $c^\text{BC}_{(u,v)}$s induce two separate graph $G$s, on which we can apply the spanning tree construction to derive in-trees and out-trees, respectively. The linear program formulation is as followed:
\begin{equation}\label{lp:replusbc}
	\begin{array}{ll}
		\text{max}  & \sum_{v\in V_c} x_{v} \\
		\text{s.t.}& \makebox[5.3cm][l]{$F(s,t;\vec{G})\geq \sum_{v\in V_c} x_{v},$}\forall t\in V_c \\
			             & \qquad\text{w.r.t. $0\leq f^{\text{s,t}}_{(s,v)}\leq x_v$ and $0\leq f^{\text{s,t}}_{(u,v)}\leq c^{\text{BC}}_{(u,v)}$}\\
			             & \makebox[5.3cm][l]{$F(t,s;\vec{G})\geq\sum_{v\in V_c} x_{v},$}\forall t\in V_c \\
			             & \qquad\text{w.r.t. $0\leq f^{t,s}_{(v,s)}\leq x_v$ and $0\leq f^{t,s}_{(u,v)}\leq c^{\text{RE}}_{(u,v)}$}\\
						 & \makebox[5.3cm][l]{$c^{\text{RE}}_{(u,v)}+c^{\text{BC}}_{(u,v)}\leq b_{(u,v)},$}\forall (u,v)\!\in\! E_G\\
			             & \makebox[5.3cm][l]{$c^{\text{RE}}_{(u,v)},c^{\text{BC}}_{(u,v)}\geq 0,$}\forall (u,v)\!\in\! E_G\\
			             & \makebox[5.3cm][l]{$x_v\geq 0.$}\forall v\in V_c
	\end{array}
\end{equation}
Note that for reduce in-trees, we constraint the maxflow from $V_c$ to $s$ instead of $s$ to $V_c$, and the $x_v$s are also capacities from $v\!\in\!V_c$ to $s$. The solution to LP~(\ref{lp:replusbc}) yields the optimal allreduce performance.

The linear program is sufficient for switch-free topology. In a switch topology $G$, similar to the algorithm in the main text, we need to convert it into a switch-free topology before applyting the LP. We are unable to solve the LP and then apply edge splitting technique, as the $c^\text{RE}_{(u,v)}$s and $c^\text{BC}_{(u,v)}$s do not guarantee symmetrical bandwidth in the respective induced graphs. To remove switch nodes, we add a level of indirection by defining $b'_{(\alpha,\beta)}$s for all $(\alpha,\beta)\in V_c^2$ to replace the $b_{(u,v)}$s in LP~(\ref{lp:replusbc}). We add multi-commodity flow constraints into the LP to ensure $b'_{(\alpha,\beta)}$ commodity flow from $\alpha$ to $\beta$ in $G$ under the capacities $b_{(u,v)}$s. Thus, the linear program can automatically allocate switch bandwidth for compute-to-compute flows.

In ideal mathematics, we can obtain rational solutions for all variables to derive the in-trees and out-trees to reach optimality. However, in practice, modern LP solvers cannot guarantee rational solutions. We can only round down $x_v,c^\text{RE}_{(u,v)},c^\text{BC}_{(u,v)}$s to the nearest $1/k$ and construct spanning trees, assuming one wants at most $k\cdot\sum_{v\in V_c}x_v$ trees. This approach can approximate the optimal solution as $k$ increases. Nevertheless, the optimal objective value of LP~(\ref{lp:replusbc}) provided by any solver always suggests the optimal allreduce performance, and we can use it to verify the optimality of allreduce schedule derived by using the algorithm in main text. 

%% file: appendix/proofs.tex
\section{Proofs}\label{app:sec:proofs}

\thmbinarysearch*
\begin{proof}
    $\Rightarrow$: Suppose $1/x<\max_{S\subset V,S\not\supseteq V_c}|S\cap V_c|/B^+_G(S)$. Let $S'\subset V,S'\not\supseteq V_c$ be the set that $1/x<|S'\cap V_c|/B^+_G(S')$. Pick arbitrary $v'\in V_c-S'$. Consider the maxflow $F(s,v';\vec{G}_x)$ and $s$-$v'$ cut $(A,\bar{A})$ in $G$ that $A=S'+s$. We have
    \begin{align}\label{app:eq:swppagarbitrarycut}
    \begin{split}
        c(A,\bar{A};\vec{G}_x)
        &=c(S',\bar{A};\vec{G}_x)+\sum_{u\in\bar{A}\cap V_c}c(s,u;\vec{G}_x)\\
        &=B^+_{G}(S')+|V_c-S'|x\\
        &<|S'\cap V_c|x+|V_c-S'|x\\
        &=|V_c|x.
    \end{split}
    \end{align}
	By min-cut theorem, $\min_{v\in V_c}F(s,v;\vec{G}_x)\leq F(s,v';\vec{G}_x)\leq c(A,\bar{A};\vec{G}_x)<|V_c|x$.

    $\Leftarrow$: Suppose $1/x\geq\max_{S\subset V,S\not\supseteq V_c}|S\cap V_c|/B^+_G(S)$. Pick arbitrary $v'\in V_c$. Let $(A,\bar{A})$ be an arbitrary $s$-$v'$ cut and $S'=V\cap A=A-s$. It follows that $1/x\geq|S'\cap V_c|/B^+_G(S')$. Thus, following (\ref{app:eq:swppagarbitrarycut}),
    \begin{align*}
        c(A,\bar{A};\vec{G}_x)
        &=B^+_G(S')+|V_c-S'|x\\
        &\geq |S'\cap V_c|x+|V_c-S'|x\\
        &=|V_c|x.
    \end{align*}
	Because cut $(A,\bar{A})$ is arbitrary, we have $F(s,v';\vec{G}_x)\geq|V_c|x$. Because $v'$ is also arbitrary, $\min_{v\in V_c}F(s,v;\vec{G}_x)\geq|V_c|x$.
\end{proof}

\lmfraction*
\begin{proof}
    Because $a/b\neq c/d$, we have $ad-bc\neq 0$. Thus,
    \[
        \abs{\frac{a}{b}-\frac{c}{d}}=\abs{\frac{ad-bc}{bd}}\geq\frac{1}{bd}\geq\frac{1}{X^2}.
    \]
\end{proof}

\lmsmallestUk*
\begin{proof}
    Since $U/k=1/x^*$, we have $k=Ux^*$, so finding the smallest $k$ is to find the smallest $U$ such that (a) $Ux^*=Uq/p\in\N$ and (b) $Ub_e\in\N$ for all $e\in E$. Suppose $U=\alpha/\beta$ and $\gcd(\alpha,\beta)=1$. Because $\alpha,\beta$ are coprime, $Ub_e\in\N$ implies $\beta|b_e$ for all $e\in E$. Again, because $p,q$ are coprime, $Uq/p\in\N$ implies $p|\alpha$ and $\beta|q$. Thus, the smallest such $\alpha$ is $p$, and the largest such $\beta$ is $\gcd(q, \{b_e\}_{e\in E})$. The proposition immediately follows.
\end{proof}

\thmbangjensentrees*
\thmexisttrees*
\begin{proof}
    $\Rightarrow$: Pick arbitrary $v\in V_c$. Given the family of edge-disjoint out-trees $\{T_{u,i}\}_{u\in V_c,i\in[k]}$, we push one unit of flow from $s$ to $v$ along the path from $u$ to $v$ within tree $T_{u,i}$ for each $u\in V_c,i\in[k]$. Thus, we have constructed a flow assignment with $|V_c|k$ amount of flow. Since $v\in V_c$ is arbitrary, we have $\min_{v\in V_c}F(s,v;\vec{D}_k)\geq|V_c|k$.

    $\Leftarrow$: Suppose $\min_{v\in V_c}F(s,v;\vec{D}_k)\geq|V_c|k$. It immediately implies that $\lambda(s,v;\vec{D}_k)\geq|V_c|k$ for all $v\in V_c$. Note that $T'=V_c$, so by Theorem~\ref{app:thm:bang-jensen}, a family $\cF$ of edge-disjoint out-trees rooted at $s$ exists that each $v\in V_c$ belongs to at least $|V_c|k$ of them. Since $d^+(s)=|V_c|k$ in $\vec{D}_k$, $\cF$ has exactly $|V_c|k$ edge-disjoint out-trees rooted at $s$ and each out-tree spans $V_c$. In addition, for each $v\in V_c$, since $c(s,v;\vec{D}_k)=k$, there are exactly $k$ out-trees in $\cF$ in which $v$ is the only child of root $s$. By removing the root $s$ from every out-tree in $\cF$, we have the family of edge-disjoint out-trees $\{T_{u,i}\}_{u\in V_c,i\in[k]}$ in $D$ as desired.
\end{proof}

\thmbangjensenedgesplit*
\thmedgesplit*
\begin{proof}
    Consider the flow network $\vec{D}_k$. We construct $\vec{D}_k'$ by adding a $k$-capacity edge from each $v\in V_c$ back to $s$. It is trivial to see that $\vec{D}_k'$ is Eulerian. By Theorem~\ref{app:thm:BJedgesplitting}, given $f=(w,t)$, there exists an edge $e=(u,w)$ such that $\lambda(s,v;\vec{D}'^{ef}_k)=\lambda(s,v;\vec{D}_k')$ for all $v\in V_c$. Observe that adding edges from $V_c$ to $s$ does not affect the edge-connectivity from $s$ to any $v\in V_c$, so for all $v\in V_c$,
    \begin{multline*}
        F(s,v;\vec{D}^{ef}_{k})=\lambda(s,v;\vec{D}^{ef}_k)=\lambda(s,v;\vec{D}'^{ef}_k)\\
        =\lambda(s,v;\vec{D}_k')=\lambda(s,v;\vec{D}_k)=F(s,v;\vec{D}_k).
    \end{multline*}
    The theorem trivially follows.
\end{proof}

\thmmulticapasplit*
\begin{proof}
    First of all, one should note that for any $s$-$v$ cut $(A,\bar{A})$ with $v\in V_c$ and $A\subset V_s\cup V_c+s$, if $s,u,t\in A\land v,w\in\bar{A}$, then $(A,\bar{A})$ has the same capacity in $\vec{D}_k$ and $\widehat{D}_{(u,w),v}$, i.e., $c(A,\bar{A};\vec{D}_k)=c(A,\bar{A};\widehat{D}_{(u,w),v})$. Similarly, if $s,w\in A\land v,u,t\in\bar{A}$, then $c(A,\bar{A};\vec{D}_k)=c(A,\bar{A};\widehat{D}_{(w,t),v})$. 
    
    $\geq$: Suppose we split off $(u,w),(w,t)$ by $\gamma$ times and then $F(s,v';\vec{D}^{ef}_k)<|V_c|k$ for some $v'\in V_c$. Let $(A,\bar{A})$ be the min $s$-$v'$ cut in $\vec{D}^{ef}_k$ that $c(A,\bar{A};\vec{D}^{ef}_k)=F(s,v';\vec{D}^{ef}_k)<|V_c|k$. We assert that $(A,\bar{A})$ must cut through $(u,w)$ and $(w,t)$ such that either $s,u,t\in A\land v',w\in\bar{A}$ or $s,w\in A\land v',u,t\in\bar{A}$; otherwise, we have $F(s,v';\vec{D}_k)\leq c(A,\bar{A};\vec{D}_k)=c(A,\bar{A};\vec{D}^{ef}_k)<|V_c|k$ (note that splitting off $(u,w),(w,t)$ adds edge $(u,t)$). Suppose $s,u,t\in A\land v',w\in\bar{A}$, then $c(A,\bar{A};\vec{D}_k)=c(A,\bar{A};\widehat{D}_{(u,w),v'})$. It is trivial to see that $c(A,\bar{A};\vec{D}_k)=c(A,\bar{A};\vec{D}^{ef}_k)+\gamma$. Thus, we have
    \begin{multline*}
        F(u,w;\widehat{D}_{(u,w),v'})\leq c(A,\bar{A};\widehat{D}_{(u,w),v'})\\
        =c(A,\bar{A};\vec{D}_k)=c(A,\bar{A};\vec{D}^{ef}_k)+\gamma<|V_c|k+\gamma,
    \end{multline*}
    contradicting $\gamma\leq\min_{v\in V_c} F(u,w;\widehat{D}_{(u,w),v})-|V_c|k$. For $s,w\in A\land v',u,t\in\bar{A}$, one can similarly show a contradiction by looking at $F(w,t;\widehat{D}_{(w,t),v'})$.
    
    $\leq$: Suppose we split off $(u,w),(w,t)$ by $\gamma'>\gamma$ times and the resulting graph is $D^{ef}$. It is trivial to see that $\gamma'$ cannot be greater than $c(u,w;D)$ or $c(w,t;D)$. Suppose $\gamma'>F(u,w;\widehat{D}_{(u,w),v'})-|V_c|k$ for some $v'\in V_c$. Consider the min $u$-$w$ cut $(A,\bar{A})$ with $c(A,\bar{A};\widehat{D}_{(u,w),v'})=F(u,w;\widehat{D}_{(u,w),v'})$. Because $(u,s),(u,t),(v',w)$ have $\infty$ capacity, we have $s,u,t\in A\land v',w\in\bar{A}$ and hence $c(A,\bar{A};\vec{D}_k)=c(A,\bar{A};\widehat{D}_{(u,w),v'})$. It is again trivial to see that $c(A,\bar{A};\vec{D}^{ef}_k)=c(A,\bar{A};\vec{D}_k)-\gamma'$ and $(A,\bar{A})$ being an $s$-$v'$ cut in $\vec{D}^{ef}_k$. Hence,
    \begin{multline*}
        F(s,v';\vec{D}^{ef}_k)\leq c(A,\bar{A};\vec{D}^{ef}_k)=c(A,\bar{A};\vec{D}_k)-\gamma'\\
        =c(A,\bar{A};\widehat{D}_{(u,w),v'})-\gamma'<|V_c|k.
    \end{multline*}
    One can show similar result for $\gamma'\!>\!F(w,t;\widehat{D}_{(w,t),v'})\!-\!|V_c|k$.
\end{proof}

\thmtarjan*
\thminitialexist*
\begin{proof}
    $\Rightarrow$: Suppose $\min_{v\in V_c}F(s,v;\vec{D}_k)<|V_c|k$. Let $v'$ be the vertex that $F(s,v';\vec{D}_k)<|V_c|k$. By min-cut theorem, there exists an $s$-$v'$ cut $(A,\bar{A})$ in $\vec{D}_k$ such that $c(A,\bar{A};\vec{D}_k)=F(s,v';\vec{D}_k)<|V_c|k$. Let $S=V_c\cap A$, then $A=S+s$, $\bar{S}=V_c-S=V_c+s-A=\bar{A}$, and hence
    \begin{multline*}
        c(S,\bar{S};D)=c(A,\bar{A};\vec{D}_k)-\sum_{u\in\bar{A}}c(s,u;\vec{D}_k)\\
        <|V_c|k-|V_c-S|k=|S|k.
    \end{multline*}
    $\Leftarrow$: Suppose there exists $S\subset V_c,S\neq V_c$ such that $c(S,\bar{S};D)<|S|k$. Pick arbitrary $v'\in V_c-S$. Consider $s$-$v'$ cut $(A,\bar{A})$ such that $A=S+s$. By min-cut theorem, we have
    \begin{multline*}
        F(s,v';\vec{D}_k)\leq c(A,\bar{A};\vec{D}_k)=c(S,\bar{S};D)+\sum_{u\in\bar{A}}c(s,u;\vec{D}_k)\\
        <|S|k+|V_c-S|k=|V_c|k.
    \end{multline*}
\end{proof}

\thmpolynomialpacking*
\thmcomputemu*
\begin{proof}
    For simplicity of notation, let $L=\min\{c(S,\bar{S};D)-p(S;D):x\in S,y\in\bar{S},R_1\not\subseteq S\}$. We will prove (\ref{app:eq:computemu}) by showing that either $L=F(x,y;\overline{D})-\sum_{i\neq 1}m(R_i)$ or $L\geq F(x,y;\overline{D})-\sum_{i\neq 1}m(R_i)\geq m(R_1)$. Let $S\subset V_c$ be arbitrary that $x\in S,y\in\bar{S},R_1\not\subseteq S$, and Let $A=S\cup\{s_i\ |\ R_i\subseteq S\}$. It follows that $(A,\bar{A})$ is an $x$-$y$ cut in $\overline{D}$ and hence
    \begin{equation*}
    \begin{aligned}
        c(S,\bar{S};D)-p(S;D)
        &\textstyle=c(S,\bar{S};D)-\sum\{m(R_i)\ |\ R_i\subseteq S\}\\
        &\textstyle=c(S,\bar{S};D)+\sum\{m(R_i)\ |\ i\neq 1,R_i\not\subseteq S\}\\
        &\textstyle\quad-\sum_{i\neq 1}m(R_i)\\
        &\textstyle=c(A,\bar{A};\overline{D})-\sum_{i\neq 1}m(R_i)\\
        &\textstyle\geq F(x,y;\overline{D})-\sum_{i\neq 1}m(R_i).
    \end{aligned}
    \end{equation*}
    The second equality is due to $R_1\not\subseteq S$, so $\sum\{m(R_i)\ |\ R_i\subseteq S\}=\sum\{m(R_i)\ |\ i\neq 1,R_i\subseteq S\}$. Since $S$ is arbitrary, we have $L\geq F(x,y;\overline{D})-\sum_{i\neq 1}m(R_i)$.
    
    Let $(A',\overline{A'})$ be the min $x$-$y$ cut in $\overline{D}$ and $S'=A'\cap V_c$. We assert that for any $i\neq 1,R_i\subseteq S'$, we have $s_i\in A'$; otherwise, by moving $s_i$ from $\overline{A'}$ to $A'$, we create a cut with lower capacity, contradicting $(A',\overline{A'})$ being min-cut. We also assert that for any $i\neq 1,R_i\not\subseteq S'$, we have $s_i\in\overline{A'}$; otherwise, there exists $v\in R_i-S'$ that $\infty$ edge $(s_i,v)$ crosses $(A',\overline{A'})$. Thus, we have
    \begin{equation}\label{app:eq:computemuproof2}
    \begin{aligned}
        &\textstyle F(x,y;\overline{D})-\sum_{i\neq 1}m(R_i)\\
        =&\textstyle c(A',\overline{A'};\overline{D})-\sum_{i\neq 1}m(R_i)\\
        =&\textstyle c(S',\overline{S'};D)+\sum\{m(R_i)\ |\ i\neq 1,R_i\not\subseteq S'\}-\sum_{i\neq 1}m(R_i)\\
        =&\textstyle c(S',\overline{S'};D)-\sum\{m(R_i)\ |\ i\neq 1,R_i\subseteq S'\}.
    \end{aligned}
    \end{equation}
    Now, we consider two cases:
    \begin{enumerate}[label=(\alph*)]
        \item Suppose $R_1\not\subseteq S'$. Then, $c(S',\overline{S'};D)-p(S';D)\geq L$. By (\ref{app:eq:computemuproof2}), we have
        \begin{align*}
            L
            &\textstyle\geq F(x,y;\overline{D})-\sum_{i\neq 1}m(R_i)\\
            &\textstyle=c(S',\overline{S'};D)-\sum\{m(R_i)\ |\ i\neq 1,R_i\subseteq S'\}\\
            &\textstyle=c(S',\overline{S'};D)-p(S';D)
        \end{align*}
        Thus, $L=c(S',\overline{S'};D)-p(S';D)=F(x,y;\overline{D})-\sum_{i\neq 1}m(R_i)$ and (\ref{app:eq:computemu}) holds.
        \item Suppose $R_1\subseteq S'$. Because the existence of spanning trees is guaranteed, we have
        \[
            \textstyle c(S',\overline{S'};D)\geq p(S';D)=m(R_1)+\sum\{m(R_i)\ |\ i\neq 1,R_i\subseteq S'\}.
        \]
        Hence,
        \begin{align*}
            L
            &\textstyle\geq F(x,y;\overline{D})-\sum_{i\neq 1}m(R_i)\\
            &\textstyle=c(S',\overline{S'};D)-\sum\{m(R_i)\ |\ i\neq 1,R_i\subseteq S'\}\\
            &\textstyle\geq m(R_1).
        \end{align*}
        Thus, $\mu=\min\{g(x,y),m(R_1)\}$ and (\ref{app:eq:computemu}) also holds.
    \end{enumerate}
\end{proof}

\thmfixedk*
\begin{proof}
	$\Leftarrow$: Suppose $\{T_{u,i}\}_{u\in V_c,i\in[k]}$ is edge disjoint in $G(\{\lfloor Ub_e\rfloor\}_{e\in E})$, then
    \begin{multline*}
        T_{\text{comm}}=\frac{M}{Nk}\cdot\max_{e\in E}\frac{1}{b_e}\sum_{T\in \{T_{u,i}\}}\I[e\in T]\\
        \leq\frac{M}{Nk}\cdot\max_{e\in E}\frac{\lfloor Ub_e\rfloor}{b_e}\leq\frac{M}{Nk}\cdot U.
    \end{multline*}
	$\Rightarrow$: Suppose $\{T_{u,i}\}_{u\in V_c,i\in[k]}$ achieves $\frac{M}{Nk}\cdot U$ communication time, then
    \begin{multline*}
        \max_{e\in E}\frac{1}{b_e}\sum_{T\in \{T_{u,i}\}}\I[e\in T]\leq U\\
        \Longrightarrow\quad\sum_{T\in \{T_{u,i}\}}\I[e\in T]\leq Ub_e\quad\text{for all $e\in E$.}
    \end{multline*}
	Since $\sum_{T\in \{T_{u,i}\}}\I[e\in T]$ must be an integer, the edge-disjointness trivially follows.
\end{proof}

\thmfixedkbinarysearch*
\begin{proof}
	$\Rightarrow$: The existence of edge-disjoint $\{T_{u,i}\}_{u\in V_c,i\in[k]}$ in $G(\{\lfloor Ub_e\rfloor\}_{e\in E})$ with $U<U^*$ simply contradicts $\frac{M}{Nk}\cdot U^*$ being the lowest communication time. $\Leftarrow$: Let $\{T^*_{u,i}\}_{u\in V_c,i\in[k]}$ be the family of out-trees with the lowest communication time, then by Theorem~\ref{app:thm:fixedk}, it is edge-disjoint in $G(\{\lfloor Ub_e\rfloor\}_{e\in E})$ for all $U\geq U^*$. 
\end{proof}

\thmfixedkapprox*
\begin{proof}
    Let $U=\max_{e\in E}\lceil kb_e/x^*\rceil/b_e$ where $1/x^*=\max_{S\subset V,S\not\supseteq V_c}|S\cap V_c|/B^+_G(S)$. For each edge $(u,v)$ in $G(\lfloor Ub_e\rfloor)$, we have
    \begin{multline*}
        c(u,v;G(\lfloor Ub_e\rfloor))=\left\lfloor b_{(u,v)}\cdot\max_{e\in E}\frac{\lceil kb_e/x^*\rceil}{b_e}\right\rfloor\\
        \geq\left\lfloor b_{(u,v)}\cdot\frac{\lceil kb_{(u,v)}/x^*\rceil}{b_{(u,v)}}\right\rfloor=\lceil kb_{(u,v)}/x^*\rceil.
    \end{multline*}
    Thus, each edge in $\vec{G}_k(\lfloor Ub_e\rfloor)$ has at least $k/x^*$ times the capacity in $\vec{G}_{x^*}$, so
    \[
        \min_{v\in V_c}F(s,v;\vec{G}_k(\lfloor Ub_e\rfloor))\geq(k/x^*)\min_{v\in V_c}F(s,v;\vec{G}_{x^*})\geq|V_c|k.
    \]
    Therefore, $\frac{M}{Nk}\cdot U$ is achievable and hence $U^*\leq U$ by Theorem~\ref{app:thm:fixedkbinarysearch}.
    \begin{multline*}
        \frac{U^*}{k}\bigg/\frac{1}{x^*}\leq\frac{U}{k}\bigg/\frac{1}{x^*}=\frac{\max_{e\in E}\lceil kb_e/x^*\rceil/b_e}{k/x^*}\\
        \leq\max_{e\in E}\frac{\lceil kb_e/x^*\rceil}{kb_e/x^*}\leq 1+\max_{e\in E}\frac{1}{kb_e/x^*}=1+\frac{x^*}{k\cdot\min_{e\in E}b_e}.
    \end{multline*}
    The theorem trivially follows.
\end{proof}

%% file: main.bbl
\begin{thebibliography}{10}

\bibitem{cdna2}
{AMD CDNA{\textsuperscript{TM}} 2 Architecture}.
\newblock \url{https://www.amd.com/content/dam/amd/en/documents/instinct-business-docs/white-papers/amd-cdna2-white-paper.pdf}.

\bibitem{cdna3}
{AMD CDNA{\textsuperscript{TM}} 3 Architecture}.
\newblock \url{https://www.amd.com/content/dam/amd/en/documents/instinct-tech-docs/white-papers/amd-cdna-3-white-paper.pdf}.

\bibitem{dbrx}
{DBRX 132B}.
\newblock \url{https://www.databricks.com/blog/introducing-dbrx-new-state-art-open-llm}.

\bibitem{nvidiarail}
{Doubling all2all Performance with NVIDIA Collective Communication Library 2.12}.
\newblock \url{https://developer.nvidia.com/blog/doubling-all2all-performance-with-nvidia-collective-communication-library-2-12/}.

\bibitem{msccl}
{Microsoft Collective Communication Library (MSCCL)}.
\newblock \url{https://github.com/Azure/msccl}.

\bibitem{mixtral8x22b}
{Mixtral 8x22B}.
\newblock \url{https://mistral.ai/news/mixtral-8x22b/}.

\bibitem{nccl}
{NVIDIA Collective Communication Library (NCCL)}.
\newblock \url{https://github.com/NVIDIA/nccl}.

\bibitem{dgxv100}
{NVIDIA DGX-1 With Tesla V100 System Architecture}.
\newblock \url{https://images.nvidia.com/content/pdf/dgx1-v100-system-architecture-whitepaper.pdf}.

\bibitem{dgxa100}
{NVIDIA DGX A100 System Architecture}.
\newblock \url{https://resources.nvidia.com/en-us-dgx-systems/dgxa100-system}.

\bibitem{dgxh100}
{NVIDIA H100 Tensor Core GPU Architecture Overview}.
\newblock \url{https://resources.nvidia.com/en-us-tensor-core/gtc22-whitepaper-hopper}.

\bibitem{rccl}
{ROCm Collective Communication Library (RCCL)}.
\newblock \url{https://github.com/ROCm/rccl}.

\bibitem{compcommdma}
{\sc Agrawal, A., Aga, S., Pati, S., and Islam, M.}
\newblock Optimizing ml concurrent computation and communication with gpu dma engines, 2024.

\bibitem{olmo}
{\sc {Ai2}}.
\newblock Olmo: Accelerating the science of language models, 2024.

\bibitem{fattree}
{\sc Al-Fares, M., Loukissas, A., and Vahdat, A.}
\newblock A scalable, commodity data center network architecture.
\newblock In {\em Proceedings of the ACM SIGCOMM 2008 Conference on Data Communication\/} (New York, NY, USA, 2008), SIGCOMM '08, Association for Computing Machinery, p.~63–74.

\bibitem{azurerdma}
{\sc Bai, W., Abdeen, S.~S., Agrawal, A., Attre, K.~K., Bahl, P., Bhagat, A., Bhaskara, G., Brokhman, T., Cao, L., Cheema, A., Chow, R., Cohen, J., Elhaddad, M., Ette, V., Figlin, I., Firestone, D., George, M., German, I., Ghai, L., Green, E., Greenberg, A., Gupta, M., Haagens, R., Hendel, M., Howlader, R., John, N., Johnstone, J., Jolly, T., Kramer, G., Kruse, D., Kumar, A., Lan, E., Lee, I., Levy, A., Lipshteyn, M., Liu, X., Liu, C., Lu, G., Lu, Y., Lu, X., Makhervaks, V., Malashanka, U., Maltz, D.~A., Marinos, I., Mehta, R., Murthi, S., Namdhari, A., Ogus, A., Padhye, J., Pandya, M., Phillips, D., Power, A., Puri, S., Raindel, S., Rhee, J., Russo, A., Sah, M., Sheriff, A., Sparacino, C., Srivastava, A., Sun, W., Swanson, N., Tian, F., Tomczyk, L., Vadlamuri, V., Wolman, A., Xie, Y., Yom, J., Yuan, L., Zhang, Y., and Zill, B.}
\newblock Empowering azure storage with {RDMA}.
\newblock In {\em 20th USENIX Symposium on Networked Systems Design and Implementation (NSDI 23)\/} (Boston, MA, Apr. 2023), USENIX Association, pp.~49--67.

\bibitem{bang-jensen}
{\sc Bang-Jensen, J., Frank, A., and Jackson, B.}
\newblock Preserving and increasing local edge-connectivity in mixed graphs.
\newblock {\em SIAM Journal on Discrete Mathematics 8}, 2 (1995), 155--178.

\bibitem{berczi2010packing}
{\sc B{\'e}rczi, K., and Frank, A.}
\newblock Packing arborescences (combinatorial optimization and discrete algorithms).
\newblock {\em RIMS Kokyuroku Bessatsu B23\/} (2010), 1--31.

\bibitem{npcomplete}
{\sc Bermond, J.-C., and Fraigniaud, P.}
\newblock Broadcasting and {NP}-completeness.
\newblock In {\em Graph Theory Notes of New York\/} (1992), vol.~XXII, pp.~8--14.

\bibitem{sccl}
{\sc Cai, Z., Liu, Z., Maleki, S., Musuvathi, M., Mytkowicz, T., Nelson, J., and Saarikivi, O.}
\newblock Synthesizing optimal collective algorithms.
\newblock In {\em Proceedings of the 26th ACM SIGPLAN Symposium on Principles and Practice of Parallel Programming\/} (New York, NY, USA, 2021), PPoPP '21, Association for Computing Machinery, p.~62–75.

\bibitem{syccl}
{\sc Cao, J., Shi, S., Gao, J., Liu, W., Yang, Y., Xu, Y., Zheng, Z., Guan, Y., Qian, K., Liu, Y., Xu, M., Wang, T., Wang, N., Dong, J., Fu, B., Cai, D., and Zhai, E.}
\newblock Syccl: Exploiting symmetry for efficient collective communication scheduling.
\newblock In {\em Proceedings of the ACM SIGCOMM 2025 Conference\/} (New York, NY, USA, 2025), SIGCOMM '25, Association for Computing Machinery, p.~645–662.

\bibitem{conc_comp}
{\sc Chan, E., Heimlich, M., Purkayastha, A., and van~de Geijn, R.}
\newblock Collective communication: Theory, practice, and experience: Research articles.
\newblock {\em Concurr. Comput. : Pract. Exper. 19}, 13 (sep 2007), 1749--1783.

\bibitem{flux}
{\sc Chang, L.-W., Bao, W., Hou, Q., Jiang, C., Zheng, N., Zhong, Y., Zhang, X., Song, Z., Yao, C., Jiang, Z., Lin, H., Jin, X., and Liu, X.}
\newblock Flux: Fast software-based communication overlap on gpus through kernel fusion, 2024.

\bibitem{centauri}
{\sc Chen, C., Li, X., Zhu, Q., Duan, J., Sun, P., Zhang, X., and Yang, C.}
\newblock Centauri: Enabling efficient scheduling for communication-computation overlap in large model training via communication partitioning.
\newblock In {\em Proceedings of the 29th ACM International Conference on Architectural Support for Programming Languages and Operating Systems, Volume 3\/} (New York, NY, USA, 2024), ASPLOS '24, Association for Computing Machinery, p.~178–191.

\bibitem{rina}
{\sc Chen, Z., Liu, X., Li, M., Hu, Y., Mei, H., Xing, H., Wang, H., Shi, W., Liu, S., and Xu, Y.}
\newblock Rina: Enhancing ring-allreduce with in-network aggregation in distributed model training.
\newblock In {\em 2024 IEEE 32nd International Conference on Network Protocols (ICNP)\/} (2024), pp.~1--12.

\bibitem{blueconnect}
{\sc Cho, M., Finkler, U., and Kung, D.}
\newblock Blueconnect: Novel hierarchical all-reduce on multi-tired network for deep learning.
\newblock In {\em Proceedings of the 2nd SysML Conference\/} (2019).

\bibitem{c-cube}
{\sc Cho, S., Son, H., and Kim, J.}
\newblock Logical/physical topology-aware collective communication in deep learning training.
\newblock In {\em 2023 IEEE International Symposium on High-Performance Computer Architecture (HPCA)\/} (2023), pp.~56--68.

\bibitem{flashattention2}
{\sc Dao, T.}
\newblock Flash{A}ttention-2: Faster attention with better parallelism and work partitioning.
\newblock In {\em International Conference on Learning Representations (ICLR)\/} (2024).

\bibitem{flashattention}
{\sc Dao, T., Fu, D.~Y., Ermon, S., Rudra, A., and R{\'e}, C.}
\newblock Flash{A}ttention: Fast and memory-efficient exact attention with {IO}-awareness.
\newblock In {\em Advances in Neural Information Processing Systems (NeurIPS)\/} (2022).

\bibitem{flare}
{\sc De~Sensi, D., Di~Girolamo, S., Ashkboos, S., Li, S., and Hoefler, T.}
\newblock Flare: flexible in-network allreduce.
\newblock In {\em Proceedings of the International Conference for High Performance Computing, Networking, Storage and Analysis\/} (New York, NY, USA, 2021), SC '21, Association for Computing Machinery.

\bibitem{edmonds}
{\sc Edmonds, J.}
\newblock Edge-disjoint branchings.
\newblock {\em Combinatorial algorithms\/} (1973), 91--96.

\bibitem{frank}
{\sc Frank, A.}
\newblock On connectivity properties of eulerian digraphs.
\newblock In {\em Graph Theory in Memory of G.A. Dirac}, L.~D. Andersen, I.~T. Jakobsen, C.~Thomassen, B.~Toft, and P.~D. Vestergaard, Eds., vol.~41 of {\em Annals of Discrete Mathematics}. Elsevier, 1988, pp.~179--194.

\bibitem{trainium}
{\sc Fu, X., Zhang, Z., Fan, H., Huang, G., El-Shabani, M., Huang, R., Solanki, R., Wu, F., Diamant, R., and Wang, Y.}
\newblock Distributed training of large language models on aws trainium.
\newblock In {\em Proceedings of the 2024 ACM Symposium on Cloud Computing\/} (New York, NY, USA, 2024), SoCC '24, Association for Computing Machinery, p.~961–976.

\bibitem{metardma}
{\sc Gangidi, A., Miao, R., Zheng, S., Bondu, S.~J., Goes, G., Morsy, H., Puri, R., Riftadi, M., Shetty, A.~J., Yang, J., Zhang, S., Fernandez, M.~J., Gandham, S., and Zeng, H.}
\newblock Rdma over ethernet for distributed training at meta scale.
\newblock In {\em Proceedings of the ACM SIGCOMM 2024 Conference\/} (New York, NY, USA, 2024), ACM SIGCOMM '24, Association for Computing Machinery, p.~57–70.

\bibitem{gemma2}
{\sc {Gemma Team, Google DeepMind}}.
\newblock Gemma 2: Improving open language models at a practical size, 2024.

\bibitem{ringallreduce}
{\sc Gibiansky, A.}
\newblock Bringing hpc techniques to deep learning.
\newblock {\em Baidu Research, Tech. Rep.\/} (2017).
\newblock \url{https://andrew.gibiansky.com/blog/machine-learning/baidu-allreduce/}.

\bibitem{preflowpush}
{\sc Goldberg, A.~V., and Tarjan, R.~E.}
\newblock A new approach to the maximum-flow problem.
\newblock {\em J. ACM 35}, 4 (oct 1988), 921–940.

\bibitem{rdmaethernet}
{\sc Guo, C., Wu, H., Deng, Z., Soni, G., Ye, J., Padhye, J., and Lipshteyn, M.}
\newblock Rdma over commodity ethernet at scale.
\newblock In {\em Proceedings of the 2016 ACM SIGCOMM Conference\/} (New York, NY, USA, 2016), SIGCOMM '16, Association for Computing Machinery, p.~202–215.

\bibitem{deadlocksdatacenter}
{\sc Hu, S., Zhu, Y., Cheng, P., Guo, C., Tan, K., Padhye, J., and Chen, K.}
\newblock Deadlocks in datacenter networks: Why do they form, and how to avoid them.
\newblock In {\em Proceedings of the 15th ACM Workshop on Hot Topics in Networks\/} (New York, NY, USA, 2016), HotNets '16, Association for Computing Machinery, p.~92–98.

\bibitem{multitree}
{\sc Huang, J., Majumder, P., Kim, S., Muzahid, A., Yum, K.~H., and Kim, E.~J.}
\newblock Communication algorithm-architecture co-design for distributed deep learning.
\newblock In {\em Proceedings of the 48th Annual International Symposium on Computer Architecture\/} (2021), ISCA '21, IEEE Press, p.~181–194.

\bibitem{jackson}
{\sc Jackson, B.}
\newblock Some remarks on arc-connectivity, vertex splitting, and orientation in graphs and digraphs.
\newblock {\em Journal of Graph Theory 12}, 3 (1988), 429--436.

\bibitem{coconet}
{\sc Jangda, A., Huang, J., Liu, G., Sabet, A. H.~N., Maleki, S., Miao, Y., Musuvathi, M., Mytkowicz, T., and Saarikivi, O.}
\newblock Breaking the computation and communication abstraction barrier in distributed machine learning workloads.
\newblock In {\em Proceedings of the 27th ACM International Conference on Architectural Support for Programming Languages and Operating Systems\/} (New York, NY, USA, 2022), ASPLOS '22, Association for Computing Machinery, p.~402–416.

\bibitem{flexflow}
{\sc Jia, Z., Zaharia, M., and Aiken, A.}
\newblock Beyond data and model parallelism for deep neural networks.
\newblock In {\em Proceedings of Machine Learning and Systems\/} (2019), A.~Talwalkar, V.~Smith, and M.~Zaharia, Eds., vol.~1, pp.~1--13.

\bibitem{megascale}
{\sc Jiang, Z., Lin, H., Zhong, Y., Huang, Q., Chen, Y., Zhang, Z., Peng, Y., Li, X., Xie, C., Nong, S., Jia, Y., He, S., Chen, H., Bai, Z., Hou, Q., Yan, S., Zhou, D., Sheng, Y., Jiang, Z., Xu, H., Wei, H., Zhang, Z., Nie, P., Zou, L., Zhao, S., Xiang, L., Liu, Z., Li, Z., Jia, X., Ye, J., Jin, X., and Liu, X.}
\newblock {MegaScale}: Scaling large language model training to more than 10,000 {GPUs}.
\newblock In {\em 21st USENIX Symposium on Networked Systems Design and Implementation (NSDI 24)\/} (Santa Clara, CA, Apr. 2024), USENIX Association, pp.~745--760.

\bibitem{inaatp}
{\sc Lao, C., Le, Y., Mahajan, K., Chen, Y., Wu, W., Akella, A., and Swift, M.}
\newblock {ATP}: In-network aggregation for multi-tenant learning.
\newblock In {\em 18th USENIX Symposium on Networked Systems Design and Implementation (NSDI 21)\/} (Apr. 2021), USENIX Association, pp.~741--761.

\bibitem{TTO}
{\sc Laskar, S., Majhi, P., Kim, S., Mahmud, F., Muzahid, A., and Kim, E.~J.}
\newblock Enhancing collective communication in mcm accelerators for deep learning training.
\newblock In {\em 2024 IEEE International Symposium on High-Performance Computer Architecture (HPCA)\/} (2024), pp.~1--16.

\bibitem{pytorchddp}
{\sc Li, S., Zhao, Y., Varma, R., Salpekar, O., Noordhuis, P., Li, T., Paszke, A., Smith, J., Vaughan, B., Damania, P., and Chintala, S.}
\newblock Pytorch distributed: experiences on accelerating data parallel training.
\newblock {\em Proc. VLDB Endow. 13}, 12 (Aug. 2020), 3005–3018.

\bibitem{kaichen}
{\sc Li, W., Liu, X., Li, Y., Jin, Y., Tian, H., Zhong, Z., Liu, G., Zhang, Y., and Chen, K.}
\newblock Understanding communication characteristics of distributed training.
\newblock In {\em Proceedings of the 8th Asia-Pacific Workshop on Networking\/} (New York, NY, USA, 2024), APNet '24, Association for Computing Machinery, p.~1–8.

\bibitem{nnscaler}
{\sc Lin, Z., Miao, Y., Zhang, Q., Yang, F., Zhu, Y., Li, C., Maleki, S., Cao, X., Shang, N., Yang, Y., Xu, W., Yang, M., Zhang, L., and Zhou, L.}
\newblock {nnScaler}: {Constraint-Guided} parallelization plan generation for deep learning training.
\newblock In {\em 18th USENIX Symposium on Operating Systems Design and Implementation (OSDI 24)\/} (Santa Clara, CA, July 2024), USENIX Association, pp.~347--363.

\bibitem{inadt}
{\sc Liu, S., Wang, Q., Zhang, J., Wu, W., Lin, Q., Liu, Y., Xu, M., Canini, M., Cheung, R. C.~C., and He, J.}
\newblock In-network aggregation with transport transparency for distributed training.
\newblock In {\em Proceedings of the 28th ACM International Conference on Architectural Support for Programming Languages and Operating Systems, Volume 3\/} (New York, NY, USA, 2023), ASPLOS 2023, Association for Computing Machinery, p.~376–391.

\bibitem{teccl}
{\sc Liu, X., Arzani, B., Kakarla, S. K.~R., Zhao, L., Liu, V., Castro, M., Kandula, S., and Marshall, L.}
\newblock Rethinking machine learning collective communication as a multi-commodity flow problem.
\newblock In {\em Proceedings of the ACM SIGCOMM 2024 Conference\/} (New York, NY, USA, 2024), ACM SIGCOMM '24, Association for Computing Machinery, p.~16–37.

\bibitem{llama3}
{\sc {Llama Team, AI @ Meta}}.
\newblock The llama 3 herd of models, 2024.

\bibitem{syndicate}
{\sc Mahajan, K., Chu, C.-H., Sridharan, S., and Akella, A.}
\newblock Better together: Jointly optimizing {ML} collective scheduling and execution planning using {SYNDICATE}.
\newblock In {\em 20th USENIX Symposium on Networked Systems Design and Implementation (NSDI 23)\/} (Boston, MA, Apr. 2023), USENIX Association, pp.~809--824.

\bibitem{touvron2023llama}
{\sc Meta}.
\newblock Llama 2: Open foundation and fine-tuned chat models, 2023.

\bibitem{jgrapht}
{\sc Michail, D., Kinable, J., Naveh, B., and Sichi, J.~V.}
\newblock Jgrapht—a java library for graph data structures and algorithms.
\newblock {\em ACM Trans. Math. Softw. 46}, 2 (May 2020).

\bibitem{megatrondeepak}
{\sc Narayanan, D., Shoeybi, M., Casper, J., LeGresley, P., Patwary, M., Korthikanti, V., Vainbrand, D., Kashinkunti, P., Bernauer, J., Catanzaro, B., Phanishayee, A., and Zaharia, M.}
\newblock Efficient large-scale language model training on gpu clusters using megatron-lm.
\newblock In {\em Proceedings of the International Conference for High Performance Computing, Networking, Storage and Analysis\/} (New York, NY, USA, 2021), SC '21, Association for Computing Machinery.

\bibitem{nemotron}
{\sc Nvidia}.
\newblock Nemotron-4 340b technical report, 2024.

\bibitem{T3}
{\sc Pati, S., Aga, S., Islam, M., Jayasena, N., and Sinclair, M.~D.}
\newblock T3: Transparent tracking \& triggering for fine-grained overlap of compute \& collectives.
\newblock In {\em Proceedings of the 29th ACM International Conference on Architectural Support for Programming Languages and Operating Systems, Volume 2\/} (New York, NY, USA, 2024), ASPLOS '24, Association for Computing Machinery, p.~1146–1164.

\bibitem{tictac}
{\sc Peng, Y., Zhu, Y., Chen, Y., Bao, Y., Yi, B., Lan, C., Wu, C., and Guo, C.}
\newblock A generic communication scheduler for distributed dnn training acceleration.
\newblock In {\em Proceedings of the 27th ACM Symposium on Operating Systems Principles\/} (New York, NY, USA, 2019), SOSP '19, Association for Computing Machinery, p.~16–29.

\bibitem{pope2023efficiently}
{\sc Pope, R., Douglas, S., Chowdhery, A., Devlin, J., Bradbury, J., Heek, J., Xiao, K., Agrawal, S., and Dean, J.}
\newblock Efficiently scaling transformer inference.
\newblock {\em Proceedings of Machine Learning and Systems 5\/} (2023).

\bibitem{alibabahpn}
{\sc Qian, K., Xi, Y., Cao, J., Gao, J., Xu, Y., Guan, Y., Fu, B., Shi, X., Zhu, F., Miao, R., Wang, C., Wang, P., Zhang, P., Zeng, X., Ruan, E., Yao, Z., Zhai, E., and Cai, D.}
\newblock Alibaba hpn: A data center network for large language model training.
\newblock In {\em Proceedings of the ACM SIGCOMM 2024 Conference\/} (New York, NY, USA, 2024), ACM SIGCOMM '24, Association for Computing Machinery, p.~691–706.

\bibitem{halving-doubling}
{\sc Rabenseifner, R.}
\newblock Optimization of collective reduction operations.
\newblock In {\em Computational Science - ICCS 2004\/} (Berlin, Heidelberg, 2004), M.~Bubak, G.~D. van Albada, P.~M.~A. Sloot, and J.~Dongarra, Eds., Springer Berlin Heidelberg, pp.~1--9.

\bibitem{cassini}
{\sc Rajasekaran, S., Ghobadi, M., and Akella, A.}
\newblock {CASSINI}: {Network-Aware} job scheduling in machine learning clusters.
\newblock In {\em 21st USENIX Symposium on Networked Systems Design and Implementation (NSDI 24)\/} (Santa Clara, CA, Apr. 2024), USENIX Association, pp.~1403--1420.

\bibitem{zero}
{\sc Rajbhandari, S., Rasley, J., Ruwase, O., and He, Y.}
\newblock Zero: Memory optimizations toward training trillion parameter models.
\newblock In {\em Proceedings of the International Conference for High Performance Computing, Networking, Storage and Analysis\/} (2020), SC '20, IEEE Press.

\bibitem{switchml}
{\sc Sapio, A., Canini, M., Ho, C.-Y., Nelson, J., Kalnis, P., Kim, C., Krishnamurthy, A., Moshref, M., Ports, D., and Richtarik, P.}
\newblock Scaling distributed machine learning with {In-Network} aggregation.
\newblock In {\em 18th USENIX Symposium on Networked Systems Design and Implementation (NSDI 21)\/} (Apr. 2021), USENIX Association, pp.~785--808.

\bibitem{schrijver}
{\sc Schrijver, A.}
\newblock Combinatorial optimization : polyhedra and efficiency, 2003.

\bibitem{horovod}
{\sc Sergeev, A., and Balso, M.~D.}
\newblock Horovod: fast and easy distributed deep learning in tensorflow, 2018.

\bibitem{taccl}
{\sc Shah, A., Chidambaram, V., Cowan, M., Maleki, S., Musuvathi, M., Mytkowicz, T., Nelson, J., Saarikivi, O., and Singh, R.}
\newblock {TACCL}: Guiding collective algorithm synthesis using communication sketches.
\newblock In {\em 20th USENIX Symposium on Networked Systems Design and Implementation (NSDI 23)\/} (Boston, MA, Apr. 2023), USENIX Association, pp.~593--612.

\bibitem{mscclpp}
{\sc Shah, A., Jangda, A., Li, B., Rocha, C., Hwang, C., Jose, J., Musuvathi, M., Saarikivi, O., Cheng, P., Zhou, Q., Dathathri, R., Maleki, S., and Yang, Z.}
\newblock Msccl++: Rethinking gpu communication abstractions for cutting-edge ai applications, 2025.

\bibitem{megatronlm}
{\sc Shoeybi, M., Patwary, M., Puri, R., LeGresley, P., Casper, J., and Catanzaro, B.}
\newblock Megatron-lm: Training multi-billion parameter language models using model parallelism, 2020.

\bibitem{tarjan}
{\sc Tarjan, R.~E.}
\newblock A good algorithm for edge-disjoint branching.
\newblock {\em Information Processing Letters 3}, 2 (1974), 51--53.

\bibitem{blink}
{\sc Wang, G., Venkataraman, S., Phanishayee, A., Devanur, N., Thelin, J., and Stoica, I.}
\newblock Blink: Fast and generic collectives for distributed ml.
\newblock In {\em Proceedings of Machine Learning and Systems\/} (2020), I.~Dhillon, D.~Papailiopoulos, and V.~Sze, Eds., vol.~2, pp.~172--186.

\bibitem{domino}
{\sc Wang, G., Zhang, C., Shen, Z., Li, A., and Ruwase, O.}
\newblock Domino: Eliminating communication in llm training via generic tensor slicing and overlapping, 2024.

\bibitem{mlt}
{\sc Wang, H., Tian, H., Chen, J., Wan, X., Xia, J., Zeng, G., Bai, W., Jiang, J., Wang, Y., and Chen, K.}
\newblock Towards {Domain-Specific} network transport for distributed {DNN} training.
\newblock In {\em 21st USENIX Symposium on Networked Systems Design and Implementation (NSDI 24)\/} (Santa Clara, CA, Apr. 2024), USENIX Association, pp.~1421--1443.

\bibitem{roar}
{\sc Wang, R., Dong, D., Lei, F., Ma, J., Wu, K., and Lu, K.}
\newblock Roar: A router microarchitecture for in-network allreduce.
\newblock In {\em Proceedings of the 37th ACM International Conference on Supercomputing\/} (New York, NY, USA, 2023), ICS '23, Association for Computing Machinery, p.~423–436.

\bibitem{overlapdecomp}
{\sc Wang, S., Wei, J., Sabne, A., Davis, A., Ilbeyi, B., Hechtman, B., Chen, D., Murthy, K.~S., Maggioni, M., Zhang, Q., Kumar, S., Guo, T., Xu, Y., and Zhou, Z.}
\newblock Overlap communication with dependent computation via decomposition in large deep learning models.
\newblock In {\em Proceedings of the 28th ACM International Conference on Architectural Support for Programming Languages and Operating Systems, Volume 1\/} (New York, NY, USA, 2022), ASPLOS 2023, Association for Computing Machinery, p.~93–106.

\bibitem{wang2023build}
{\sc Wang, W., Ghobadi, M., Shakeri, K., Zhang, Y., and Hasani, N.}
\newblock How to build low-cost networks for large language models (without sacrificing performance)?, 2023.

\bibitem{wang2024rail}
{\sc Wang, W., Ghobadi, M., Shakeri, K., Zhang, Y., and Hasani, N.}
\newblock Rail-only: A low-cost high-performance network for training llms with trillion parameters.
\newblock In {\em 2024 IEEE Symposium on High-Performance Interconnects (HOTI)\/} (Los Alamitos, CA, USA, aug 2024), IEEE Computer Society, pp.~1--10.

\bibitem{topoopt}
{\sc Wang, W., Khazraee, M., Zhong, Z., Ghobadi, M., Jia, Z., Mudigere, D., Zhang, Y., and Kewitsch, A.}
\newblock {TopoOpt}: Co-optimizing network topology and parallelization strategy for distributed training jobs.
\newblock In {\em 20th USENIX Symposium on Networked Systems Design and Implementation (NSDI 23)\/} (Boston, MA, Apr. 2023), USENIX Association, pp.~739--767.

\bibitem{hftransformers}
{\sc Wolf, T., Debut, L., Sanh, V., Chaumond, J., Delangue, C., Moi, A., Cistac, P., Rault, T., Louf, R., Funtowicz, M., Davison, J., Shleifer, S., von Platen, P., Ma, C., Jernite, Y., Plu, J., Xu, C., Scao, T.~L., Gugger, S., Drame, M., Lhoest, Q., and Rush, A.~M.}
\newblock Transformers: State-of-the-art natural language processing.
\newblock In {\em Proceedings of the 2020 Conference on Empirical Methods in Natural Language Processing: System Demonstrations\/} (Online, Oct. 2020), Association for Computational Linguistics, pp.~38--45.

\bibitem{tacos}
{\sc Won, W., Elavazhagan, M., Srinivasan, S., Gupta, S., and Krishna, T.}
\newblock { TACOS: Topology-Aware Collective Algorithm Synthesizer for Distributed Machine Learning }.
\newblock In {\em 2024 57th IEEE/ACM International Symposium on Microarchitecture (MICRO)\/} (Los Alamitos, CA, USA, Nov. 2024), IEEE Computer Society, pp.~856--870.

\bibitem{mccs}
{\sc Wu, Y., Xu, Y., Chen, J., Wang, Z., Zhang, Y., Lentz, M., and Zhuo, D.}
\newblock Mccs: A service-based approach to collective communication for multi-tenant cloud.
\newblock In {\em Proceedings of the ACM SIGCOMM 2024 Conference\/} (New York, NY, USA, 2024), ACM SIGCOMM '24, Association for Computing Machinery, p.~679–690.

\bibitem{bfb}
{\sc Zhao, L., Pal, S., Chugh, T., Wang, W., Fantl, J., Basu, P., Khoury, J., and Krishnamurthy, A.}
\newblock Efficient {Direct-Connect} topologies for collective communications.
\newblock In {\em 22nd USENIX Symposium on Networked Systems Design and Implementation (NSDI 25)\/} (Philadelphia, PA, Apr. 2025), USENIX Association, pp.~705--737.

\bibitem{pytorchfsdp}
{\sc Zhao, Y., Gu, A., Varma, R., Luo, L., Huang, C.-C., Xu, M., Wright, L., Shojanazeri, H., Ott, M., Shleifer, S., Desmaison, A., Balioglu, C., Damania, P., Nguyen, B., Chauhan, G., Hao, Y., Mathews, A., and Li, S.}
\newblock Pytorch fsdp: Experiences on scaling fully sharded data parallel.
\newblock {\em Proc. VLDB Endow. 16}, 12 (aug 2023), 3848–3860.

\bibitem{xinjin}
{\sc Zhao, Y., Liu, X., and Jin, X.}
\newblock How useful is communication scheduling for distributed training?
\newblock In {\em 2024 International Scientific and Technical Conference Modern Computer Network Technologies (MoNeTeC)\/} (2024), pp.~1--13.

\bibitem{alpa}
{\sc Zheng, L., Li, Z., Zhang, H., Zhuang, Y., Chen, Z., Huang, Y., Wang, Y., Xu, Y., Zhuo, D., Xing, E.~P., Gonzalez, J.~E., and Stoica, I.}
\newblock Alpa: Automating inter- and {Intra-Operator} parallelism for distributed deep learning.
\newblock In {\em 16th USENIX Symposium on Operating Systems Design and Implementation (OSDI 22)\/} (Carlsbad, CA, July 2022), USENIX Association, pp.~559--578.

\end{thebibliography}
